%% file: main.tex
\documentclass[opre,nonblindrev]{informs4rad}

\OneAndAHalfSpacedXI

\input{Paper/macros}

\input{Paper/math}

\newcommand{\revcolor}[1]{{\color{black}#1}}
\newcommand{\vrevcolor}[1]{{\color{black}#1}}

\makeatletter
\DeclareRobustCommand{\qed}{%
  \ifmmode 
  \else \leavevmode\unskip\penalty9999 \hbox{}\nobreak\hfill
  \fi
  \quad\hbox{\qedsymbol}}
\newcommand{\qedsymbol}{\openbox}
\renewenvironment{proof}[1][\proofname]{\par
  \normalfont
  \topsep6\p@\@plus6\p@ \trivlist
  \item[\hskip\labelsep\itshape
    #1.]\ignorespaces
}{%
  \qed\endtrivlist
}
\newcommand{\proofname}{Proof}
\makeatother

\EquationsNumberedThrough    

\TheoremsNumberedThrough     
\ECRepeatTheorems  %

\MANUSCRIPTNO{}

\begin{document}

\RUNAUTHOR{Feng, Manshadi, Niazadeh, Neyshabouri}

\RUNTITLE{Robust Dynamic Staffing with Predictions}

\TITLE{Robust Dynamic Staffing with Predictions}

\ARTICLEAUTHORS{
\AUTHOR{Yiding Feng}
\AFF{Hong Kong University of Science and Technology, \EMAIL{ydfeng@ust.hk}}

\AUTHOR{Vahideh Manshadi}
\AFF{Yale School of Management, \EMAIL{vahideh.manshadi@yale.edu }}

\AUTHOR{Rad Niazadeh}
\AFF{The University of Chicago, Booth School of Business, \EMAIL{rad.niazadeh@chicagobooth.edu}}

\AUTHOR{Saba Neyshabouri}
\AFF{Amazon, \EMAIL{s.neyshabouri@gmail.com}}
} 

\ABSTRACT{%
\input{Paper/abstract}
}%



\maketitle
\setcounter{page}{1}
\newpage

\section{Introduction}
\input{Paper/intro}

\section{Preliminaries}
\input{Paper/prelim}

\section{Minimax Optimal Algorithm in the Base Model}
\input{Paper/base-model}
\section{Minimax Optimal Algorithms in Extension Models}
\input{Paper/extension}

\vspace{-2mm}
\section{Conclusion \& Future Directions}
\input{Paper/conclude}

\setlength{\bibsep}{0.0pt}
\bibliographystyle{plainnat}
{\footnotesize
\bibliography{refs}}

\renewcommand{\theHchapter}{A\arabic{chapter}}
\renewcommand{\theHsection}{A\arabic{section}}

\newpage
\ECSwitch
\ECDisclaimer

\section{Further Related Work}
\input{Paper/related-work}

\revcolor{
\section{Numerical Experiments}
\input{Paper/numerical}

}

\revcolor{
\section{Extra Cost under Probabilistic Miscoverage Shocks}
\input{Paper/apx-miscoverage}
}

\section{Missing Technical Details of Jointly Minimizing Cost of Hiring and Staffing}
\input{Paper/apx-combined-objective}

\section{Missing Technical Details of Workforce Planning with Costly Hiring and Releasing}

\input{Paper/apx-cancellation}

\section{Other Robust Criteria}
\input{Paper/apx-regret}

\revcolor{
\section{Refined Characterization of \texorpdfstring{$\optcost$}{GammaStar} in Section~\ref{sec:simple instance}}
\input{Paper/apx-refined-gamma}}

\section{Missing Proofs}
\input{Paper/apx-proofs}

\end{document}

%% file: Paper/macros.tex
\usepackage{tikz}
\usepackage{pgfplots}
\pgfplotsset{compat=newest}
\usetikzlibrary{patterns}
\usepgfplotslibrary{fillbetween}
\usetikzlibrary{intersections}

\usepackage{color-edits}
\addauthor{yf}{purple}    
\addauthor{vm}{blue}
\addauthor{rn}{red}

\usepackage[caption=false]{subfig}

\usepackage{bbm}

\usepackage[breakable, skins]{tcolorbox}
\definecolor{myLightGray}{gray}{0.9} 

\DeclareRobustCommand{\myboxtwo}[2][gray!15]{
\begin{tcolorbox}[ 
        colback=white,      
        colframe=gray,  
        boxrule=0.2pt,      
        arc=2pt,outer arc=2pt,
        left=12pt,
        right=12pt,
        top=5pt,
        bottom=5pt,
        width=1.07\linewidth,
        enlarge left by=-0.55cm,
        before upper=\renewcommand{\baselinestretch}{1.3}\selectfont,
        after upper=\normalfont
        ]
 #2
 \end{tcolorbox}
}

\DeclareRobustCommand{\mybox}[2][myLightGray]{%
\begin{tcolorbox}[
        left=0.5pt,
        right=0.5pt,
        top=0.5pt,
        bottom=0.5pt,
        colback=#1,
        colframe=#1,
        width=\dimexpr\textwidth\relax, 
        boxsep=3pt,
        arc=2pt,outer arc=2pt
        ]
        #2    
\end{tcolorbox}
}
\usepackage{csquotes}

\renewenvironment*{displayquote}
  {\begingroup
   \setlength{\leftmargini}{0.2cm}%
   \linespread{1.4}\selectfont 
   \csq@getcargs{\csq@bdquote{}{}}}
  {\csq@edquote\endgroup}
\makeatother

\usepackage{mathtools, bm, amsmath}
\usepackage{xcolor}
\usepackage[normalem]{ulem}
\usepackage{thm-restate}
\usepackage[T1]{fontenc}    

\usepackage{comment}
\usepackage{nicefrac}       
\usepackage{float}
\usepackage{xspace}
\usepackage{xfrac}
\usepackage{enumitem}
\usepackage{graphicx}

\usepackage{txfonts}

\usepackage{natbib}
 \bibpunct[, ]{(}{)}{,}{a}{}{,}%
 \def\bibfont{\small}%
 \def\bibsep{\smallskipamount}%

\definecolor{cornellred}{rgb}{0.7, 0.11, 0.11}
\definecolor{maroon}{rgb}{0.52, 0, 0}
\definecolor{dgreen}{rgb}{0.0, 0.5, 0.0}
\definecolor{ballblue}{rgb}{0.13, 0.67, 0.8}
\definecolor{royalblue(web)}{rgb}{0.25, 0.41, 0.88}
\definecolor{bleudefrance}{rgb}{0.19, 0.55, 0.91}
\definecolor{royalazure}{rgb}{0.0, 0.22, 0.66}
\usepackage[hypertexnames=false]{hyperref}
\hypersetup{
	colorlinks = true,
	linkcolor=cornellred,
	citecolor=royalblue(web),
	urlcolor= royalazure,
    pdfborder={0 0 1},       
    pdfborderstyle={/S/U/W 1} 
}

\usepackage{cleveref}

\usepackage[algo2e,ruled,vlined,linesnumbered]{algorithm2e}
\SetNoFillComment

\usepackage{setspace}
\usepackage{fix-cm}
\crefname{algocf}{alg.}{algs.}
\Crefname{algocf}{Algorithm}{Algorithms}
\usepackage{algpseudocode}
\usepackage{algorithm}

\usepackage{cleveref}
\newenvironment{myprocedure}[1][htb]
{  
\begin{algorithm2e}[#1]%
}{\end{algorithm2e}}

\newcommand{\squishlist}{
\begin{list}{{{\small{$\bullet$}}}}
{\setlength{\itemsep}{3pt}      \setlength{\parsep}{1pt}
\setlength{\topsep}{1pt}       \setlength{\partopsep}{0pt}
\setlength{\leftmargin}{1em} \setlength{\labelwidth}{1em}
\setlength{\labelsep}{0.5em} } }
\newcommand{\squishend}{  \end{list}}

\newcommand{\primed}{^{\dagger}}
\newcommand{\doubleprimed}{^{\ddagger}}

\newcommand{\xhdr}[1]{\vspace{2mm} \noindent{\bf #1}}

\newcommand{\Cost}[2][]{{\textsc{cost}}\ifthenelse{\not\equal{}{#1}}{_{#1}}{}\!\left[{\def\givenn{\middle|}#2}\right]}
\newcommand{\CostInf}[2][]{{\textsc{E-cost}}\ifthenelse{\not\equal{}{#1}}{_{#1}}{}\!\left[{\def\givenn{\middle|}#2}\right]}
\newcommand{\CostOne}[2][]{{\textsc{U-cost}}\ifthenelse{\not\equal{}{#1}}{_{#1}}{}\!\left[{\def\givenn{\middle|}#2}\right]}
\newcommand{\CostTotal}[2][]{{\widetilde{\textsc{cost}}}\ifthenelse{\not\equal{}{#1}}{_{#1}}{}\!\left[{\def\givenn{\middle|}#2}\right]}

\newcommand{\demand}{d}
\newcommand{\demandj}{\demand_j}

\newcommand{\supply}{s}
\newcommand{\supplyi}{\supply_i}

\newcommand{\wdiscount}{\rho}

\newcommand{\wdiscountit}{\wdiscount_{it}}

\newcommand{\alloc}{x}
\newcommand{\allocijt}{\alloc_{ijt}}
\newcommand{\price}{p}

\newcommand{\dbf}{\boldsymbol{d}}
\newcommand{\xbf}{\boldsymbol{x}}

\newcommand{\undercost}{c}
\newcommand{\overcost}{C}

\newcommand{\reals}{\mathbb{R}}

\newcommand{\prediction}{\mathcal{P}}
\newcommand{\predictions}{\boldsymbol{\prediction}}
\newcommand{\predictiont}{\prediction_t}
\newcommand{\pL}{L}
\newcommand{\pLjt}{\pL_{jt}}
\newcommand{\pLjzero}{\pL_{j0}}
\newcommand{\pR}{R}
\newcommand{\pRjt}{\pR_{jt}}
\newcommand{\pRjzero}{\pR_{j0}}
\newcommand{\perror}{\Delta}
\newcommand{\perrorjt}{\perror_{jt}}
\newcommand{\perrorjk}{\perror_{jk}}

\newcommand{\instance}{\mathcal{I}}

\newcommand{\ALG}{\texttt{ALG}}
\newcommand{\OPTALG}{\ALG^*}
\newcommand{\optcost}{\Gamma^*}

\newcommand{\budget}{B}

\newcommand{\undercostj}{\undercost_j}
\newcommand{\overcostj}{\overcost_j}

\newcommand{\allocit}{\alloc_{it}}
\newcommand{\priceit}{\price_{it}}
\newcommand{\pRzero}{\pR_{0}}
\newcommand{\perrork}{\perror_k}
\newcommand{\perrort}{\perror_t}

\newcommand{\pLzero}{\pL_0}
\newcommand{\pLt}{\pL_t}
\newcommand{\pRt}{\pR_t}

\newcommand{\pRtHat}{\hat{\pR}_t}

\newcommand{\alloct}{\alloc_t}
\newcommand{\revoke}{y}

\newcommand{\cprice}{q}

\newcommand{\wdiscountt}{\wdiscount_t}
\newcommand{\cptotal}{L}
\newcommand{\naturals}{\mathbb{N}}

\newcommand{\switchseq}{J}
\newcommand{\switchseqspace}{\mathcal{J}}
\newcommand{\switchseqell}{\switchseq_{\ell}}
\newcommand{\switchseqonetoell}{\switchseq_{1:\ell}}
\newcommand{\rgapJ}{\rgap(\switchseq)}

\newcommand{\revokeit}{\revoke_{it}}

\newcommand{\lgapJ}{\lgap(\switchseq)}

\newcommand{\cpriceit}{\cprice_{it}}
\newcommand{\cpriceiled}{\cprice_{\ell}}
\newcommand{\cinterval}{\mathcal{T}}
\newcommand{\cintervalell}{\cinterval_{\ell}}
\newcommand{\cintervalellplus}{\cinterval_{\ell}^+}

\newcommand{\zerobf}{\boldsymbol{0}}
\newcommand{\ybf}{\boldsymbol{y}}

\newcommand{\targetcost}{\Gamma}
\newcommand{\targetcosts}{\boldsymbol{\targetcost}}

\newcommand{\lambdabf}{\boldsymbol{\lambda}}
\newcommand{\thetabf}{\boldsymbol{\theta}}
\newcommand{\lgap}{\lambda}
\newcommand{\lgapjk}{\lgap_{jk}}
\newcommand{\lgapj}{\lgap_{j}}
\newcommand{\rgap}{\theta}
\newcommand{\rgapj}{\rgap_j}
\newcommand{\ked}{^{(k)}}

\newcommand{\randomalloc}{X}
\newcommand{\randomallocijt}{\randomalloc_{ijt}}
\newcommand{\randomallocit}{\randomalloc_{it}}
\newcommand{\Ted}{^{(T)}}
\newcommand{\ted}{^{(t)}}

\newcommand{\canalloc}{\tilde{\alloc}}
\newcommand{\canallocit}{\canalloc_{it}}
\newcommand{\canxbf}{\tilde\xbf}

\newcommand{\csuitted}{^{\clubsuit}}
\newcommand{\spsuitted}{^{\spadesuit}}

\newcommand{\revokei}{\revoke_i}
\newcommand{\status}{S}
\newcommand{\cumalloc}{z}
\newcommand{\cumalloci}{\cumalloc_i}
\newcommand{\cumallocs}{\boldsymbol{\cumalloc}}
\newcommand{\remainsupply}{\bar{\supply}}
\newcommand{\remainsupplyi}{\remainsupply_i}
\newcommand{\remainsupplies}{\boldsymbol{\remainsupply}}
\newcommand{\remainbudget}{\bar{\budget}}

\newcommand{\allocitJ}{\alloc_{it}(\switchseq)}

\newcommand{\curpL}{\bar\pL}
\newcommand{\curpR}{\bar\pR}

\newcommand{\supplys}{\boldsymbol{\supply}}
\newcommand{\revokeiell}{\revoke_{i\ell}}
\newcommand{\revokeiellJ}{\revoke_{i\ell}(\switchseq)}

\newcommand{\canrevokeited}{\tilde{\revoke}_{i}^{(t)}}
\newcommand{\canrevoke}{\tilde{\revoke}}

\newcommand{\randomrevoke}{Y}
\newcommand{\randomrevokeit}{\randomrevoke_{it}}

\newcommand{\cumallociBar}{\bar\cumalloc_i}
\newcommand{\cumallociHat}{\hat\cumalloc_i}

\newcommand{\OPTSim}{\textsc{LP-single-switch-Emulator}}

\newcommand{\OPTReS}{\textsc{LP-single-switch-Resolving}}

\newcommand{\OPTSimCan}{\textsc{LP-release-Emulator}}

\newcommand{\OPTSimTilde}{\textsc{LP-joint-cost-Emulator}}

	\DeclarePairedDelimiter{\abs}{\lvert}{\rvert}

\newcommand{\OPTSimInfty}{\textsc{LP-multi-station-Emulator}}
\newcommand{\OPTSimOne}{\textsc{LP-multi-station-Emulator}}

\newcommand{\lpcancelsubproblem}{\text{\ref{eq:opt cancellation}}[\ell,\status]}

\newcommand{\commentcolor}{blue}

\newcommand{\pbias}{\varepsilon}
\newcommand{\pbiast}{\pbias_t}
\newcommand{\pbiass}{\boldsymbol{\predictionbias}}
\newcommand{\biaseddemand}{\hat\demand}

\newcommand{\perrors}{\boldsymbol{\perror}}
\newcommand{\predictionbias}{\pbias}

\newcommand{\pprob}{\delta}
\newcommand{\pprobb}{\boldsymbol{\pprob}}

\newcommand{\dailydemand}{\xi}
\newcommand{\dailydemands}{\boldsymbol{\dailydemand}}
\newcommand{\dailydemandt}{\dailydemand_t}
\newcommand{\prior}{\pi}
\newcommand{\priors}{\boldsymbol{\prior}}
\newcommand{\priort}{\prior_t}
\newcommand{\priork}{\prior_k}

\newcommand{\dailydemandsample}{\tilde{\dailydemand}}
\newcommand{\dailydemandsamples}{\tilde{\dailydemands}}
\newcommand{\dailydemandsamplest}{\dailydemandsamples\ted}

\newcommand{\dailydemandsampletk}{\dailydemandsample\ted_k}

\newcommand{\estimateddemand}{\tilde{\demand}}
\newcommand{\estimateddemandk}{\estimateddemand\ked}

\newcommand{\FMNN}{{\sf Minimax-OPT}}
\newcommand{\FMNNPlus}{{\sf Minimax-OPT++}}
\newcommand{\EMDP}{{\sf Empirical MDP}}
\newcommand{\FMDP}{{\sf Full Info MDP}}
\newcommand{\AGR}{{\sf Naive Greedy}}

\newcommand{\ASIM}{{\sf Naive Bayesian}}

\newcommand{\OLS}{\texttt{OLS}}
\newcommand{\Ridge}{\texttt{Ridge}}
\newcommand{\RF}{\texttt{RF}}

\newcommand{\OLSt}{\OLS^{(t)}}
\newcommand{\Ridget}{\Ridge^{(t)}}
\newcommand{\RFt}{\RF^{(t)}}
\newcommand{\OLSweight}{w_{\OLS}}
\newcommand{\Ridgeweight}{w_{\Ridge}}
\newcommand{\RFweight}{w_{\RF}}

%% file: Paper/math.tex
\newcommand{\condition}{\,\mid\,}

\newcommand{\prob}[2][]{\text{Pr}\ifthenelse{\not\equal{}{#1}}{_{#1}}{}\!\left[{\def\givenn{\middle|}#2}\right]}
\newcommand{\expect}[2][]{\mathbb{E}\ifthenelse{\not\equal{}{#1}}{_{#1}}{}\!\left[{\def\givenn{\middle|}#2}\right]}

\newcommand{\tparen}{\big}
\newcommand{\tprob}[2][]{\text{Pr}\ifthenelse{\not\equal{}{#1}}{_{#1}}{}\tparen[{\def\given{\tparen|}#2}\tparen]}
\newcommand{\texpect}[2][]{\mathbb{E}\ifthenelse{\not\equal{}{#1}}{_{#1}}{}\tparen[{\def\given{\tparen|}#2}\tparen]}

\newcommand{\sprob}[2][]{\text{Pr}\ifthenelse{\not\equal{}{#1}}{_{#1}}{}[#2]}
\newcommand{\sexpect}[2][]{\mathbb{E}\ifthenelse{\not\equal{}{#1}}{_{#1}}{}[#2]}

\newcommand{\indicator}[1]{{\mathbbm{1}\left\{ #1 \right\}}}

\newcommand{\plus}[1]{{\left( #1 \right)^+}}

%% file: Paper/abstract.tex

\revcolor{We consider a natural dynamic staffing problem in which a decision-maker sequentially hires workers over a finite horizon to meet an unknown demand revealed at the end. Predictions about demand arrive over time and become increasingly accurate, while worker availability decreases. This creates a fundamental trade-off between hiring early to avoid understaffing (when workers are more available but forecasts are less reliable) and hiring late to avoid overstaffing (when forecasts are more accurate but availability is lower). This problem is motivated by last-mile delivery operations, where companies such as Amazon rely on gig-economy workers whose availability declines closer to the operating day.

To address practical limitations of Bayesian models (in particular, to remain agnostic to the underlying forecasting method), we study this problem under \emph{adversarial predictions}. In this model, sequential predictions are  adversarially chosen uncertainty intervals that (approximately) contain the true demand. The objective is to minimize worst-case staffing imbalance cost. Our main result is a simple and computationally efficient online algorithm that is minimax optimal. We first characterize the minimax cost against a restricted adversary via a polynomial-size linear program, then show how to \emph{emulate} this solution in the general case. While our base model focuses on a single demand, we extend the framework to multiple demands (with egalitarian or utilitarian objectives), to settings with costly reversals of hiring decisions, and to inconsistent prediction intervals. We also introduce a practical ``re-solving'' variant of our algorithm, which we prove is also minimax optimal. Finally, motivated by our collaboration with Amazon Last-Mile, we conduct numerical experiments showing that our algorithms outperform Bayesian heuristics in both cost and speed, and are competitive with (approximate or exact) Bayesian-optimal policies when those can be computed.}

\newpage

%% file: Paper/intro.tex
\label{sec:intro}
Managing inventory to meet uncertain future demand is a core paradigm deeply rooted in the classical literature in operations and economics---exemplified by foundational models such as the \emph{newsvendor model}, introduced by \citet{Edg-88} in the $19^{\text{th}}$ century and reformulated in the seminal work of \citet{AHM-51,Whi-55}. These models underscore the careful balance needed between the costs of understocking and overstocking in inventory decisions, emphasizing how decision-makers can utilize information about uncertain demand to better navigate this balance. True to its original motivation, the newsvendor model considers scenarios in which the decision-maker places a single order before the demand is realized. These scenarios align well with contexts where placing multiple orders before demand realization is impractical, for example, due to long lead times or requiring advance commitments. Consequently, inventory decisions in such settings typically rely on a \emph{single-shot forecasting} of unknown demand---often derived from historical data and modeled as a prior distribution---and do not incorporate new information that emerges afterward.

However, in many modern inventory management and workforce planning applications, particularly within gig-economy platforms, shorter lead times (e.g. due to gig workers' short response times) enable decision-makers to sequentially make multiple ordering (or staffing) decisions over a planning horizon before demand is realized. As these decisions occur sequentially, newly available data or signals about unknown demand can be incorporated to refine subsequent decisions, naturally allowing \emph{sequential forecasting} of demand rather than relying solely on an initial forecast. Sequential forecasts typically become increasingly accurate as more information becomes available, allowing decision-makers to progressively improve their decisions.

\revcolor{Motivated by our collaboration with \href{https://www.amazon.science/tag/last-mile-delivery}{Amazon Last-Mile Delivery}, we focus on \emph{dynamic staffing for last-mile operations} as our primary example of the scenario described above. Sequential demand forecasts are particularly valuable in this context, as workforce availability and hiring costs change throughout the planning horizon. A rich literature in operations research has explored modeling sequential forecasts in  somewhat similar (but quite stylized) inventory planning settings; e.g., \cite{FR-96} studies a two-order setting, while \cite{WAK-12,SZ-12} consider extensions to multiple orders. By and large, this literature adopts a (specific) full-information \emph{Bayesian} modeling approach, relying on strong distributional assumptions and knowledge about the underlying forecast generation process. However, modern forecasting systems employed by platforms such as Amazon often combine multiple “black box” machine learning/time series algorithms---an approach that does not naturally align with Bayesian modeling.

In light of the aforementioned shortcomings of the Bayesian approach, in this paper, we introduce a novel, practical way of modeling sequential forecasts in dynamic staffing by adopting a robust, distribution-free approach. We then formally study the interaction between sequential forecasts and variations in workforce supply and hiring costs, and its impact on staffing decisions. As we elaborate in the following, the resulting algorithmic framework (and a related fundamental trade-off that we identify) not only addresses our motivating application, but is rather general and can potentially be applied broadly to other contexts.}



\noindent\textbf{Last-mile staffing with sequential forecasting.}  Consider a platform such as Amazon that must plan its workforce for under-the-roof tasks at a last-mile station (e.g., loading trucks or sorting packages) on a particular operating day.\footnote{For simplicity, we mostly focus on a single operating day, which can be thought of as a ``peak'' day with a demand burst that requires separate major planning. Later in \Cref{sec:extension} we study the extension with joint planning for multiple operating days and stations.} Planning usually starts a few weeks in advance and leads up to the operating day. Over this planning horizon, the platform dynamically makes staffing decisions by drawing workers from multiple pools with different initial sizes, whose availability may vary over time at different rates. For example, while platforms commonly rely on full-time \emph{fixed} workers who must be scheduled well in advance of the operating day (analogous to a supply pool with long lead times), they increasingly also use gig-economy \emph{ready} workers, who could be temporarily hired through various third-party online platforms similar to Amazon Flex~\citep{amazonflex2025}. Unlike fixed workers who become unavailable soon after initial staffing, ready workers provide a flexible supply pool with potential availability until the operating day and significantly shorter lead times, thus can be dynamically staffed over the planning horizon.

The required number of workers on the operating day depends on an uncertain target demand, which becomes fully known only on that day itself. This forces the platform to forecast this demand to carefully balance (and minimize) the potential costs of \emph{overstaffing} and \emph{understaffing} (details in \Cref{sec:prelim}). Indeed, Amazon employs a broad array of machine learning and time-series methods to generate sequential forecasts, which then will be utilized by the dynamic (or  \emph{online}) algorithms making these staffing decisions. Notably, the forecasts become increasingly accurate as the operating day approaches due to additional data or signals about demand (see our numerical case study in \Cref{sec:numerical} for a concrete example). As such, it might initially appear optimal to delay hiring as late as possible to leverage the most accurate demand information. However, fixed workers cannot be hired close to the operating day, as they must be scheduled well in advance. In addition, even the pools of ready workers gradually diminish over time, as individuals become less responsive to shorter notices---making the last-minute hiring either infeasible or prohibitively expensive.

These simultaneous changes in supply availability and forecast information over the planning horizon (illustrated in \Cref{fig:tradeoff}) give rise to a fundamental trade-off:
the platform making staffing decisions can either secure workers early, risking overstaffing due to limited demand information, or postpone hiring until later, when forecasts become more accurate but worker availability is reduced, thus risking understaffing. This motivates the following informal research question:
\emph{can we formalize and characterize this ``optimal trade-off'' in a manner applicable to our motivating application?}

\begin{figure}[htb] 
    \centering
    \includegraphics[width=0.65\textwidth]{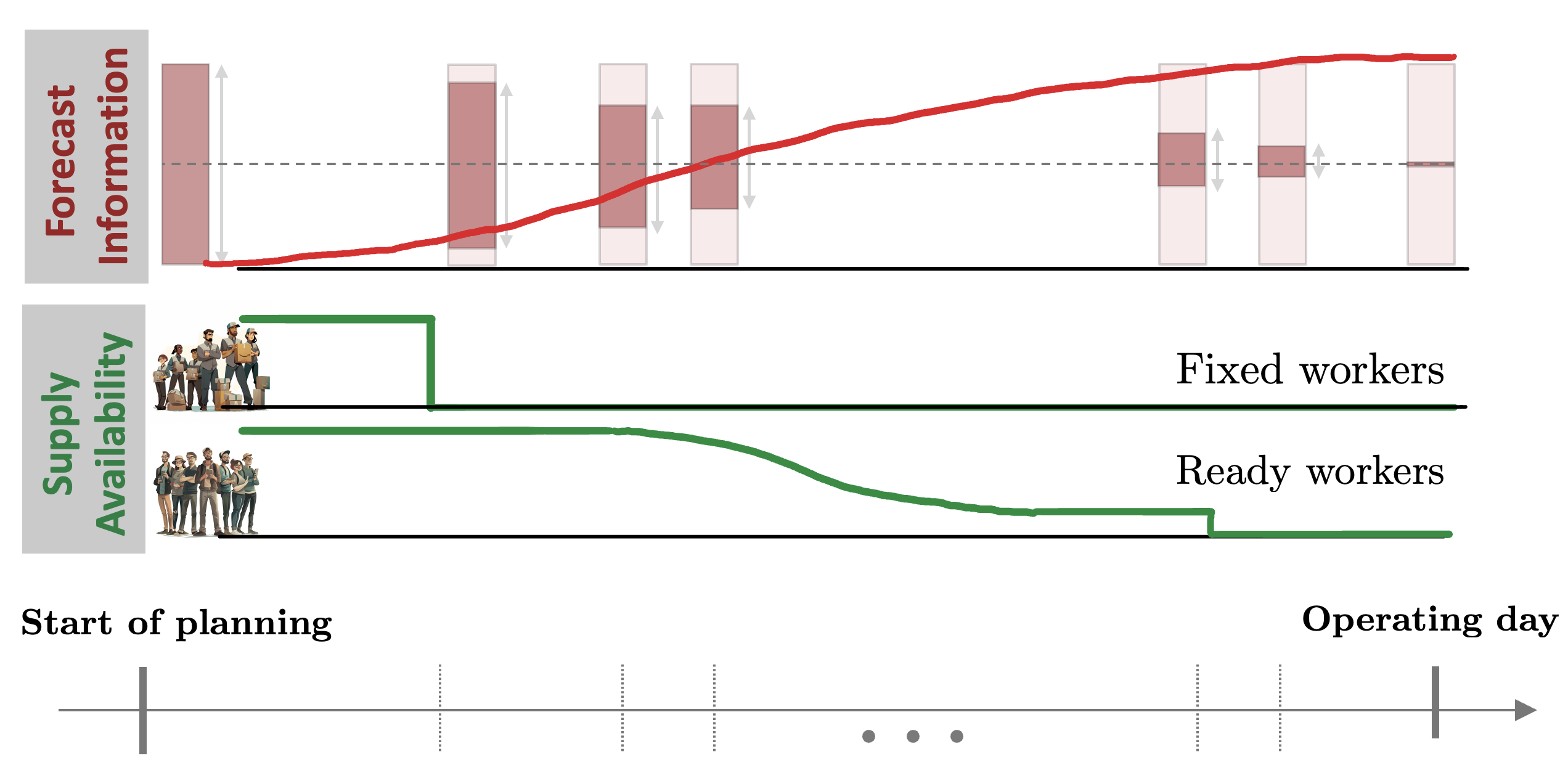} 
    \caption{The fundamental tradeoff between using supply availability and prediction information.}
    \label{fig:tradeoff}
    \vspace{-2mm}
\end{figure}


\smallskip
\noindent\textbf{Robust dynamic staffing \& adversarial predictions.} 
Toward studying the above question, we consider a finite-horizon, discrete-time online staffing problem---with access to sequential forecasts---where at each time period, an online algorithm makes (irrevocable) staffing decisions, i.e., how many workers to hire from each pool. Consistent with the prior work, we assume that workers' availability over time is known. As alluded to earlier, a major departure of our work from the previous literature lies in our approach to modeling sequential forecasts. In principle, à la most of this literature, one could use a standard Bayesian framework to model improving sequential forecasts. This can be done using randomized sequences of Blackwell ordered distributions~\citep{bla-53}, where each distribution is more informative than the preceding one.\footnote{For example, the Martingale Model of Forecast Evolution (MMFE), developed by \cite{hau-69,HJ-94,OO-13}, is a common stylized method in the literature for modeling evolving forecasts using this Bayesian approach.} 

However, this approach is not desirable in our context. Primarily, it requires the decision-maker to know the exact underlying information structure over time, i.e., how precisely the belief evolves, which is overly restrictive. Also, in most non-stylized cases, there is no succinct representation of how the information evolves. Moreover, this approach ties decisions to a specific forecasting method, whereas modern applications often employ multiple ad-hoc, machine-learning-based forecasting algorithms, for which specifying an accurate probabilistic model of forecast is challenging or even impossible.\footnote{For instance, Amazon utilizes forecasting techniques drawn from a wide variety of ML/time-series approaches, including algorithms such as Convolutional Neural Network Quantile Regression (CNN-QR), Deep Recurrent Neural Network (RNN) time-series forecasting (DeepAR+), and Non-Parametric Time Series (NPTS), among others. See \citet{amazonforecast2025} for details.} Last but not least, computing the optimal online policy (or even an approximation) given the sequential information structure is typically computationally demanding~\citep{PT-87} and suffers from the curse of dimensionality, or requires stylized assumptions on how the information evolve (see \Cref{sec:numerical} for a concrete example).

\revcolor{At the same time, it is generally feasible to measure or approximate the ``frequentist accuracy'' of these machine-learning forecasts, including their bias and variance---either empirically or theoretically---based on the amount (and quality) of available prediction data (see \Cref{sec:model justificaiton} for details). To model sequential forecasts in a more practical way using this perspective and to capture such frequentist accuracy measures, we assume that the algorithm observes a \emph{prediction interval}  at each time step, analogous to similar well-studied frequentist notions such as confidence intervals or conformal predictions\footnote{\label{footnote:confidence interval conformal prediction interpretation}Modeling forecast errors via confidence intervals is commonly used to quantify uncertainty in offline statistical and machine learning models~\citep{GA-90,smi-03}, and more recently for adaptive forecasts~\citep{CR-24}. Conformal predictions, both offline and online~\citep{SV-08, gibbs2021adaptive, AB-23}, which calibrate machine learning models to generate uncertainty sets containing the ground truth with a specified probability, are also increasingly prevalent in theory and practice due to their versatility. Similarly, in the robust optimization literature~\citep{BEN-09}, uncertainty sets often take the form of high-dimensional boxes, analogous to our prediction intervals. Finally, in mathematical finance, prediction intervals are commonly employed to characterize the uncertain trajectories of Brownian motions~\citep{MP-10}.}. To decouple staffing decisions from the specifics of demand forecasting, and to ensure robustness and agnosticism with respect to particular forecasting methods (a desirable property in practical applications such as last-mile staffing), we adopt an \emph{``adversarial predictions''} framework. In this framework, we assume that the length of each interval---termed as the \emph{prediction error}---is bounded and known to the platform upfront, mirroring the knowledge of the platform about the amount and quality of the prediction data. Moreover, we assume that the prediction intervals are {\em consistent}, which means that they contain the target demand (we relax this assumption to approximate consistency later in the paper). Other than these, we impose \emph{no} additional structural assumptions on these intervals, effectively allowing them to be selected adversarially, either obliviously or adaptively. }

Given the above ingredients, our goal is to design a computationally efficient online algorithm that at each time observes a prediction interval before making its staffing decision. The algorithm aims to minimize the staffing imbalance cost at the end of horizon against the worst-case adversary, where the adversary selects the final demand and a valid sequence of (history-dependent) prediction intervals, subject to the given prediction error bounds and (approximate) consistency. We now pose the following formal research question. 
\begin{displayquote}
   \myboxtwo{\emph{Can we design and characterize a computationally efficient, robust online algorithm that achieves the optimal worst-case staffing imbalance cost against any (adversarial) sequence of prediction intervals?}}
\end{displayquote}

\noindent
\textbf{Our main contributions.}  We view robust dynamic staffing with predictions as a high-dimensional \emph{min-max optimization (or game)} with imbalance cost as payoff, in which the maximizer (i.e., the adversary) selects an action from the space of valid prediction intervals and target demands, while the minimizer (i.e., the decision-maker) chooses among feasible online algorithms. Now the above question seems challenging, as the spaces from which the decision-maker and the adversary choose their actions are doubly exponential and exponential in size, respectively. Yet, somewhat surprisingly, we can answer this question affirmatively in various settings. More formally, we establish the following computational/algorithmic result.

\mybox{
\begin{displayquote}
\textbf{\underline{(Informal) Main Result}:} There exists a simple, interpretable, and deterministic polynomial-time online algorithm, informally called the \emph{``LP-based emulator,''} that is minimax-optimal, with suitable polynomial-time generalizations to various practical extensions of our problem. 
\end{displayquote}
}
Through the LP-based emulator, we essentially show how to algorithmically capture the optimal tradeoff between early greedy staffing (when workforce supply is abundant, but demand information is limited) and later staffing (when demand information is more accurate, but worker availability has decreased). This online staffing algorithm, formally described as \Cref{alg:opt} for the base model with a single-station and multiple heterogeneous supply pools, operates in two main steps:

\smallskip
\begin{enumerate}[label=(\emph{\roman*}),leftmargin=0.22in]
\item \revcolor{\emph{Solving an offline Linear Programming (LP):} Given known model parameters---i.e., the initial prediction interval, as well as supply availability and prediction error trajectories over the horizon (reflecting changes in supply and information)---we identify particular polynomial-size LPs (e.g., \ref{eq:opt reduced form} in \Cref{sec:base model} for the base model). We then show that these LPs \emph{exactly} characterize the worst-case staffing imbalance costs of minimax-optimal algorithms, or equivalently, the minimax values of these problems. By solving this LP upfront, our algorithm not only computes the minimax value of the underlying game, but also obtains the optimal staffing decisions against a \emph{restricted} family of prediction intervals (explained shortly). This LP solution subsequently serves as a ``guide'' for future staffing decisions.}

\item\revcolor{\emph{Running online emulation:} Next, we introduce novel procedures (such as Procedure~\ref{alg:emulator} in \Cref{sec:base model} for the base model) that leverage the LP solutions to guide online staffing decisions. These procedures enable online algorithms to achieve minimax optimality against \emph{any possible} sequence of prediction intervals, rather than only against those that belong to the aforementioned restricted family. In essence, our algorithms iteratively project the LP solutions onto feasible online decisions, dynamically adjusting them as new prediction intervals are revealed---thus effectively ``emulating'' the offline LP in every instance.}

\end{enumerate}
\smallskip

We begin by examining a simple single-station, single-pool setting with perfectly consistent predictions as a warm-up in \Cref{sec:simple instance}.  We then work our way up by expanding to more general settings with multiple supply pools and approximately consistent predictions in \Cref{sec:base model result}---our base model---and subsequently explore practically relevant extensions in \Cref{sec:extension}. These include settings with (i) multiple stations (or operating days) sharing the same workforce pools, each with its own target demand (\Cref{sec:multi-station}), (ii)~costly hiring and releasing with budgets, where these costs are heterogeneous for different pools and vary over time, typically becoming more expensive as the operating day approaches (\Cref{sec:costly-hiring} and \Cref{sec:cancellation}), and (iii)~jointly minimizing the cost of hiring and over-/under-staffing (Section~\ref{apx:combinded objective}). Our framework and algorithmic results are sufficiently flexible to extend to \emph{all} these practical variants.



\revcolor{We also consider a variant algorithm in \Cref{sec:lp-resolving-main} by refining our minimax optimal LP-based emulator approach using the idea of \emph{resolving}. This new algorithm updates our original LP each time a new prediction interval arrives by plugging in the new history-dependent state of the algorithm---which includes the current staffing levels and available workers in each supply pool, as well as the last prediction interval. Then it resolves this updated LP at each time to obtain the staffing decision at that time. We theoretically analyze this refined algorithm (\Cref{alg:lp-resolve}) and show that it remains minimax optimal. In addition, due to its resolving nature, this new algorithm is expected to outperform the original algorithm in practical instances.

Finally, inspired by our primary motivating application in last-mile delivery involving two worker pools (fixed and ready workers), we also conduct comprehensive numerical simulations in \Cref{sec:numerical} to empirically evaluate the performance of our LP-based emulator algorithm and its practical refinement based on resolving. We compare their performance with other heuristic benchmarks used in practice, as well as high-dimensional near-optimal online policies (when they are feasible to compute). Importantly, our algorithms take advantage of prediction intervals generated by aggregating various machine learning forecasting methods, while the other benchmarks use estimated or exact distributional knowledge of the final demand. Our numerical results demonstrate that both of our algorithms with access to sequential prediction intervals significantly outperform these Bayesian benchmarks in terms of staffing costs and computational efficiency.  See \Cref{tab:numerical:long instance,tab:numerical:short instance} in \Cref{sec:numerical} for details.}

\subsection{Summary of our Techniques}
Our work draws on a variety of techniques from online algorithms and game theory to design our algorithms and establish their minimax optimality. Below, we highlight some of these methods and key technical ideas.

\smallskip
\noindent\textbf{Greedy staffing with target overstaffing upper bound.} In our simple warm-up setting in \Cref{sec:simple instance},  the online algorithm decides only on the number of workers to hire from a single pool in each round, while receiving a sequence of perfectly consistent prediction intervals that always contain the true demand (hence, these intervals could be assumed to be ``nested'' without loss of generality). This simplicity allows us to clearly isolate and study the trade-off between hiring early and late by focusing exclusively on the dynamics of supply availability and forecast accuracy.

\revcolor{We make a straightforward, yet crucial observation: An algorithm aiming to keep the overstaffing below a certain target can do so by ``underestimating'' the demand using the lower bound of each prediction interval and then properly adjusting the staffing level to ensure that it never exceeds a certain threshold equal to this underestimated demand plus the target. Since hiring earlier is always preferable from a supply-availability perspective---and reduces the risk of understaffing---this insight naturally suggests a simple greedy algorithm (see \Cref{alg:opt simple instance}): given a target upper bound on allowable overstaffing, hire workers as early as possible while maintaining supply feasibility and respecting a certain upper threshold on the staffing level (as described above). We then show how to optimally set this target upper bound via a fixed-point argument, resulting in a minimax-optimal algorithm (see \Cref{fig:geometric-proof} and \Cref{prop:opt alg simple instance}). Finally, we present numerical examples to illustrate how our algorithm resolves the early-versus-late hiring trade-off (see \Cref{fig:tradeoff illustration example} in \Cref{sec:simple instance}).}

Through the above investigation, we identify an important structural property of the worst-case adversary in our warm-up setting: the adversary selects prediction sequences with a \emph{single switching} structure. Such an adversary initially chooses prediction intervals that ``signal'' high demand; then, at a specific adversarially-chosen time, it switches to intervals signaling low demand (while remaining consistent with the prediction history). In the simplest model, this corresponds to first changing only the lower bound of the prediction interval (keeping the upper bound fixed) and subsequently switching to change only the upper bound (keeping the lower bound fixed). Intuitively, the adversary benefits from initially signaling high demand before switching to low demand---rather than the reverse---since the online algorithm can only increase its staffing level. We leverage this structural observation to establish our main result.


\smallskip
\noindent\textbf{Zero-sum games, single-switch adversaries, \& LPs.} 
Building upon this warm-up, in our base model in \Cref{sec:base model result} we study a more general setting in which the platform must make staffing decisions for a single station, given access to multiple heterogeneous supply pools---each characterized by its own initial size and availability dynamics. We also allow prediction intervals that not only have limited (yet improving) accuracy but are also only \emph{approximately consistent} (see \Cref{asp:prediction sequence}). This setting is considerably more challenging: some of the key monotonicity properties that previously guaranteed the optimality of a greedy algorithm in \Cref{sec:simple instance} no longer hold, as staffing decisions across different pools become coupled through the objective function and the adversary’s selection of target demand. Moreover, while earlier staffing continues to improve supply feasibility, determining the desired threshold on overstaffing for each pool at each time becomes computationally difficult in this setting, as it now involves a high-dimensional search.


As mentioned earlier, our minimax problem can be viewed as a two-player zero-sum (Stackelberg) game with staffing imbalance cost as the zero-sum payoff. The pre-specified supply availability and prediction error trajectories---which determine the availability rate for each pool and the lengths of prediction intervals over time---are fixed in advance. The \emph{leader} (the ``min-player'') is the decision-maker, who designs an online staffing algorithm that respects supply feasibility. The \emph{follower} (the ``max-player'') is an adversary who selects the target demand and a sequence of (approximately) consistent prediction intervals, constrained by the predetermined prediction error trajectories. Importantly, as alluded to earlier, characterizing the equilibrium of this Stackelberg game is challenging, due to both players having high-dimensional action spaces: the decision-maker considers all feasible (and possibly randomized) online algorithms, while the adversary explores all potential (possibly history-dependent) prediction sequences and demands.\footnote{Without discretization, infinitely many prediction sequences could be chosen by the adversary.}


To determine the optimal min-player strategy in the base model, we leverage the structural insight gained from analyzing the worst-case adversary in the simpler warm-up scenario. Specifically, guided by that insight, we restrict the adversary's action space to a smaller subset of single-switch prediction sequences. Facing this constrained single-switch adversary, we demonstrate that the minimax-optimal staffing algorithm can be characterized by a polynomial-size LP. In this LP, decision variables represent staffing levels, while constraints capture each possible adversarial switching time to limit overstaffing, given the demand prediction errors and inconsistencies defined in \Cref{asp:prediction sequence}. Additional constraints ensure supply feasibility and control for potential understaffing at the end of the horizon (see \ref{eq:opt reduced form} for our base model with a single station and multiple pools). Finally, the objective function of this LP is the staffing imbalance cost. 

This LP formulation provides a lower bound (i.e., a relaxation) on the minimax value of our original Stackelberg game, which is equal to the optimal worst-case staffing cost.
Furthermore, we extend it to several practically relevant generalizations: \ref{eq:opt reduced form multi station} addresses scenarios with multiple stations (\Cref{sec:multi-station}); \ref{eq:opt cancellation} incorporates costly hiring and releasing with budget constraints (\Cref{sec:costly-hiring}); and \ref{eq:opt variant cost} captures jointly minimizing staffing and hiring costs (\Cref{apx:combinded objective}).


\smallskip
\noindent\textbf{Online emulations \& minimax optimal online algorithm.} 
The argument above does not fully characterize the minimax value of the original game, as it addresses only the surrogate relaxation game where the adversary is restricted to single-switch prediction sequences. Therefore, as the second step of our framework, we show that constraining the adversary to single-switch predictions is actually \emph{without loss}. Specifically, we introduce a novel online emulation step (Procedure~\ref{alg:emulator}), which uses as input the minimax optimal algorithm against a single-switch adversary (i.e., the optimal solution to the LP), and outputs feasible staffing decisions in an online manner, regardless of whether the actual prediction sequence is single-switch. Intuitively, the resulting online algorithm (\Cref{alg:opt}) closely tracks the optimal LP solution over time, dynamically adjusting staffing decisions (in fact, lowering them) as predictions of target demand evolve. By establishing a critical \emph{invariant property} of our online emulation step, we show that the staffing cost under any adversarial prediction sequence is never greater than the optimal staffing cost against the single-switch adversary. This ensures that our online algorithm is indeed minimax optimal against all possible adversarial strategies.


\smallskip
\noindent\textbf{LP-based emulators for extensions: configurations LPs \& resolving.} As natural extensions of our base model, in \Cref{sec:multi-station}, we focus on a setting with multiple stations using shared worker pools, where the decision-maker aims to minimize an objective that combines staffing imbalance costs across stations---either by taking the maximum (an egalitarian approach) or the sum (a utilitarian approach). Next, in \Cref{sec:costly-hiring} and \Cref{sec:cancellation}, we consider a setting with costly hiring and releasing, where the platform pays to hire workers and can \emph{reverse} earlier hiring decisions by paying a cost, subject to a fixed budget. Finally, in \Cref{apx:combinded objective}, we analyze a setting where hiring incurs a cost that is integrated into the objective function, resulting in a \emph{mixed objective} to be minimized. Although these extensions are substantially more complex than our base model, we still develop minimax-optimal online algorithms that run in polynomial time. At a high level, these algorithms adopt a similar architecture to \Cref{alg:opt}, emulating the optimal solution of a linear program. However, in some cases, more intricate LP formulations are required---\emph{configuration LPs} that capture combinatorial allocations from hiring pools---or even a sequence of LPs that must be resolved in each step (as in the case of costly release). We defer further technical details to the later sections.

\subsection{Practical \& Managerial Insights}
\label{sec:insights}

\revcolor{
\noindent\textbf{Numerical results.} 
To empirically evaluate our algorithms, we conduct numerical experiments with synthetic data in \Cref{sec:numerical}. In our setting, the platform hires workers from two supply pools---ready workers and fixed workers---whose availability follows the framework described earlier (see also \Cref{fig:tradeoff}). While our theoretical results focus on adversarial environments without distributional assumptions on demand, our experiments adopt a Bayesian perspective: the final demand accumulates from partial daily demands, which are drawn independently from (unknown) distributions and revealed sequentially to the platform. In this Bayesian setting, the \vrevcolor{full-information} instance-optimal online algorithm is well-defined and can be obtained by solving a finite-horizon Markov Decision Process (MDP). \vrevcolor{However, this policy requires the exact knowledge of the transition probabilities---and therefore distributional knowledge about the partial demand generative process. Moreover, discretization is required for tractability, as both the state and action spaces are continuous.}

\vrevcolor{We study two setups for forming prediction intervals used by our own algorithms, based on access to information about future partial daily demands. In the first setting, we assume that the platform has only access to samples of future demands. In such settings, prediction intervals are constructed from  realized partial demands combined with future samples.} In the second setting, \vrevcolor{we assume that the platform does not directly have sample access to future partial demands; as such, it relies on predictions of three machine learning models---linear regression, ridge regression, and random forest---which are trained using offline empirical samples as input data, and generate sequential point-estimates of the final demand by taking realized partial demands at each time as features.}
In both setups, forecast accuracy naturally improves over time as \vrevcolor{more partial daily demands} are observed and the uncertainty in the remaining days diminishes.





In the first setup, we benchmark our proposed algorithms, originally designed for adversarial environments, against \emph{empirical discretized optimal online algorithm} (EmDisOPT). \vrevcolor{This benchmark is the solution of an estimated discretized MDP, where the action and state spaces are discretized and offline empirical samples of future partial demands are used to estimate the MDP transition probabilities. Notably, this benchmark converges to the Bayesian optimum as the sample size increases and the discretization becomes finer.} In 14-day horizon experiments, our algorithm \emph{weakly} outperforms EmDisOPT, reducing costs by 5.6\% on average while running 18,194 times faster (see \Cref{tab:numerical:long instance}).
Finally, when compared with simple heuristics motivated by our industry collaborator, our approach consistently achieves substantial cost reductions. 

In the second setup, we compare our algorithms against the same set of heuristic policies. (Implementing the analog of EmDisOPT is not feasible here due to the even greater computational complexity of treating the entire demand history as the state.) Consistent with the first setup, our algorithms deliver significant cost savings across a broad range of parameter settings. In particular, because the point forecasts produced by machine learning models are typically biased, heuristics based directly on them perform poorly. By contrast, our algorithms achieve costs that are nearly five times lower in 14-day horizon experiments (see \Cref{tab:numerical:long instance 3ML}).}


\smallskip
\noindent\textbf{Further insights \& takeaways.} 
Our framework shares structural similarities with various algorithms in Bayesian online decision-making contexts (e.g., prophet inequalities and stochastic online matching). Such algorithms typically solve an ex-ante relaxation (or fluid approximation) to derive an optimal offline solution, and then employ this solution as a ``canonical solution'' to inform online decisions. This is often achieved through online rounding techniques, such as online contention resolution schemes, which dynamically adjust the ex-ante relaxation (e.g., \citealp{ANSS-19, FNS-24}). However, our approach diverges significantly by addressing an adversarial rather than Bayesian environment, making both our analog of the ex-ante relaxation (\ref{eq:opt reduced form}) and the corresponding online adjustment procedure (Procedure~\ref{alg:emulator}) substantially more involved. This novel framework may therefore be of independent interest.

Beyond the dynamic staffing problem, our setting highlights how classical decision-making problems can be revisited under a new informational paradigm, where adversarial yet progressively improving predictions are revealed sequentially over time. This feature is common in many other sequential decision-making problems. For example, in the ski rental problem \citep{KMMO-94,BE-05}, it is natural to receive increasingly accurate forecasts about the remaining ski season, or in single-leg revenue management \citep{BQ-09,BKK-23,GJZ-23}, it makes sense to have sequentially improving forecasts of future demand. Incorporating adversarial prediction models into these problems opens up intriguing algorithmic questions and can lead to more realistic decision-making solutions.

\revcolor{We conclude this introduction by highlighting that our work is related to various lines of work in operations research, computer science, and economics. We defer the discussion of further related work to \Cref{sec:further}.}

%% file: Paper/prelim.tex
\label{sec:prelim}
Motivated by applications in last-mile delivery, we study the \emph{dynamic staffing with adversarial predictions} problem. In the following, we describe various components of our base model with multiple workforce pools and a single station. Extension models---including settings with multiple stations and operating days (\Cref{sec:multi-station}), and scenarios with costly hiring and releasing of workers under budget constraints (\Cref{sec:costly-hiring})---are discussed in \Cref{sec:extension}.

\xhdr{Setting \& notations.} 
Consider a platform tasked with sequential workforce planning over a finite time horizon of $T+1$ days \revcolor{(also referred to as \emph{times}),} where $T\in\naturals$. The platform sequentially hires workers during the first $T$ days to meet a target demand $\demand\in\mathbb{R}_+$ on day $T+1$ (for example, the number of workers needed to deliver packages). We refer to day $T+1$ as the ``operating day.'' The platform can hire workers from $n$ heterogeneous worker pools, each initially containing $\supplyi\in\mathbb{Z}_{\geq 0}$ workers. Each worker is in one of two possible states---\emph{available} or \emph{unavailable}---on each day. Initially, all workers are available. At the beginning of each day $t\in[T]$, each available worker in pool $i\in[n]$ becomes unavailable with probability $\alpha_{it}\in[0,1]$ (independently across time); otherwise, the worker remains available. Once unavailable, workers remain unavailable thereafter (e.g., due to commitments to other jobs). Given the available workers on day $t\in[T]$, the platform irrevocably hires $\allocit\in\mathbb{Z}_{\geq 0}$ workers from each pool $i$.\footnote{In our base model, we assume that the algorithm cannot reverse its hiring decisions. However, in some practical scenarios, workers may be released or recalled at an additional cost. We explore this extension in \Cref{sec:costly-hiring} and \Cref{sec:cancellation}.} A sequence of staffing profiles $\{\allocit\}_{i\in[n],t\in[T]}$ is \emph{supply feasible} (or simply \emph{feasible}) if the number of hired workers from each pool $i\in[n]$ on each day $t\in [T]$ does not exceed the number of available workers in that pool on that day.

Given a staffing profile $\xbf=\{\allocit\}_{i\in[n],t\in[T]}$ and demand $\demand$, the platform's staffing cost $\Cost[\demand]{\xbf}$ (incurred on the operating day $T+1$) is defined as:
\begin{align*}
    \Cost[\demand]{\xbf} \triangleq 
    &~
    {
    \undercost\cdot
    \plus{\demand - \sum\nolimits_{i\in[n]}\sum\nolimits_{t\in[T]}\allocit}
    }
    +
    {
    \overcost\cdot
    \plus{\sum\nolimits_{i\in[n]}\sum\nolimits_{t\in[T]}\allocit - \demand}
    }
\end{align*}
where we use the notation $\plus{x} \triangleq \max\{0,x\}$, and the parameters $\undercost\in\reals_+$ and $\overcost\in\reals_+$ represent the \emph{per-unit understaffing} and \emph{overstaffing costs}, respectively.\footnote{The staffing cost can equivalently be expressed as $\Cost[\demand]{\xbf}=\max\left\{\undercost\cdot
    (\demand - \sum\nolimits_{i\in[n]}\sum\nolimits_{t\in[T]}\allocit),\overcost\cdot 
    (\sum\nolimits_{i\in[n]}\sum\nolimits_{t\in[T]}\allocit - \demand)\right\}$.} For simplicity of exposition, we assume linear cost functions in the remainder of the paper. Almost all our results extend immediately to more general settings involving cost functions $c(\cdot)$ and $C(\cdot)$, whose inputs are under-staffing and over-staffing, respectively. We only require these functions to be weakly increasing and weakly convex on $\reals_+$ with $c(0) = C(0) = 0$.

\xhdr{Unknown demand and sequential forecasts.}
In our model, the demand $\demand\in[\pLzero, \pRzero]$ is not revealed to the platform until the operating day $T + 1$, where $[\pLzero, \pRzero]$ is the initial demand range. However, the platform receives sequential forecasts for the unknown demand $\demand$ at the beginning of each day $t\in[T]$. Specifically, the initial interval $[L_0,R_0]$ is known in advance, and on each day $t=1,2,\ldots,T$, the platform observes a \emph{prediction interval} $\predictiont = [\pLt,\pRt]$ (simply referred to as a ``prediction''). We impose the following regularity assumption on the prediction sequence $\predictions\triangleq\{\predictiont\}_{t\in[T]}$.

\begin{assumption}[\textbf{Regularity of predictions}]
\label{asp:prediction sequence}
    The prediction sequence $\predictiont = [\pLt,\pRt]$ satisfies the following properties for all $t=1,2,\ldots,T$:
    \begin{enumerate}
        \item \underline{{$(\pbiass,\pprobb)$-Consistency}}: The prediction $\predictiont$ is $(\pbias_t,\pprob_t)$-consistent; that is, there exists an unknown point estimate of demand $\biaseddemand_t$ such that $\prob{~\abs{\demand - \biaseddemand_t}\leq \pbias_t~~~\&~~~\biaseddemand_t \in [\pLt,\pRt]~}\geq 1-\pprob_t$.
        \item \underline{{$\perrors$-Bounded error}}: The prediction error is bounded by $\perrort$; that is, $\pRt - \pLt \leq \perrort$.
    \end{enumerate}
\end{assumption}
We refer to $\pbiass = \{\predictionbias_t\}_{t\in[T]}$ and $\perrors = \{\perrort\}_{t\in[T]}$ as the \emph{prediction inconsistency upper bounds} and \emph{prediction error upper bounds}, respectively, both of which are assumed to be known to the platform. We refer to $\pprobb=\{\pprob_t\}_{t\in[T]}$ as the \emph{miscoverage probability}, which should be thought of as a small quantity, say $\mathcal{O}(\frac{1}{T})$ or even smaller, such that $\sum_{t\in[T]}\pprob_t=\mathcal{O}(1)$. \revcolor{We say that the predictions are \emph{perfectly consistent} if $\pbias_t=\pprob_t=0,~ \forall t\in [T]$.} Further interpretation and justification of this regularity assumption are provided in \Cref{sec:model justificaiton}. For analytical convenience, we also introduce the dummy notations $\predictionbias_0 = 0$ and $\perror_0 = \pRzero - \pLzero$.

\xhdr{Fluid approximation.} With ``large'' systems in mind---\vrevcolor{that is, scaling supply and demand sizes to be large while keeping other parameters including $T$ constant---}we consider a deterministic \emph{fluid approximation} of our problem. First, we allow demand, supply, and staffing decisions at each time to take fractional values (after normalization by the large market scale), that is, $\demand\in[\pLzero, \pRzero]\subseteq\reals_+$, and $\supplyi,\allocit\in\reals_+$. Second, we simplify the stochastic evolution of worker availability by replacing the number of available workers in each supply pool at each time with its expectation. Specifically, suppose pool $i$ has $\supply$ available workers at the end of day $t-1$; on day~$t$, $(1-\alpha_{it})\cdot \supply$ workers remain available, while the remaining $\alpha_{it}\cdot \supply$ become unavailable. Given this fluid approximation, a (fractional) staffing profile $\{\allocit\}_{i\in[n],t\in[T]}$ is said to be supply feasible (or simply feasible) if:
\begin{align*}
\forall i\in[n],~\forall t\in[T]:~~\allocit\leq \Bigg(\Big(\ldots\big(\left(\supplyi\left(1-\alpha_{i1}\right)-x_{i1}\right)\left(1-\alpha_{i2}\right)-x_{i2}\big)\ldots\Big)\left(1-\alpha_{i{t-1}}\right)-x_{i{t-1}}\Bigg)(1-\alpha_{it})~.
\end{align*}
By defining the \emph{availability rate} of pool $i$ at time $t$ as $\wdiscountit\triangleq \prod_{\tau\in[t]}(1-\alpha_{i\tau})$ (that is, the probability that a worker in pool $i$ remains available during days $[1:t]$), the above $nT$ constraints can equivalently be rewritten as $n$ constraints, one for each pool $i\in[n]$, by rearranging terms:
\begin{align}
\label{eq:supply-feasible}\tag{\textsc{Supply-Feasibility}}
\forall i\in[n]:~~~~~~~~~\displaystyle\sum\nolimits_{t\in[T]}\frac{1}{\wdiscountit}\allocit \leq \supplyi~.
\end{align}
Note that $\frac{1}{\wdiscountit}\allocit$ is essentially the effective number of workers needed in the initial pool $i$, so that $\allocit$ number of these workers remain available on day $t$. For simplicity, we focus on this fluid approximation throughout the paper.\footnote{Using standard independent randomized rounding and concentration bounds, all our results (up to an additive small error) naturally extend to the original stochastic setting when supply pool sizes are large and the state of each pool is observable at any time.} We also refer to the tuple $\instance\triangleq\left(n, T, \{\supplyi,\wdiscountit\}_{i\in[n],t\in[T]} ,\pLzero,\pRzero,\{\pbias_t, \perrort\}_{t\in[T]},\undercost,\overcost\right)$ as an \emph{instance} of the dynamic staffing with adversarial predictions problem.

\xhdr{Timeline.} We formalize the timeline of the model below.
\begin{itemize}
    \item On day 0: platform has the following prior information: the number of days $T$, initial supply pool sizes $\{\supplyi\}_{i\in[n]}$, availability rates $\{\wdiscountit\}_{i\in[n],t\in[T]}$, initial demand range $[\pLzero,\pRzero]$, prediction inconsistency and error upper bounds $\{\pbias_t,\perrort\}_{t\in[T]}$,  per-unit understaffing cost $\undercost$, and per-unit overstaffing cost $\overcost$.
    \item On each day $t\in[T]$: 
    \begin{itemize}
        \item prediction $\predictiont = [\pLt,\pRt]$ is revealed to the platform,
        \item workers' availability is updated and the platform observes current supply pool sizes,
        \item the platform chooses a feasible staffing profile $\{\allocit\}_{i\in[n]}$.
    \end{itemize}
    \item On day $T + 1$: demand $\demand$ is revealed and the total cost $\Cost[\demand]{\xbf}$ is computed.
\end{itemize}

\xhdr{Robust online algorithm design under worst-case cost.} A \emph{feasible online algorithm} is an algorithm that (i)~at any time $t$, makes (fractional) staffing decisions $\{\allocit\}_{i\in[n]}$ only based on its prior information (on day $0$ as outlined above), history, and current prediction $\predictiont = [\pLt,\pRt]$, and (ii) its resulting staffing profile $\{\allocit\}_{i\in[n],t\in[T]}$ is feasible (in fluid approximation). Focusing on robust performance, we evaluate the performance of any feasible online algorithm used by the platform with its \emph{cost guarantee}, defined formally below.

\begin{definition}[\textbf{Cost guarantee}]
\label{def:minmaxcost}
    Given an instance $\instance$, the \emph{cost guarantee} of an online algorithm $\ALG$ is defined as its staffing cost against worst-case adversarial predictions and demand, that is,
    \begin{align*}
        \max_{\predictions,\demand}\expect{\Cost[\demand]{\ALG(\predictions)}}
    \end{align*}
    where $\ALG(\predictions)$ is the (possibly randomized) staffing profile generated by algorithm $\ALG$ in an online fashion under prediction sequence $\predictions= \{\predictiont\}_{t\in[T]}$.
\end{definition}

A feasible online algorithm $\OPTALG$ is said to be \emph{minimax optimal} if it has the minimum cost guarantee among all feasible online algorithms, i.e., 
\begin{align*}
    \OPTALG \in \underset{\substack{\textrm{feasible online}\\\textrm{algorithm}~\ALG}}{\argmin}\max_{\predictions,\demand}\expect{\Cost[\demand]{\ALG(\predictions)}}
\end{align*}
For a given instance, we refer to 
 $\optcost\triangleq\max_{\predictions,
    \demand}\expect{\Cost[\demand]{\OPTALG(\predictions)}}$ as the \emph{optimal minimax cost}, i.e., the cost guarantee of the minimax optimal online algorithm.\footnote{Mathematically speaking, we need to use ``$\sup$'' and ``$\inf$'' when defining our minimax optimal algorithm; however, as we establish in this paper, the equilibrium will indeed be achieved, and hence ``$\max$'' and ``$\min$'' are well-defined.}

\subsection{Discussion on the Model Primitives}
\label{sec:model justificaiton}
We next explain several key modeling choices made in our problem formulation.

\xhdr{Understaffing and overstaffing costs.} Our model accommodates asymmetric understaffing and overstaffing costs. \revcolor{In the last-mile delivery context, understaffing costs capture operational expenses from relying on overtime work, as retailers typically aim to avoid delays or failures in package deliveries. These costs are high, both financially and in terms of compliance with labor laws.} Overstaffing costs, on the other hand, reflect the opportunity costs of assigning excess workers to specific tasks, as well as operational costs from last-minute rescheduling. For simplicity of exposition, we assume linear cost functions in our base model; however, most of our results naturally extend to general cost functions for understaffing and overstaffing that are weakly increasing and weakly convex (see \Cref{apx:optalginfty,apx:optalgcan}).


\xhdr{Unknown demand and known supply.} In the last-mile delivery industry, the magnitude of uncertainty on the demand side typically differs significantly from that on the supply side. Specifically, demand uncertainty tends to be much greater, as it can be influenced by various factors such as sales events or social media trends. In contrast, the availability of workforce pools (supply) is usually more stable and predictable based on historical data. Motivated by this discrepancy, our model incorporates sequential forecasts on the demand side, while assuming known fluid trajectories for supply availability. Importantly, our results \emph{directly generalize} to scenarios in which the platform only has consistent interval predictions regarding availability rates $\{\wdiscountit\}$. In such settings, it is optimal for the adversary to pick the actual availability rate equal to the lower bound of the prediction intervals, and thus reduces to our base model from the platform's perspective.

\xhdr{$(\pbiass,\pprobb)$-consistent and $\perrors$-error-bounded predictions.}
The prediction intervals can be naturally interpreted as ``uncertainty sets'' or ``confidence intervals.'' As discussed in the introduction (Footnote~\ref{footnote:confidence interval conformal prediction interpretation}), this approach for expressing uncertainty is commonly used across various literature such as robust optimization~\citep{BEN-09}, machine learning~\citep{GA-90,AB-23}, pricing and mechanism design~\citep{CLL-17}, and mathematical finance~\citep{MP-10}. In PAC-learning-based forecasting methods (e.g., regression-based predictions), the length of the uncertainty set---captured by the prediction error $\perror$ in \Cref{asp:prediction sequence}---can be explicitly calculated using the sample complexity of the underlying prediction method (e.g., the universal PAC-learning bound with constant VC-dimension~\citep{SB-14}, that with $\mathcal{O}\left(\ln(\frac{1}{\pprob})/{\perror^2}\right)$ samples we can have an uncertainty set of length $\perror$ that is valid with probability at least $1-\pprob$). Motivated by this, we assume prior knowledge of the upper bounds on prediction errors, reflecting the knowledge of the sample size of the dataset used in demand forecasts for each day. Due to the logarithmic dependency of sample complexity on $1/\pprob$, it is also realistic to consider regimes where $\pprob=\frac{1}{T^\gamma}$ for sufficiently large $\gamma>0$, as this increases the sample complexity only by logarithmic factors. \revcolor{We also highlight that the number of days $T$ in our model is finite, and in practically relevant regimes of our problem is not extremely large (say an integer between 5 to 21).}


As for the interpretation of the $\pbiass$-consistency assumption in our motivating application, the unknown demand $\demand$ on the operating day $T+1$ might evolve over the planning horizon. In such a case, the sequence $\{\biaseddemand_t\}_{t\in[T]}$ in \Cref{asp:prediction sequence} represents the trajectory of these demand changes for $t = 1, \dots, T$. Such changes may arise from various sources, e.g. external shocks due to unexpected high-volume traffic on the Amazon website, which cannot be accurately captured by standard machine-learning-based forecasts. Following this perspective, these changes can also be seen as representing the inherent \emph{bias} present in the predictive models used. See \Cref{sec:numerical} for a demonstration in a simulated case study.




\xhdr{Oblivious vs.\ adaptive adversary.} In \Cref{def:minmaxcost}, we consider an \emph{oblivious adversary} who selects the prediction sequence $\predictions$ and demand $\demand$ non-adaptively. As becomes clear later, since our proposed algorithms are deterministic, our results automatically extend to the case of an \emph{adaptive adversary} who can choose the prediction $\predictiont$ on day $t$ (resp. demand $\demand$ on day $T + 1$) after observing the realizations of the randomized staffing profiles in previous $t-1$ days (resp. $T$ days).

\xhdr{Worst-case cost vs.\ other robust criteria.} In this work, we evaluate the performance of online algorithms by its worst-case cost guarantee. Our proposed algorithms are also optimal for other robust criteria (regret and competitive ratio) under a mild assumption. See \Cref{apx:regret}.

%% file: Paper/base-model.tex
\label{sec:base model}
In this section, we focus on the baseline model introduced in \Cref{sec:prelim} and demonstrate how to design and analyze a minimax-optimal online algorithm. First, in \Cref{sec:simple instance}, we build intuition for our main technical ideas by analyzing a simple special case of our base model that includes a single supply pool and perfectly consistent predictions. We then formally present our general algorithm for the complete version of our base model and establish its minimax optimality in \Cref{sec:base model result}.
\vspace{-1mm}
\subsection{Warmup: Single Pool and Perfectly Consistent Predictions}
\label{sec:simple instance}
Focusing on the special case with a single supply pool, we omit the pool index $i$ from our notation. Assuming perfectly consistent predictions (i.e., $\pbias_t = 0$ and $\pprob_t=0$ for all $t\in[T]$), we can, without loss of generality, further assume that (i) the prediction intervals $\{[\pLt,\pRt]\}_{t\in[T]}$ are nested, meaning $[\pL_{t+1},\pR_{t+1}]\subseteq[\pLt,\pRt]$ for each $t\in[0:T-1]$, and (ii) the prediction error upper bound $\perrort$ is weakly decreasing over time and smaller than the initial demand range, i.e., $\perror_t \leq \pRzero - \pLzero$.

As discussed earlier in \Cref{sec:intro}, the novel aspect of our model is the inherent tension between supply availability, which decreases over time, and demand information, which becomes progressively more accurate. This tension is clearly illustrated in our simplified warm-up instance, where the online algorithm faces a fundamental trade-off: On one hand, the algorithm could wait and hire later (e.g., on day $T$, immediately before the operating day) when demand predictions are most accurate, thus reducing the risk of overstaffing but potentially causing understaffing due to limited supply. On the other hand, it could hire earlier (e.g., on day $1$), when supply is abundant, thereby reducing understaffing risks but potentially leading to overstaffing because earlier predictions are less accurate.

Note that if there are enough workers available on day $T$ regardless of demand (i.e., $\wdiscount_T\cdot\supply \geq \pR_0$), the algorithm should simply wait until the last day to hire. In many applications, including dynamic staffing for last-mile delivery discussed in \Cref{sec:intro}, this assumption typically does not hold. Therefore, to avoid high understaffing costs, the platform may need to ``spread'' its staffing decisions over the planning horizon and hire some workers earlier, despite less accurate predictions. We now highlight three key observations that precisely characterize how the minimax-optimal algorithm $\OPTALG$ balances its hiring decisions.

\xhdr{Observation (i):} The first observation is related to the overstaffing cost. Fix an arbitrary online algorithm $\ALG$ and let $\targetcost$ be its cost guarantee (defined in \Cref{def:minmaxcost}). Now for any day $t\in[T]$ and any prediction sequence $\{\prediction_\tau=[\pL_\tau,\pR_\tau]\}_{\tau\in[t]}$ revealed so far, the algorithm's staffing profile $\{\alloc_\tau\}_{\tau\in[t]}$ must satisfy:
\begin{align}
\label{eq:simple instance allocate upper bound}
    \left(\textrm{total \# of hires by the end of time $t$}\right)\equiv\sum\nolimits_{\tau\in[t]}\alloc_\tau\leq \pLt  + \frac{\targetcost}{\overcost}
\end{align}
\revcolor{To see this upper bound, consider an adversary that selects future predictions $\{[L_\tau,R_\tau]\}_{\tau\in[t+1,T]}$ with $L_\tau=\pLt$ for all times $\tau>t$, and eventually selects the final target demand to be $\demand=\pLt$---which is a valid choice due to the perfect consistency assumption of prediction sequences (Assumption~\ref{asp:prediction sequence} with $\pprob_t=\pbias_t=0$ for all $t\in[T]$).} Since staffing decisions are irrevocable, the total hires made by the algorithm can only increase after day $t$, resulting in an overstaffing cost of at least $C\cdot (\sum\nolimits_{\tau\in[t]}\alloc_\tau-\pLt)$. Because $\targetcost$ is the maximum cost across all possible adversarial choices of prediction sequences and demands, the overstaffing cost cannot exceed $\targetcost$. Therefore, any online algorithm must satisfy inequality~\eqref{eq:simple instance allocate upper bound} for every $t\in[T]$. Note also that the converse holds: if an online algorithm satisfies inequality~\eqref{eq:simple instance allocate upper bound} for all $t\in[T]$ given some  $\Gamma$, then its overstaffing cost is at most $\Gamma$.



\xhdr{Observation (ii):}
The second observation is related to the understaffing cost. In simple words, it states that hiring earlier than later is always a (weakly) preferable strategy for the objective of minimizing the understaffing cost. More formally,
consider any two staffing profiles $\xbf= \{\alloc_\tau\}_{\tau\in[T]}$ and $\xbf\primed= \{\alloc_\tau\primed\}_{\tau\in[T]}$ hiring same number of workers, but $\xbf\primed$ hires earlier than $\xbf$, that is, 
\begin{align*}
    \sum\nolimits_{\tau\in[t]}\alloc_\tau\leq \sum\nolimits_{\tau\in[t]}\alloc_\tau\primed
\end{align*}
for every $t\in[T-1]$ and the equality holds when $t = T$. Then, the following hold: (i) if staffing profile $\xbf$ is feasible, then profile $\xbf\primed$ is also feasible, and (ii) for any choice of unknown demand $\demand$, the understaffing cost under profile $\xbf\primed$ is equal to its counterpart under profile $\xbf$. 


Combining the above observations suggests a candidate online algorithm to minimize the understaffing cost (under any adversarial demand $\demand$), while guaranteeing a target upper bound of $\Gamma$ on the overstaffing cost: make staffing decisions greedily by hiring as many workers as possible on each day $t$, subject to the supply feasibility and inequality~\eqref{eq:simple instance allocate upper bound} in Observation~(i). We formalize this greedy-staffing decision in \Cref{alg:opt simple instance}.
\vspace{-1mm}
\begin{algorithm}[htb]
    \SetKwInOut{Input}{input}
    \SetKwInOut{Output}{output}
    \Input{target overstaffing cost $\targetcost$, initial pool size $\supply$, availability rates $\{\wdiscountt\}_{t\in[T]}$}
    \Output{staffing profile $\xbf$}

    \vspace{2mm}
   {\small\color{\commentcolor}\tcc{on day 1, prediction $\prediction_1 = [\pL_1,\pR_1]$ is revealed.}}
    
    \vspace{1mm}
    hire $\alloc_1 = \min\{\pL_1 + \frac{\targetcost}{\overcost}, \wdiscount_1\cdot\supply\}$ available workers.
    
    \vspace{2mm}
    \For{each day $t = 2,\ldots,T$}{
    \vspace{1mm}
    {\small\color{\commentcolor}\tcc{prediction $\predictiont = [\pLt,\pRt]$ is revealed.}}
    \vspace{1mm}
        hire $\displaystyle\alloct = \min\bigg\{\pLt - \pL_{t - 1}~,~\wdiscountt\cdot \Big(\supply - \sum\nolimits_{\tau\in[t - 1]}\frac{\alloc_{\tau}}{\wdiscount_\tau}\Big)\bigg\}$ available workers.
    }
    \caption{\textsc{Greedy-staffing with target overstaffing cost}}
    \label{alg:opt simple instance}
    \vspace{-1mm}
\end{algorithm}
\begin{remark}
\label{remark:overstaffing}
\Cref{alg:opt simple instance} hires $\alloc_1 = \pL_1 + \frac{\Gamma}{\overcost}$ workers on day $1$ (if supply permits), and on subsequent days $t = 2,3,\dots$, it hires $\alloc_t = \pLt - \pL_{t-1}$ workers until supply feasibility becomes binding on some day $t^\dagger$ (at which point all remaining available workers are hired). No further hires occur afterward. By this construction, summing $\alloc_\tau$ from $\tau=1$ to $\tau=t$, the algorithm satisfies \eqref{alg:opt simple instance} with equality for days $t\in[1:t^\dagger-1]$, and with inequality for days $t\in[t^\dagger:T]$. Thus, it ensures an overstaffing cost of at most $\targetcost$. 
\end{remark}
\vspace{-1mm}
\Cref{alg:opt simple instance} nearly identifies the minimax-optimal algorithm; the main remaining challenge is determining the ``correct'' value of $\targetcost$ in inequality~\eqref{eq:simple instance allocate upper bound}, which serves as the input to the algorithm.
We claim that setting $\targetcost = \optcost$ in \Cref{alg:opt simple instance} yields a minimax-optimal algorithm (recall that $\optcost$ is the minimax value of the game, i.e., the cost guarantee of $\OPTALG$). Indeed, the resulting algorithm's overstaffing cost is at most $\optcost$ (by \Cref{remark:overstaffing}). Moreover, thanks to Observations~(i) and~(ii) and the greedy staffing rule of \Cref{alg:opt simple instance}, among all algorithms with an overstaffing cost no greater than $\optcost$ (including the minimax-optimal algorithm satisfying inequality~\eqref{eq:simple instance allocate upper bound} with $\targetcost=\optcost$), \Cref{alg:opt simple instance} achieves the smallest understaffing cost in every instance. Thus, its understaffing cost is at most that of the minimax-optimal algorithm, which cannot exceed $\optcost$. Putting the pieces together, we conclude that our algorithm is also minimax optimal. The only question left is whether we can compute $\optcost$ efficiently, which we address through our third observation below.



\xhdr{Observation (iii):} Focusing on the greedy staffing rule in \Cref{alg:opt simple instance}, it is straightforward to characterize the adversarial prediction sequence that maximizes the worst-case understaffing cost for any given choice of $\targetcost$. Intuitively, the adversary selects predictions that prompt the algorithm to hire workers earlier; in this way, the algorithm runs out of supply earlier, i.e., at a time when the prediction is inaccurate and the intervals are large. Formally, we state the following lemma (with the proof provided in \Cref{apx:maxsupplysequence}).
\begin{restatable}{lemma}{maxsupplysequence}
    \label{lem:max-supply-consuming-sequence}
The understaffing cost of \Cref{alg:opt simple instance} against worst-case demand is maximized when facing the prediction sequence $\overline\predictions=\{\overline{L}_t,\overline{R}_t\}_{t\in[T]}$, defined as $\overline{L}_t\triangleq \pRzero - \perrort,$ and $\overline{R}_t\triangleq \pRzero.$
\end{restatable}
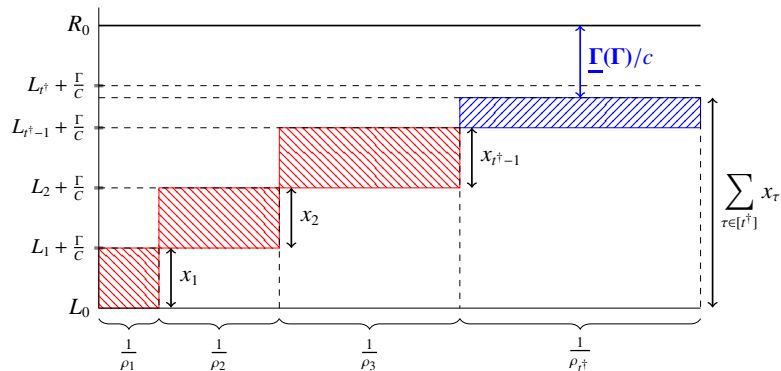
\begin{figure}[h]
    \centering
\input{Paper/figs/fig-fixed-point-argument}
\vspace{1mm}
\caption{
Geometric interpretation of a run of \Cref{alg:opt simple instance} when $\targetcost\in[0,\overcost\cdot(\rho_1\cdot s-L_1)]$. The heights of red rectangles represent hiring decisions before the supply constraint binds. The height of the blue rectangle represents hiring on day $t^\dagger$, the last day of hiring, when all remaining workers are hired. The area of each rectangle corresponds to the (effective) amount of initial supply used on that day; hence, the total area of all rectangles equals the initial supply $s$.}
\label{fig:geometric-proof}
\vspace{-4mm}
\end{figure}
\begin{proof}[Proof sketch of \Cref{lem:max-supply-consuming-sequence}] 
Although the proof of \Cref{lem:max-supply-consuming-sequence} in \Cref{apx:maxsupplysequence} is somewhat subtle, the underlying intuition is simple. On day $t^\dagger-1$, just before supply feasibility becomes binding, the number of workers hired by the algorithm is exactly $\pL_{t^\dagger-1} + \frac{\targetcost}{C}$. Thus, the gap between $\pR_{t^\dagger-1}$ and the number of workers hired equals $\perror_{t^\dagger-1} - \frac{\targetcost}{C}$ (which equals  the maximum understaffing if we ignore day $t^{\dagger}$). Since $\perror_t$ is weakly decreasing in $t$, this understaffing quantity is maximized for the prediction sequence that triggers earlier binding of supply feasibility, which is precisely the sequence $\overline\predictions$. The only subtlety arises from the hiring on the supply binding day $t^\dagger$ (the blue rectangle in \Cref{fig:geometric-proof}), which also influences the understaffing cost. Nevertheless, it turns out that even after accounting for day $t^\dagger$, the worst-case prediction sequence remains the same.
\end{proof}

Based on Observation~(iii), we examine \Cref{alg:opt simple instance} when run with a given target overstaffing cost $\Gamma$ and facing the adversarial prediction sequence $\overline\predictions$. \Cref{fig:geometric-proof} illustrates such a run for a fixed $\Gamma$. Let $\{x_{\tau}\}_{\tau\in[1:t^\dagger]}$ be the sequence of hiring decision in this run, with $t^\dagger$ be the day on which the supply feasibility binds. As depicted, the algorithm makes its final hiring decision on day $t^{\dagger}$ (blue rectangle). At this point, the adversary has two choices to maximize the imbalance cost: either select $d = L_{t^{\dagger}}$, resulting in an overstaffing cost of at most $\Gamma$, or select $d = R_{t^{\dagger}} = R_0$, resulting in an understaffing cost of $c\cdot\left(R_0-\sum_{\tau=1}^{t^\dagger}x_\tau\right)$, denoted by $\underline{\Gamma}$  (see \Cref{fig:geometric-proof}).

Next, we show that $\underline{\Gamma}(\cdot)$ is a (weakly) decreasing function of $\Gamma$, a crucial property that allows efficient computation of $\optcost$. If we increase $\Gamma$ by $\partial\Gamma$, the height of the first red rectangle increases by $\frac{\partial\Gamma}{C}$, while the heights of the remaining red rectangles stay unchanged. Because the total area of all rectangles equals the fixed initial supply $s$, the height of the blue rectangle must decrease by $\frac{\partial\Gamma}{C}\cdot\big(\frac{1/\rho_1}{1/\rho_{t^\dagger}}\big)$. \revcolor{Thus, the total number of hires increases by $\left(\partial\Gamma - \partial\Gamma\big(\frac{1/\rho_1}{1/\rho_{t^\dagger}}\big)\right)/C \geq 0$, since $\rho_1 \geq \rho_{t^\dagger}$. Consequently, we have $\underline{\Gamma}(\Gamma+\partial\Gamma) \leq \underline{\Gamma}(\Gamma)$.}


Given the weak monotonicity of the function $\underline{\Gamma}(\cdot)$, one of these two cases can occur depending on the amount of initial supply pool size $s$:

\begin{enumerate}[label=(\Roman*)]
    \item \emph{Sufficient initial supply:} $\underline{\Gamma}(\cdot)$ has a \emph{fixed point} in $[0,\overcost\cdot(\rho_1\cdot s-\overline{L}_1)]$, that is, there exists $\hat\Gamma\in [0,\overcost\cdot(\rho_1\cdot s-\overline{L}_1)]$ such that $\underline{\Gamma}(\hat\Gamma)=\hat\Gamma$. In this case, we claim that $\optcost=\hat\targetcost$. Suppose by contradiction $\optcost\neq\hat\targetcost$. If $\optcost<\hat\targetcost$, then as $\underline{\Gamma}(\cdot)$ is weakly decreasing we have $\underline{\Gamma}(\optcost)\geq \underline{\Gamma}(\hat\Gamma)=\hat\Gamma>\optcost$, a contradiction, since the understaffing cost of the minimax optimal algorithm (i.e.,  \Cref{alg:opt simple instance} with $\optcost$ as input) cannot exceed $\optcost$. If $\optcost>\hat\targetcost$, then $\underline{\Gamma}(\hat\Gamma)=\hat\Gamma<\optcost$, so both understaffing and overstaffing costs of \Cref{alg:opt simple instance} with $\targetcost$ are strictly smaller than $\optcost$, again a contradiction to the minimax optimality of \Cref{alg:opt simple instance} with $\optcost$ as input.
      \item \emph{Low initial supply:} $\underline{\Gamma}(\cdot)$ has no fixed point in $[0,\overcost\cdot(\rho_1\cdot s-\overline{L}_1)]$, that is, $\underline{\Gamma}(\Gamma')>\Gamma'$ for all $\Gamma'\in [0,\overcost\cdot(\rho_1\cdot s-\overline{L}_1)]$. Since$\underline{\Gamma}(\cdot)$ is weakly decreasing, this case occurs if and only if (see \Cref{fig:geometric-proof}): 
     \begin{align}
     \label{eq:low-supply-fixed-point}
    \underline{\Gamma}\left(\overcost\cdot(\rho_1\cdot s-\overline{L}_1)\right)=\underbrace{\undercost\cdot (\pRzero - \wdiscount_1\cdot\supply)}_{\substack{\textrm{max understaffing cost}\\\textrm{when $x_1=\wdiscount_1\cdot\supply$}}} 
    >
    \underbrace{\overcost\cdot (\wdiscount_1 \cdot \supply - \overline{L}_1)}_{\substack{\textrm{max overstaffing cost} \\\textrm{when $x_1=\wdiscount_1\cdot\supply$}}}
\end{align}
Now we claim that $\optcost = \underline{\Gamma}\left(\overcost\cdot(\rho_1\cdot s-\overline{L}_1)\right)=\undercost\cdot (\pRzero - \wdiscount_1\cdot\supply)$. To see this, first note that we must have $\optcost>\overcost\cdot(\rho_1\cdot s-\overline{L}_1)$, because otherwise $\underline{\Gamma}(\optcost)>\optcost$, contradicting the definition of $\optcost$ (as the maximum understaffing cost of the minimax-optimal algorithm, i.e., \Cref{alg:opt simple instance} with input $\optcost$, cannot exceed $\optcost$). Thus, running \Cref{alg:opt simple instance} with this $\optcost$ hires $x_1=\rho_1\cdot s$ workers on day $1$ and exhausts the supply for any future hiring. Due to inequality~\eqref{eq:low-supply-fixed-point}, the maximum understaffing cost exceeds the maximum overstaffing cost, implying $\optcost=\undercost\cdot(\pRzero - \wdiscount_1\cdot\supply)$.
\end{enumerate}

As demonstrated by the two cases above, $\optcost$ can be computed efficiently in polynomial time: first by checking inequality~\eqref{eq:low-supply-fixed-point}, and then, if needed, finding the fixed point of the weakly decreasing function $\underline{\Gamma}$ (e.g., via binary search). Once $\optcost$ is computed, the minimax-optimal algorithm is simply \Cref{alg:opt simple instance} with input $\targetcost = \optcost$. We summarize this result in the following proposition and provide a formal proof in \Cref{apx:optalgsimpleinstance}.

\begin{restatable}{proposition}{optalgsimpleinstance}
    \label{prop:opt alg simple instance}
    For the special case of the problem with single-pool and perfectly consistent predictions, \Cref{alg:opt simple instance} with $\targetcost=\optcost$ is minimax optimal, where $\optcost=\undercost\cdot (\pRzero - \wdiscount_1\cdot\supply)$ if inequality~\eqref{eq:low-supply-fixed-point} holds, and otherwise it is the fixed point of the function $\underline{\Gamma}$ (see \Cref{fig:geometric-proof} for an illustration). 
\end{restatable}

To further show case how the minimax optimal algorithm (\Cref{alg:opt simple instance} with $\targetcost = \optcost$ as in \Cref{prop:opt alg simple instance}) handles the tradeoff mentioned earlier, we use three simple numerical examples (\Cref{fig:tradeoff illustration example}):


\begin{itemize}
\item An instance in which the initial supply $\supply$ is sufficiently large and the demand is fully revealed on the last day $T$ (i.e., $\pR_T - \pL_T=0$). By hiring exactly enough workers to match the lower bound $\pLt$ of demand $\demand$ on each day $t < T$, our algorithm perfectly matches the demand revealed on day $T$, incurring zero staffing cost. In contrast, any algorithm that waits until day $T$ to hire may suffer a strictly positive understaffing cost, as the number of available workers could be fewer than the revealed demand. See \Cref{fig:example large supply}.
\item An instance in which the initial supply $\supply$ and availability rates $\{\wdiscountt\}$ are sufficiently small. In this scenario, the algorithm hires workers primarily in the early days and exhausts the available supply. See \Cref{fig:example small supply}.
\item An instance in which predictions $\predictiont = [\pLt,\pRt]$ remain uninformative until the last day $T$ (i.e., $\perror_T < \perror_{T - 1} = \dots = \perror_1 = \pRzero - \pLzero$). The algorithm then hires workers only on days $1$ and $T$, making no hires from day $2$ to day $T - 1$ since demand predictions do not improve during this period. See \Cref{fig:example uninformative prediction}.
\end{itemize}

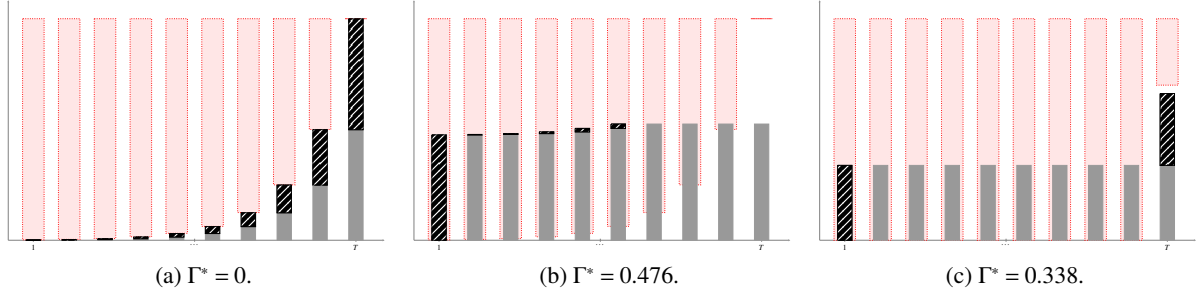
\begin{figure}
    \centering
    \subfloat[$\optcost = 0$.]{
    \input{Paper/figs/fig-example-large-supply}
        \label{fig:example large supply}
    }
    \subfloat[$\optcost = 0.476$.]{
    \input{Paper/figs/fig-example-small-supply}
        \label{fig:example small supply}
    }
    \subfloat[$\optcost = 0.338$.]{
    \input{Paper/figs/fig-example-uninformative-prediction}
        \label{fig:example uninformative prediction}
    }

    \caption{Graphical illustration of the minimax optimal algorithm facing the  prediction sequence $\predictions$ in example instances of the single-pool problem (with perfectly consistent predictions). We set $T= 10$, $\undercost = \overcost = 1$, $\rho_t = 1 - 0.8^{T - t + 1}$ for all  $t$, $\pLzero = 0,\pRzero = 1$. In (a), (b), $\perror_t = 1 - 0.5^{T - t}$ for all $t\in[T]$; while in (c) $\perror_t = 1$ for all $t\in[T - 1]$ and $\perror_T = 0.3$. Supply $\supply = 4, 0.6, 2$ in (a), (b), (c), respectively.
    The red dotted bar is the prediction revealed in $\predictions$. The black (gray) bar is the number of workers hired on (before) that day. The optimal minimax cost $\optcost$ is reported in each subfigure.}
    \label{fig:tradeoff illustration example}
    \vspace{-2mm}
\end{figure}
\revcolor{
\begin{remark}
    The fixed-point characterization of $\Gamma^*$ in \Cref{prop:opt alg simple instance} could possibly be refined under certain model primitives that give a more explicit form for the function $\underline{\Gamma}$. We exemplify this point in \Cref{apx:refined-gamma-single-pool}.
\end{remark}
}

We conclude this subsection by highlighting a key insight from our analysis. Although the adversary initially had infinitely many possible prediction sequences, our main idea in \Cref{prop:opt alg simple instance} was to show that restricting the adversary to a smaller subset of prediction sequences is without loss. Under this restriction, the algorithm design reduces to solving a single-dimensional fixed-point calculation. In \Cref{sec:base model result}, we generalize this idea by identifying a carefully chosen subset of prediction sequences that simplifies the adversary's maximization task. We then show how, for general instances, the analogous fixed-point calculation can be formulated and solved via a specific \emph{linear program}.


\vspace{-1mm}
\subsection{Optimal Staffing via LP-based Emulator}
\label{sec:base model result}

We now turn our attention to the general case with multiple pools and approximately consistent predictions, and show how the ideas in \Cref{sec:simple instance} can be extended to design and analyze the minimax optimal algorithm for these general instances. 

\revcolor{In contrast to our warm-up setting in \Cref{sec:simple instance} with perfectly consistent predictions, we make no assumptions on prediction inconsistency upper bounds $\{\pbias_t\}_{t\in[T]}$ of Assumption~\ref{asp:prediction sequence} in this section. Regarding the probabilities of miscoverage $\{\delta_t\}$ in Assumption~\ref{asp:prediction sequence}, as mentioned earlier in \Cref{sec:model justificaiton}, we generally expect these probabilities to be quite small, say $\mathcal{O}(\frac{1}{T^\gamma})$ for large enough $\gamma>0$\footnote{Here, to quantify the error, we think of asymptotic behavior of our additive errors with respect to $T$ while considering other parameters fixed. This should not be confused with our large market assumption.} .Given this, we could simply apply the union bound over all days $t\in[T]$ to bound the total miscoverage probability, and reduce this setting to the case with $\delta_t=0$ by introducing a small extra additive error in the expected imbalance cost in the order of $\mathcal{O}\left(\sum_{t\in[T]}\delta_t\right)=o(1)$ when $\gamma>1$. Another relevant practical scenario is when miscoverage of the prediction interval on a day---e.g., due to temporary anomalies or shocks on that day in the underlying forecasting
method---is possibly high-probability but \emph{detectable} by the dynamic staffing algorithm. In that case, we can show that by slight modifications of our algorithms the problem can be reduced again to the case with $\delta_t=0$, but this time with a smaller additive error of $\mathcal{O}\left(\max_{t\in[T]}\delta_t\right)$; therefore, in some sense, the error due to miscoverage on a day does not propagate. We defer the technical details to \Cref{apx:prob-miscoverage-shocks}. In light of these reductions and also for the brevity of exposition, we assume $\delta_t=0$ for all $t\in[T]$ in the rest of this section.}


We begin by highlighting two major new technical challenges compared to the simplified instances in our warm-up setting in \Cref{sec:simple instance}. First, with multiple supply pools having heterogeneous sizes $\{\supplyi\}$ and availability rates $\{\wdiscountit\}_{t\in[T]}$, the platform must determine not only the total number of hires each day, but also the specific allocations across pools. Second, without perfectly consistent predictions, we cannot assume that prediction intervals are nested or that the prediction error upper bound is weakly decreasing. Given these challenges, it appears unlikely that a simple procedure similar to \Cref{alg:opt simple instance}, which takes only the optimal minimax cost $\optcost$ as input and guarantees this cost under any prediction sequence, would exist. Furthermore, computing the optimal minimax cost $\optcost$ itself becomes more involved.


\xhdr{High-level sketch of our approach.}
While the algorithmic results of \Cref{sec:simple instance} do not directly extend, we can still leverage their intuition and follow a similar high-level approach to design our general algorithm. Specifically, we first identify a small subset of prediction sequences $\{\predictions\ked\}$, enabling us to characterize the optimal minimax cost $\optcost$ via a linear program~\ref{eq:opt reduced form}. We show that the feasible set of this linear program encompasses certain non-adaptive staffing strategies---called \emph{canonical staffing profiles}---which are candidate optimal solutions against prediction sequences in the set $\{\predictions\ked\}$. Next, we introduce Procedure~\ref{alg:emulator}, which we refer to as our \emph{emulator oracle}. This powerful subroutine allows the algorithm to take any canonical staffing profile as input, and adaptively generate staffing decisions for arbitrary prediction sequences, guaranteeing that its cost does not exceed the worst-case cost of the canonical staffing profile over prediction sequences in $\{\predictions\ked\}$. Finally, by invoking Procedure~\ref{alg:emulator} with the optimal solution of program~\ref{eq:opt reduced form}, we obtain our minimax-optimal algorithm. We now describe this approach in detail.


As mentioned earlier, the adversary has infinitely many possibilities for its prediction sequence. Now consider restricting the adversary to the following subset of $T$ \emph{``single-switch''} prediction sequences $\{\predictions\ked\}_{k\in[T]}$, where $\predictions\ked = \{[\pLt\ked,\pRt\ked]\}_{t\in[T]}$ is defined as follows:
\begin{align}
\label{eq:single switch prediction construction}
\begin{array}{lll}
    t\in[k]:&
    \qquad
    \pLt\ked \gets \pRzero - \pbias_t - \perrort~, &
    \pRt\ked \gets \pRzero - \pbias_t~,
    \\
    t\in[k+1:T]:&
    \qquad
    \pLt\ked \gets \max\limits_{\tau\in[0:k]}{\pRzero - \perror_{\tau} - 2\pbias_{\tau}}~,&
    \pRt\ked \gets \left(\max\limits_{\tau\in[0:k]}{\pRzero - \perror_{\tau} - 2\pbias_{\tau}}\right) + \perrort~.
\end{array}
\end{align}
Also, let $\pLzero\ked = \pLzero$ and $\pRzero\ked = \pRzero$. Here, the prediction sequence $\predictions\ked$ represents an adversary who \revcolor{(i) always picks prediction intervals of exact length $\perrort$ on each day $t$ to minimally satisfy the prediction error bounds,} and (ii) initially signals high demand during days $t\in[1:k]$ by setting the right endpoint of each prediction interval to $\pRzero - \pbias_t$ \revcolor{(and therefore, the left endpoint to $\pRzero - \pbias_t - \perrort$),} and then switches to signaling low demand during days $[k+1:T]$ by keeping the left endpoint fixed at $\max_{\tau\in[0:k]}{\pRzero - \perror_{\tau} - 2\pbias_{\tau}}$. We note that the prediction sequence $\overline\predictions$ from \Cref{lem:max-supply-consuming-sequence} coincides with $\predictions^{(T)}$ when $\pbiass = \zerobf$. This allows us to characterize the worst-case understaffing cost as in \Cref{sec:simple instance}. However, unlike our simpler approach in \Cref{sec:simple instance}, here we allow the adversary to choose from the larger set $\{\predictions\ked\}_{k\in[T]}$ to also identify the worst-case overstaffing cost.

Restricting the adversary to be single-switch (as described above in \ref{eq:single switch prediction construction}), it is straightforward to characterize the online algorithm with the minimum cost guarantee against this restricted adversary. The  prediction sequences $\{\predictions\ked\}_{k\in[T]}$ are designed such that, on each day $t$, the prediction intervals $\{[\pL_\tau\ked,\pR_\tau\ked]\}_{\tau\in[t]}$ are identical across all sequences $\predictions\ked$ with $k \geq t$. Hence, informally, no online algorithm can guess the exact day when the adversary will switch in future, if it has not yet switched by day $t$. Formally, the staffing decision at day $t$ must therefore be the same for all prediction sequences $\predictions\ked$ with $k \geq t$. This motivates introducing the following linear program~\ref{eq:opt reduced form}, with decision variables $\{\allocit,\targetcost\}_{i\in[n],t\in[T]}$, to encode the staffing decisions of the minimax-optimal algorithm and its cost guarantee under this restricted adversary:
\begin{align}
\tag{$\textsc{LP-single-switch}$}
    \label{eq:opt reduced form}
    &\arraycolsep=1.4pt\def\arraystretch{1.2}
    \begin{array}{llll}
    \min\limits_{\substack{\xbf,\targetcost\geq \zerobf}}\quad\quad
    &
    \targetcost
    & 
    \text{s.t.}
    \\
    &
    \displaystyle\sum\nolimits_{t\in[T]}\frac{1}{\wdiscountit}\allocit \leq \supplyi
    &
    i\in[n]
    \\
    &
    \displaystyle\sum\nolimits_{i\in[n]}\sum\nolimits_{t\in[k]}
    \allocit \leq \max_{\tau\in[0:k]}\left({\pRzero - \perror_{\tau} - 2\pbias_{\tau}}\right) + \frac{\targetcost}{\overcost}
    \quad\quad
    &
    k\in[T]
    \\
    &
    \displaystyle\sum\nolimits_{i\in[n]}\sum\nolimits_{t\in[T]}
    \allocit \geq \pRzero - \frac{\targetcost}{\undercost}
    &
\end{array}
\end{align}

Several comments about this LP are in order. First, the variable $\allocit$ in program~\ref{eq:opt reduced form} represents the fractional number of workers hired from pool $i$ on day $t$ under any prediction sequence $\predictions\ked$ with $k\geq t$ (no workers are hired after day $k$, when the adversary signals low demand). Second, the terms $\targetcost/\overcost$ and $\targetcost/\undercost$ in the program represent the potential numbers of overstaffing and understaffing, respectively. In fact, beyond the supply feasibility constraints, the other two constraints ensure that both the overstaffing and understaffing costs remain bounded by $\targetcost$, regardless of the adversary’s chosen switching time. Finally, since we have restricted the adversary, the optimal objective value of this LP provides a lower bound on the optimal minimax cost $\optcost$. This is formally established in the following lemma (see \Cref{apx:reduced-form-lower} for the proof). As we will see later, this lower bound is actually tight: the LP's optimal value equals $\optcost$.

\begin{restatable}{lemma}{lemmalower}
\label{lem:opt reduced form lower bound optimal minimax cost}
    For any (possibly randomized) online algorithm $\ALG$, the cost guarantee is at least the optimal objective value of program~\ref{eq:opt reduced form}.
\end{restatable}

As alluded to earlier, we refer to any feasible solution $\canxbf = \{\canallocit\}$ of program~\ref{eq:opt reduced form} as a ``canonical staffing profile.'' Each canonical profile $\canxbf$ encodes staffing profiles $\xbf\ked = \{\allocit\ked\}$, defined as $\allocit\ked \triangleq \canallocit\cdot \indicator{k \geq t}$ for each prediction sequence $\predictions\ked$. A natural follow-up question is: \emph{Can we use the canonical staffing profile $\canxbf$ to construct a staffing profile in an online fashion for general prediction sequences, achieving performance (in terms of imbalance cost) at least as good as $\xbf\ked$ under the prediction sequences $\predictions\ked$, for all $k\in[T]$?}

\begin{myprocedure}
    \SetKwInOut{Input}{input}
    \SetKwInOut{Output}{output}
    \Input{canonical staffing profile $\tilde\xbf=\{\canallocit\}_{i\in[n],t\in[T]}$,  initial demand range $[\pLzero,\pRzero]$, prediction inconsistency upper bounds  $\{\pbiast\}_{t\in[T]}$}
    \Output{staffing profile $\xbf=\{\allocit\}_{i\in[n],t\in[T]}$}
    
    \For{each day $t\in[T]$}{
    {\small\color{\commentcolor}\tcc{prediction $[\pLt,\pRt]$ reveals}}
        solve for an arbitrary staffing profile $\{\allocit\}_{i\in[n]}$ satisfying:
        \begin{align*}
        \displaystyle
        \sum\nolimits_{i\in[n]}
            \allocit = 
            \plus{
            \sum\nolimits_{\tau\in[t]}\sum\nolimits_{i\in[n]}
            \canalloc_{i\tau}
            -
            \sum\nolimits_{\tau\in[t-1]}\sum\nolimits_{i\in[n]}
            \alloc_{i\tau}
            -
            \left(\pRzero - \left(\min_{\tau\in [0:t]}\pR_{\tau}+\pbias_\tau\right)\right)
            }~,
        \end{align*}
        and $0 \leq \allocit \leq \canallocit$ for all $i\in[n]$ 
        
        \vspace{1mm}
        
         {\small\color{\commentcolor}\tcc{see the feasibility of the resulting staffing profile in \Cref{lem:emulator}}}
        \vspace{1mm}
        hire $\allocit$ available workers from each pool $i\in[n]$
    }
    \caption{\textsc{Emulator Oracle}}
    \label{alg:emulator}
\end{myprocedure}
We answer this question affirmatively by introducing Procedure~\ref{alg:emulator}, referred to as the ``emulator oracle,'' which serves as the core building block for the minimax-optimal algorithms in our base model and extensions (\Cref{sec:extension} and \Cref{apx:combinded objective}). 
\revcolor{Before stating crucial properties of the staffing profile obtained by Procedure~\ref{alg:emulator}, let us unpack this procedure by starting simple: Suppose $\pbias_t = 0$, $t \in [T]$ and $\pRt=\pRzero$, $t \in [T]$, i.e., the upper bound on demand never improves; then Procedure~\ref{alg:emulator} can simply follow the canonical staffing profile $\tilde\xbf$ throughout the horizon, i.e., outputs $\allocit = \canallocit$ on each day $t$. Now imagine $\pRt = \pRzero$ for $t \in [T-1]$ but $\pR_T < \pRzero$, ruling out the prediction sequence of last day to be $\predictions^{(T)}$;  then up to time $T-1$, the emulator oracle (Procedure~\ref{alg:emulator}) makes exactly the same staffing decision as the canonical profile; however, not for the last day. For this day,  Procedure~\ref{alg:emulator} hires \emph{less} workers compared to the canonical staffing profile, as the prediction interval suggests a lower target demand than $\predictions^{(T)}$ by ruling out having target demand in the interval $(R_T, R_0]$.}

The above observation holds more generally. At a high level, Procedure~\ref{alg:emulator} copies the canonical staffing profile $\tilde\xbf$---which protects against single-switch prediction sequences $\{\predictions\ked\}_{k\in[T]}$---up to the first day $t'$ when the observed prediction interval deviates from this pattern. On this day, the procedure adjusts its staffing decision by subtracting the deviation from the canonical profile's suggested number of hires. After this adjustment, it continues mimicking the canonical profile, applying similar modifications iteratively to subsequent decisions suggested by the canonical profile based on observed prediction intervals. Procedure~\ref{alg:emulator}'s theoretical guarantees are formally stated in \Cref{lem:emulator}.



\begin{lemma}
    \label{lem:emulator}
    For any canonical staffing profile $\canxbf$, there always exists a solution $\{\allocit\}_{i\in[n]}$ for step~2 of Procedure~\ref{alg:emulator} that is supply feasible. Furthermore, the staffing profile $\xbf$ returned by Procedure~\ref{alg:emulator} satisfies the following:
    \begin{enumerate}
        \item \underline{Bounded overstaffing cost:} $\sum\nolimits_{i\in[n]}\sum\nolimits_{t\in[T]}\allocit - (\max_{\tau\in[0:T]} \pL_{\tau} - \pbias_{\tau}) \leq 
        \max\limits_{k\in[T]}
        \sum\nolimits_{i\in[n]}\sum\nolimits_{t\in[k]}\canallocit - 
        \pL_T\ked $~,
        \item \underline{Bounded understaffing cost:}
        $(\min_{\tau\in[0:T]} \pR_\tau + \pbias_\tau) - \sum\nolimits_{i\in[n]}\sum\nolimits_{t\in[T]}\allocit \leq 
        \pR_T\Ted - \sum\nolimits_{i\in[n]}\sum\nolimits_{t\in[T]}\canallocit$~,
    \end{enumerate}
    where $\{\pL_T\ked\}_{k\in[T]}$ and $\pR_T\Ted$ are constructed as described in ~\eqref{eq:single switch prediction construction}.
\end{lemma}

\begin{proof}
To simplify the presentation, we introduce auxiliary notation defined as
\begin{align*}
\hat{\pR}_t \triangleq \min_{\tau\in[0:t]}(\pR_\tau + \pbias_\tau)
\quad\text{and}\quad
\hat{\pL}_t \triangleq \max_{\tau\in[0:t]}(\pL_\tau - \pbias_\tau)
\quad\text{for all } t\in[T].
\end{align*}
Note that, due to the $\pbiass$-consistency condition in \Cref{asp:prediction sequence}, the interval $[\hat\pL_t,\hat\pR_t]$ captures the range of the unknown demand $\demand$ given all predictions from days $[0:t]$. By construction, $\hat{\pR}_t$ is weakly decreasing, and $\hat{\pL}_t$ is weakly increasing over time. Additionally, from the $\perrors$-bounded error condition in \Cref{asp:prediction sequence}, we have $\hat{\pR}_t - \hat{\pL}_t \leq \min_{\tau\in[0:t]}(\perror_\tau + 2\pbias_\tau)$ for every day $t\in[T]$.

We first show that there exists a solution for step~2 of Procedure~\ref{alg:emulator} on every day $t\in[T]$ that is supply feasible. It suffices to show that for each day $t\in[T]$,
    \begin{align}
    \label{eq:alg feasibility ineq}
            \plus{
            \sum\nolimits_{\tau\in[t]}\sum\nolimits_{i\in[n]}
            \canalloc_{i\tau}
            -
            \sum\nolimits_{\tau\in[t-1]}\sum\nolimits_{i\in[n]}
            \alloc_{i\tau}
            -
            \left(\pRzero - \pRtHat\right)
            }
            \leq 
            \sum\nolimits_{i\in[n]}\canallocit~.
    \end{align}
    We prove the above inequality by induction on $t$.

\xhdr{Base case ($\boldsymbol{t = 0}$ or $\boldsymbol{t = 1}$):}
When $t = 0$, inequality~\eqref{eq:alg feasibility ineq} holds trivially. Now suppose $t = 1$. In this case, the left-hand side of~\eqref{eq:alg feasibility ineq} equals $\left(\sum_{i\in[n]} \canalloc_{i1} - (\pRzero - \hat{\pR}_{1})\right)^{+}$, which is clearly less than or equal to the right-hand side $\sum_{i\in[n]} \canalloc_{i1}$, since $\hat{\pR}_{1} \leq \pRzero$ by definition (recall $\pbias_0 = 0$).


    \xhdr{Inductive step for $\boldsymbol{t \geq 2}$:}
    Suppose inequality~\eqref{eq:alg feasibility ineq} holds for all $\tau \in[0:t - 1]$. 
    Note that the inequality for $t$ holds trivially if its left-hand side is zero.  
    Thus, we only consider $\sum_{\tau\in[t]}\sum_{i\in[n]}
    \canalloc_{i\tau}
    -
    \sum_{\tau\in[t-1]}\sum_{i\in[n]}
    \alloc_{i\tau}
    -
    \left(\pRzero - \pRtHat\right) > 0$.
    
    Let $k$ be the largest index from $[0: t - 1]$ such that 
    \begin{align*}
    \sum\nolimits_{\tau\in[k]}\sum\nolimits_{i\in[n]}
    \canalloc_{i\tau}
    -
    \sum\nolimits_{\tau\in[k-1]}\sum\nolimits_{i\in[n]}
    \alloc_{i\tau}
    -
    \left(\pRzero - \hat{\pR}_{k}\right) \geq 0
    \end{align*}
    Note the existence of day $k$ is ensured since this requirement holds trivially for $k = 0$.
    Now we upper bound the left-hand side of inequality~\eqref{eq:alg feasibility ineq} for day $t$ as follows,
    \begin{align*}
        &\plus{
        \sum\nolimits_{\tau\in[t]}\sum\nolimits_{i\in[n]}
        \canalloc_{i\tau}
        -
        \sum\nolimits_{\tau\in[t-1]}\sum\nolimits_{i\in[n]}
        \alloc_{i\tau}
        -
        \left(\pRzero - \pRtHat\right)
        }
        \\
        \overset{(a)}{=}&
        \sum\nolimits_{\tau\in[t]}\sum\nolimits_{i\in[n]}
        \canalloc_{i\tau}
        -
        \sum\nolimits_{\tau\in[t-1]}\sum\nolimits_{i\in[n]}
        \alloc_{i\tau}
        -
        \left(\pRzero - \pRtHat\right)
        \\
        \overset{}{=}
        &
        \sum\nolimits_{i\in[n]}\canallocit
        +
        \sum\nolimits_{\tau \in[k]}\sum\nolimits_{i\in[n]}
        (\canalloc_{i\tau} - \alloc_{i\tau})
        +
        \sum\nolimits_{\tau \in[k+1:t-1]}\sum\nolimits_{i\in[n]}
        (\canalloc_{i\tau} - \alloc_{i\tau})
        -
        \left(\pRzero - \pRtHat\right)
        \\
        \overset{(b)}{=}
        & 
        \sum\nolimits_{i\in[n]}\canallocit
        +
        (\pRzero - \hat{\pR}_{k})
        +
        \sum\nolimits_{\tau\in[k + 1:t - 1]}\sum\nolimits_{i\in[n]}\canalloc_{i\tau}
        -
        \left(\pRzero - \pRtHat\right)
        \\
        \overset{(c)}{\leq}
        & 
        \sum\nolimits_{i\in[n]}\canallocit
        +
        (\pRzero - \hat{\pR}_{k})
        +
        (\hat{\pR}_{k} - \hat{\pR}_{t-1}) 
        -
        \left(\pRzero - \pRtHat\right)
        \overset{}{=}
        \sum\nolimits_{i\in[n]}\canallocit
        +
        \pRtHat - \hat{\pR}_{t-1}
        \overset{(d)}{\leq}
        \sum\nolimits_{i\in[n]}\canallocit
    \end{align*}
    and thus inequality~\eqref{eq:alg feasibility ineq} holds for $t$. Here, equality~(a) holds due to the assumption in the beginning of the inductive step;
    equality~(b) holds due to the definition of index $k$ as well as the induction hypothesis, which implies for every $\tau\in[k + 1: t - 1]$, $\sum_{i\in[n]}\alloc_{i\tau} = 0$,
    and $\sum_{\tau \in[k]}\sum_{i\in[n]}
    (\canalloc_{i\tau} - \alloc_{i\tau}) = \pRzero - \hat{\pR}_{k}$ due to line~(2) of Procedure~\ref{alg:emulator};
    inequality~(c) holds due to the definition of index $k$, which implies 
    $\sum_{\tau\in[k + 1:t - 1]}\canalloc_{i\tau} \leq \hat{\pR}_{k} - \hat{\pR}_{t-1}$;
    and 
    inequality~(d) holds due to the definition of $\pRtHat$ and $\hat{\pR}_{t-1}$ (which guarantees $\pRtHat \leq \hat{\pR}_{t-1}$).

    Next, we show the \underline{bounded overstaffing cost} property in the lemma statement. Let $k$ be the largest index such that $\sum_{i\in[n]}\alloc_{ik} > 0$.
    If such $k$ does not exist, then the property holds trivially.
    Note that 
    \begin{align*}
    &{}~~~~\sum\nolimits_{i\in[n]}
    \sum\nolimits_{t\in[T]}
    \allocit 
    -
    \left(\max_{\tau\in[0:T]}\left(\pL_{\tau} - \pbias_\tau\right)\right)
    \overset{(a)}{=}
    \sum\nolimits_{i\in[n]}
    \sum\nolimits_{t\in[T]}
    \allocit 
    -
    \hat\pL_{T}
    \overset{(b)}{=}
    \sum\nolimits_{i\in[n]}
    \sum\nolimits_{t\in[k]}
    \allocit 
    -
    \hat\pL_{T}
    \\&\overset{(c)}{=}
    \sum\nolimits_{i\in[n]}
    \sum\nolimits_{t\in[k]}
    \canallocit
    -
    (\pRzero - \hat\pR_{k})
    -
    \hat\pL_{T}
    \overset{(d)}{\leq}
    \sum\nolimits_{i\in[n]}
    \sum\nolimits_{t\in[k]}
    \canallocit
    -
    (\pRzero - \hat\pR_{k})
    -
    \hat\pL_{k}
    \\&\overset{(e)}{\leq}
    \sum\nolimits_{i\in[n]}
    \sum\nolimits_{t\in[k]}
    \canallocit
    -
    \pRzero 
    +
    \left(\min_{\tau\in[0:k]}
    \left(\perror_\tau
    +
    2\pbias_\tau\right)
    \right)
    \overset{(f)}{=}
    \sum\nolimits_{i\in[n]}
    \sum\nolimits_{t\in[k]}
    \canallocit
    -
    \pL_T\ked
    \end{align*}
    as desired. 
    Here equality~(a) holds due to the definition of $\hat{\pL}_T$; equalities~(b) and (c) hold due to the definition of index $k$ and line~(2) of Procedure~\ref{alg:emulator};
    inequality~(d) holds since $\hat\pL_{T} \geq \hat\pL_{k}$ by definition;
    inequality~(e) holds since $\hat\pR_{k} - \hat\pL_{k}\leq \perrork + 2\pbias_k$ by definition;
    and equality~(f) holds due to the construction of $\pL_T\ked$ in \eqref{eq:single switch prediction construction}.
    
    Finally, we show the \underline{bounded understaffing cost} property in the lemma statement.
    First, due to the construction of staffing profile $\{\allocit\}$ in line~(2) of Procedure~\ref{alg:emulator}, we have 
\begin{align*}
    \sum\nolimits_{i\in[n]}\alloc_{iT} 
    &=
    \plus{\sum\nolimits_{t\in[T]}\sum\nolimits_{i\in[n]}\canallocit - 
    \sum\nolimits_{t\in[T - 1]}\sum\nolimits_{i\in[n]}\allocit
    -(\pRzero - \hat{\pR}_{T})}
    \\
    &\geq 
    \sum\nolimits_{t\in[T]}\sum\nolimits_{i\in[n]}\canallocit - 
    \sum\nolimits_{t\in[T - 1]}\sum\nolimits_{i\in[n]}\allocit
    - (\pRzero - \hat{\pR}_{T})
\end{align*}
and thus, after rearranging terms, we have 
\begin{align*}
    \left(\min_{\tau\in[0:T]} \left(\pR_\tau + \pbias_\tau\right)\right)
    -
    \sum\nolimits_{t\in[T]}\sum\nolimits_{i\in[n]}\allocit
    &\overset{(a)}{=}
    \hat{\pR}_{T} - 
    \sum\nolimits_{t\in[T]}\sum\nolimits_{i\in[n]}\allocit
    \\
    &\leq 
    \hat{\pR}_{T}
    -
    \left(
    \sum\nolimits_{t\in[T]}\sum\nolimits_{i\in[n]}\canallocit
    -
    (\pRzero - \hat{\pR}_{T})
    \right)
    \\
    & = 
    \pRzero
    -
    \sum\nolimits_{t\in[T]}\sum\nolimits_{i\in[n]}\canallocit
    \overset{(b)}{=}
    \pR_T\Ted - \sum\nolimits_{i\in[n]}\sum\nolimits_{t\in[T]}\canallocit~,
\end{align*}
as desired. Here, equality~(a) holds by definition of $\hat{\pR}_T$, and (b) holds by construction of $\pR_T\Ted$ in \eqref{eq:single switch prediction construction}.
\end{proof}

Putting all the pieces together, in particular \Cref{lem:opt reduced form lower bound optimal minimax cost} and \Cref{lem:emulator}, we are now ready to present the minimax optimal algorithm (\Cref{alg:opt}) with its optimality guarantee (\Cref{thm:opt alg}).
\begin{algorithm}[htb]
 \setcounter{AlgoLine}{0}
    \SetKwInOut{Input}{input}
    \SetKwInOut{Output}{output}
    \Input{initial pool sizes $\{\supplyi\}_{i\in[n]}$, availability rates $\{\wdiscountit\}_{i\in[n], t\in[T]}$, initial demand range $[\pLzero,\pRzero]$, prediction error upper bounds $\{\perrort\}_{t\in[T]}$, prediction inconsistency upper bounds  $\{\pbiast\}_{t\in[T]}$}
    \Output{staffing profile $\xbf$}
    \vspace{2mm}
    find an optimal solution $(\xbf^*, \targetcost^*)$ of program~\ref{eq:opt reduced form}
    
    \vspace{2mm}
    invoke Procedure~\ref{alg:emulator} with canonical staffing profile $\canxbf \gets \xbf^*$ {\small\color{\commentcolor}\tcc{facing prediction sequence $\predictions$}}
    \caption{\OPTSim}
    \label{alg:opt}
\end{algorithm}
\begin{theorem}
\label{thm:opt alg}
    {\OPTSim} is minimax optimal. Furthermore, its optimal minimax cost $\optcost$ is equal to the objective value of program~\ref{eq:opt reduced form}.
\end{theorem}
A few remarks about this result are in order. First, {\OPTSim} has a running time of $\texttt{Poly}(n, T)$. Second, supply feasibility of the staffing profile returned by {\OPTSim} is guaranteed by the supply feasibility of the optimal solution to \ref{eq:opt reduced form} and the fact that $x_{it}\leq\tilde{x}_{it}={x}^*_{it}$ by construction in Procedure~\ref{alg:emulator}. Third, randomization cannot improve the cost guarantee, as the minimax-optimal algorithm ({\OPTSim}) is deterministic. Lastly, because the minimax-optimal algorithm is deterministic, it remains minimax optimal against both oblivious and adaptive adversaries.

\subsubsection{Proof of Theorem~\ref{thm:opt alg}}
\label{sec:proof-main}
The proof consists of two steps. In the first step, using \Cref{lem:opt reduced form lower bound optimal minimax cost}, we conclude that the cost guarantee of every online algorithm is at least the optimal objective value of program~\ref{eq:opt reduced form}. In the second step, we show \Cref{lem:opt alg minimax cost} (by using \Cref{lem:emulator}), stating that the cost guarantee of {\OPTSim} is at most the optimal objective value of program~\ref{eq:opt reduced form}. Combining these two steps completes the proof of \Cref{thm:opt alg}. 
\begin{lemma}
\label{lem:opt alg minimax cost}
    The cost guarantee of {\OPTSim} is at most the optimal objective value of program~\ref{eq:opt reduced form}.
\end{lemma}
\begin{proof}
    Let $(\xbf^*, \targetcost^*)$ be the optimal solution of program~\ref{eq:opt reduced form} used in algorithm \OPTSim.
    It suffices to show that for every prediction sequence $\predictions$ and demand $\demand$, the total number of hired workers $\sum_{i\in[n]}\sum_{t\in[T]}\allocit$ satisfies the following two inequalities:
\begin{align*}
\undercost\cdot \plus{\demand - \sum\nolimits_{i\in[n]}\sum\nolimits_{t\in[T]}\allocit} \leq {\targetcost}^*
\;\;
\mbox{and}
\;\;
\overcost\cdot \plus{\sum\nolimits_{i\in[n]}\sum\nolimits_{t\in[T]}\allocit - \demand} \leq \targetcost^*~,
\end{align*}
Recall that due to the $\pbiass$-consistency in \Cref{asp:prediction sequence}, demand $\demand$ satisfies
\begin{align*}
    \hat{L}_t=\max_{t\in[0:T]} \left(\pLt - \pbias_t\right)
    \leq 
    \demand
    \leq 
    \min_{t\in[0:T]} \left(\pRt + \pbias_t\right)=\hat{R}_t
\end{align*}
Hence, it suffices to prove
\begin{align*}
\undercost\cdot \plus{ 
\left(\min_{t\in[0:T]} \left(\pRt + \pbias_t\right)\right) - \sum\nolimits_{i\in[n]}\sum\nolimits_{t\in[T]}\allocit} \leq \targetcost^*
\;\;
\mbox{and}
\;\;
\plus{\sum\nolimits_{i\in[n]}\sum\nolimits_{t\in[T]}\allocit - \left(\max_{t\in[0:T]} \left(\pLt - \pbias_t\right)\right)} \leq \targetcost^*~.
\end{align*}
Invoking the bounded over/understaffing cost properties in \Cref{lem:emulator}, along with second and third constraints on $\targetcost^*$ in program~\ref{eq:opt reduced form}, proves the two inequalities above as desired.
\end{proof}
\vspace{-2mm}
\revcolor{\subsection{LP Resolving and Minimax Optimality}
\label{sec:lp-resolving-main}

\input{Paper/lp-resolve}
}

%% file: Paper/figs/fig-fixed-point-argument.tex
\begin{tikzpicture}[scale = 0.8]
\begin{axis}[
    axis lines=middle,
    axis line style={draw=none},
    xtick=\empty,
    ytick={1, 2, 3, 3.7},
    yticklabels={$L_1+\frac{\Gamma}{C}$, $L_2+\frac{\Gamma}{C}$, $L_{t^\dagger-1}+\frac{\Gamma}{C}$, $L_{t^\dagger}+\frac{\Gamma}{C}$},
    ymin=0, ymax=5,
    xmin=0, xmax=10,
    xlabel near ticks,
    ylabel near ticks,
    clip=false,
    tick label style={font=\small},
    label style={font=\small},
    scale only axis,
    width=10cm,
    height=5cm,
    every axis plot/.append style={thick},
    tick style={ultra thick} 
]

\pgfplotsset{
    after end axis/.code={
        \draw[thin] (rel axis cs:0,0) -- (rel axis cs:1,0);
        \draw[thin] (rel axis cs:0,0) -- (rel axis cs:0,1);
    }
}

\draw[pattern=north west lines, pattern color=red, draw=red] (axis cs:0,0) rectangle (axis cs:1,1);
\draw[pattern=north west lines, pattern color=red, draw=red] (axis cs:1,1) rectangle (axis cs:3,2);
\draw[pattern=north west lines, pattern color=red, draw=red] (axis cs:3,2) rectangle (axis cs:6,3);

\draw[pattern=north east lines, pattern color=blue, draw=blue] (axis cs:6,3) rectangle (axis cs:10,3.5);

\draw[<->, thick, color=blue] (axis cs:8,4.7) -- (axis cs:8,3.5) 
    node[midway, right] {${\color{blue}\boldsymbol{{\underline{\Gamma}}(\Gamma)}/c}$};
\draw[<->, thick] (axis cs:10.2,3.5) -- (axis cs:10.2,0.01) 
    node[midway, right] {$\displaystyle\sum_{\tau\in[t^\dagger]}\alloc_\tau$};

\draw[<->, thick] (axis cs:1.2,1) -- (axis cs:1.2,0.01) 
    node[midway, right] {$x_1$};
\draw[<->, thick] (axis cs:3.2,2) -- (axis cs:3.2,1) 
    node[midway, right] {$x_2$};
\draw[<->, thick] (axis cs:6.2,3) -- (axis cs:6.2,2) 
    node[midway, right] {$x_{t^\dagger-1}$};

\draw[thick] (axis cs:0,4.7) -- (axis cs:10,4.7);
\node[anchor=east] at (axis cs:0,4.7) {$R_0$};
\node[anchor=east] at (axis cs:0,0) {$L_0$};

\draw[dashed, thin] (axis cs:1,1) -- (axis cs:1,0);
\draw[dashed, thin] (axis cs:3,2) -- (axis cs:3,0);
\draw[dashed, thin] (axis cs:6,3) -- (axis cs:6,0);
\draw[dashed, thin] (axis cs:10,3.5) -- (axis cs:10,0);


\draw[dashed, thin] (axis cs:0,1) -- (axis cs:1,1);
\draw[dashed, thin] (axis cs:0,2) -- (axis cs:3,2);
\draw[dashed, thin] (axis cs:0,3) -- (axis cs:6,3);
\draw[dashed, thin] (axis cs:0,3.5) -- (axis cs:10,3.5);
\draw[dashed, thin] (axis cs:0,3.7) -- (axis cs:10,3.7);

\draw [decorate,decoration={brace,amplitude=5pt,mirror,raise=1ex}]
(axis cs:0,0) -- (axis cs:1,0) node[midway,yshift=-2em]{$\frac{1}{\rho_1}$};
\draw [decorate,decoration={brace,amplitude=5pt,mirror,raise=1ex}]
(axis cs:1,0) -- (axis cs:3,0) node[midway,yshift=-2em]{$\frac{1}{\rho_2}$};
\draw [decorate,decoration={brace,amplitude=5pt,mirror,raise=1ex}]
(axis cs:3,0) -- (axis cs:6,0) node[midway,yshift=-2em]{$\frac{1}{\rho_3}$};
\draw [decorate,decoration={brace,amplitude=5pt,mirror,raise=1ex}]
(axis cs:6,0) -- (axis cs:10,0) node[midway,yshift=-2em]{$\frac{1}{\rho_{t^\dagger}}$};
\end{axis}
\end{tikzpicture}
\vspace{-4mm}

%% file: Paper/figs/fig-example-large-supply.tex
\begin{tikzpicture}[scale=0.27, transform shape]
\begin{axis}[
axis line style=gray,
axis lines=middle,
        ytick = \empty,
        yticklabels = \empty,
        xtick = {0.7, 5.2, 9.7},
        xticklabels = {$1$, $\dots$, $T$},
x label style={at={(axis description cs:1.0,-0.01)},anchor=north},
y label style={at={(axis description cs:0.05,0.95)},anchor=south},
ylabel = {},
label style={font=\LARGE},
xmin=0,xmax=10.5,ymin=0,ymax=6.5,
width=1.2\textwidth,
height=0.8\textwidth,
samples=50]

\addplot[color=red, fill=white!90!red, dotted,very thick] coordinates {(0.4, 6)(0.4,0.01171875)(1.0,0.01171875)(1.0, 6)(0.4, 6)};
\addplot[color=red, fill=white!90!red, dotted,very thick] coordinates {(1.4, 6)(1.4,0.0234375)(2.0,0.0234375)(2.0, 6)(1.4, 6)};
\addplot[color=red, fill=white!90!red, dotted,very thick] coordinates {(2.4, 6)(2.4,0.046875)(3.0,0.046875)(3.0, 6)(2.4, 6)};
\addplot[color=red, fill=white!90!red, dotted,very thick] coordinates {(3.4, 6)(3.4,0.09375)(4.0,0.09375)(4.0, 6)(3.4, 6)};
\addplot[color=red, fill=white!90!red, dotted,very thick] coordinates {(4.4, 6)(4.4,0.1875)(5.0,0.1875)(5.0, 6)(4.4, 6)};
\addplot[color=red, fill=white!90!red, dotted,very thick] coordinates {(5.4, 6)(5.4,0.375)(6.0,0.375)(6.0, 6)(5.4, 6)};
\addplot[color=red, fill=white!90!red, dotted,very thick] coordinates {(6.4, 6)(6.4,0.75)(7.0,0.75)(7.0, 6)(6.4, 6)};
\addplot[color=red, fill=white!90!red, dotted,very thick] coordinates {(7.4, 6)(7.4,1.5)(8.0,1.5)(8.0, 6)(7.4, 6)};
\addplot[color=red, fill=white!90!red, dotted,very thick] coordinates {(8.4, 6)(8.4,3.0)(9.0,3.0)(9.0, 6)(8.4, 6)};
\addplot[color=red, fill=white!90!red, dotted,very thick] coordinates {(9.4, 6)(9.4,6.0)(10.0,6.0)(10.0, 6)(9.4, 6)};
\addplot[white!60!black,fill=white!60!black] coordinates {(0.5, 0)(0.5,0)(0.9,0)(0.9, 0)(0.5, 0)};
\addplot[white!60!black,fill=white!60!black] coordinates {(1.5, 0)(1.5,0.005859375)(1.9,0.005859375)(1.9, 0)(1.5, 0)};
\addplot[white!60!black,fill=white!60!black] coordinates {(2.5, 0)(2.5,0.017578125)(2.9,0.017578125)(2.9, 0)(2.5, 0)};
\addplot[white!60!black,fill=white!60!black] coordinates {(3.5, 0)(3.5,0.041015625)(3.9,0.041015625)(3.9, 0)(3.5, 0)};
\addplot[white!60!black,fill=white!60!black] coordinates {(4.5, 0)(4.5,0.087890625)(4.9,0.087890625)(4.9, 0)(4.5, 0)};
\addplot[white!60!black,fill=white!60!black] coordinates {(5.5, 0)(5.5,0.181640625)(5.9,0.181640625)(5.9, 0)(5.5, 0)};
\addplot[white!60!black,fill=white!60!black] coordinates {(6.5, 0)(6.5,0.369140625)(6.9,0.369140625)(6.9, 0)(6.5, 0)};
\addplot[white!60!black,fill=white!60!black] coordinates {(7.5, 0)(7.5,0.744140625)(7.9,0.744140625)(7.9, 0)(7.5, 0)};
\addplot[white!60!black,fill=white!60!black] coordinates {(8.5, 0)(8.5,1.494140625)(8.9,1.494140625)(8.9, 0)(8.5, 0)};
\addplot[white!60!black,fill=white!60!black] coordinates {(9.5, 0)(9.5,2.994140625)(9.9,2.994140625)(9.9, 0)(9.5, 0)};
\addplot[fill=white!0!black,postaction={pattern=north east lines,pattern color=white!100!black,}] coordinates {(0.5, 0.005859375)(0.5,0)(0.9,0)(0.9, 0.005859375)(0.5, 0.005859375)};
\addplot[fill=white!0!black,postaction={pattern=north east lines,pattern color=white!100!black,}] coordinates {(1.5, 0.017578125)(1.5,0.005859375)(1.9,0.005859375)(1.9, 0.017578125)(1.5, 0.017578125)};
\addplot[fill=white!0!black,postaction={pattern=north east lines,pattern color=white!100!black,}] coordinates {(2.5, 0.041015625)(2.5,0.017578125)(2.9,0.017578125)(2.9, 0.041015625)(2.5, 0.041015625)};
\addplot[fill=white!0!black,postaction={pattern=north east lines,pattern color=white!100!black,}] coordinates {(3.5, 0.087890625)(3.5,0.041015625)(3.9,0.041015625)(3.9, 0.087890625)(3.5, 0.087890625)};
\addplot[fill=white!0!black,postaction={pattern=north east lines,pattern color=white!100!black,}] coordinates {(4.5, 0.181640625)(4.5,0.087890625)(4.9,0.087890625)(4.9, 0.181640625)(4.5, 0.181640625)};
\addplot[fill=white!0!black,postaction={pattern=north east lines,pattern color=white!100!black,}] coordinates {(5.5, 0.369140625)(5.5,0.181640625)(5.9,0.181640625)(5.9, 0.369140625)(5.5, 0.369140625)};
\addplot[fill=white!0!black,postaction={pattern=north east lines,pattern color=white!100!black,}] coordinates {(6.5, 0.744140625)(6.5,0.369140625)(6.9,0.369140625)(6.9, 0.744140625)(6.5, 0.744140625)};
\addplot[fill=white!0!black,postaction={pattern=north east lines,pattern color=white!100!black,}] coordinates {(7.5, 1.494140625)(7.5,0.744140625)(7.9,0.744140625)(7.9, 1.494140625)(7.5, 1.494140625)};
\addplot[fill=white!0!black,postaction={pattern=north east lines,pattern color=white!100!black,}] coordinates {(8.5, 2.994140625)(8.5,1.494140625)(8.9,1.494140625)(8.9, 2.994140625)(8.5, 2.994140625)};
\addplot[fill=white!0!black,postaction={pattern=north east lines,pattern color=white!100!black,}] coordinates {(9.5, 5.994140625)(9.5,2.994140625)(9.9,2.994140625)(9.9, 5.994140625)(9.5, 5.994140625)};

\end{axis}

\end{tikzpicture}

%% file: Paper/figs/fig-example-small-supply.tex
\begin{tikzpicture}[scale=0.27, transform shape]
\begin{axis}[
axis line style=gray,
axis lines=middle,
        ytick = \empty,
        yticklabels = \empty,
        xtick = {0.7, 5.2, 9.7},
        xticklabels = {$1$, $\dots$, $T$},
x label style={at={(axis description cs:1.0,-0.01)},anchor=north},
y label style={at={(axis description cs:0.05,0.95)},anchor=south},
ylabel = {},
label style={font=\LARGE},
xmin=0,xmax=10.5,ymin=0,ymax=6.5,
width=1.2\textwidth,
height=0.8\textwidth,
samples=50]

\addplot[color=red, fill=white!90!red, dotted,very thick] coordinates {(0.4, 6)(0.4,0.01171875)(1.0,0.01171875)(1.0, 6)(0.4, 6)};
\addplot[color=red, fill=white!90!red, dotted,very thick] coordinates {(1.4, 6)(1.4,0.0234375)(2.0,0.0234375)(2.0, 6)(1.4, 6)};
\addplot[color=red, fill=white!90!red, dotted,very thick] coordinates {(2.4, 6)(2.4,0.046875)(3.0,0.046875)(3.0, 6)(2.4, 6)};
\addplot[color=red, fill=white!90!red, dotted,very thick] coordinates {(3.4, 6)(3.4,0.09375)(4.0,0.09375)(4.0, 6)(3.4, 6)};
\addplot[color=red, fill=white!90!red, dotted,very thick] coordinates {(4.4, 6)(4.4,0.1875)(5.0,0.1875)(5.0, 6)(4.4, 6)};
\addplot[color=red, fill=white!90!red, dotted,very thick] coordinates {(5.4, 6)(5.4,0.375)(6.0,0.375)(6.0, 6)(5.4, 6)};
\addplot[color=red, fill=white!90!red, dotted,very thick] coordinates {(6.4, 6)(6.4,0.75)(7.0,0.75)(7.0, 6)(6.4, 6)};
\addplot[color=red, fill=white!90!red, dotted,very thick] coordinates {(7.4, 6)(7.4,1.5)(8.0,1.5)(8.0, 6)(7.4, 6)};
\addplot[color=red, fill=white!90!red, dotted,very thick] coordinates {(8.4, 6)(8.4,3.0)(9.0,3.0)(9.0, 6)(8.4, 6)};
\addplot[color=red, fill=white!90!red, dotted,very thick] coordinates {(9.4, 6)(9.4,6.0)(10.0,6.0)(10.0, 6)(9.4, 6)};

\addplot[white!60!black,fill=white!60!black] coordinates {(0.5, 0)(0.5,0)(0.9,0)(0.9, 0)(0.5, 0)};
\addplot[white!60!black,fill=white!60!black] coordinates {(1.5, 0)(1.5,2.853049850463867)(1.9,2.853049850463867)(1.9, 0)(1.5, 0)};
\addplot[white!60!black,fill=white!60!black] coordinates {(2.5, 0)(2.5,2.864768600463867)(2.9,2.864768600463867)(2.9, 0)(2.5, 0)};
\addplot[white!60!black,fill=white!60!black] coordinates {(3.5, 0)(3.5,2.888206100463867)(3.9,2.888206100463867)(3.9, 0)(3.5, 0)};
\addplot[white!60!black,fill=white!60!black] coordinates {(4.5, 0)(4.5,2.935081100463867)(4.9,2.935081100463867)(4.9, 0)(4.5, 0)};
\addplot[white!60!black,fill=white!60!black] coordinates {(5.5, 0)(5.5,3.028831100463867)(5.9,3.028831100463867)(5.9, 0)(5.5, 0)};
\addplot[white!60!black,fill=white!60!black] coordinates {(6.5, 0)(6.5,3.1469489117346603)(6.9,3.1469489117346603)(6.9, 0)(6.5, 0)};
\addplot[white!60!black,fill=white!60!black] coordinates {(7.5, 0)(7.5,3.1469489117346603)(7.9,3.1469489117346603)(7.9, 0)(7.5, 0)};
\addplot[white!60!black,fill=white!60!black] coordinates {(8.5, 0)(8.5,3.1469489117346603)(8.9,3.1469489117346603)(8.9, 0)(8.5, 0)};
\addplot[white!60!black,fill=white!60!black] coordinates {(9.5, 0)(9.5,3.1469489117346603)(9.9,3.1469489117346603)(9.9, 0)(9.5, 0)};
\addplot[fill=white!0!black,postaction={pattern=north east lines,pattern color=white!100!black,}] coordinates {(0.5, 2.853049850463867)(0.5,0)(0.9,0)(0.9, 2.853049850463867)(0.5, 2.853049850463867)};
\addplot[fill=white!0!black,postaction={pattern=north east lines,pattern color=white!100!black,}] coordinates {(1.5, 2.864768600463867)(1.5,2.853049850463867)(1.9,2.853049850463867)(1.9, 2.864768600463867)(1.5, 2.864768600463867)};
\addplot[fill=white!0!black,postaction={pattern=north east lines,pattern color=white!100!black,}] coordinates {(2.5, 2.888206100463867)(2.5,2.864768600463867)(2.9,2.864768600463867)(2.9, 2.888206100463867)(2.5, 2.888206100463867)};
\addplot[fill=white!0!black,postaction={pattern=north east lines,pattern color=white!100!black,}] coordinates {(3.5, 2.935081100463867)(3.5,2.888206100463867)(3.9,2.888206100463867)(3.9, 2.935081100463867)(3.5, 2.935081100463867)};
\addplot[fill=white!0!black,postaction={pattern=north east lines,pattern color=white!100!black,}] coordinates {(4.5, 3.028831100463867)(4.5,2.935081100463867)(4.9,2.935081100463867)(4.9, 3.028831100463867)(4.5, 3.028831100463867)};
\addplot[fill=white!0!black,postaction={pattern=north east lines,pattern color=white!100!black,}] coordinates {(5.5, 3.1469489117346603)(5.5,3.028831100463867)(5.9,3.028831100463867)(5.9, 3.1469489117346603)(5.5, 3.1469489117346603)};

\end{axis}

\end{tikzpicture}

%% file: Paper/figs/fig-example-uninformative-prediction.tex
\begin{tikzpicture}[scale=0.27, transform shape]
\begin{axis}[
axis line style=gray,
axis lines=middle,
        ytick = \empty,
        yticklabels = \empty,
        xtick = {0.7, 5.2, 9.7},
        xticklabels = {$1$, $\dots$, $T$},
x label style={at={(axis description cs:1.0,-0.01)},anchor=north},
y label style={at={(axis description cs:0.05,0.95)},anchor=south},
ylabel = {},
label style={font=\LARGE},
xmin=0,xmax=10.5,ymin=0,ymax=6.5,
width=1.2\textwidth,
height=0.8\textwidth,
samples=50]

\addplot[color=red, fill=white!90!red, dotted,very thick] coordinates {(0.4, 6)(0.4,0)(1.0,0)(1.0, 6)(0.4, 6)};
\addplot[color=red, fill=white!90!red, dotted,very thick] coordinates {(1.4, 6)(1.4,0)(2.0,0)(2.0, 6)(1.4, 6)};
\addplot[color=red, fill=white!90!red, dotted,very thick] coordinates {(2.4, 6)(2.4,0)(3.0,0)(3.0, 6)(2.4, 6)};
\addplot[color=red, fill=white!90!red, dotted,very thick] coordinates {(3.4, 6)(3.4,0)(4.0,0)(4.0, 6)(3.4, 6)};
\addplot[color=red, fill=white!90!red, dotted,very thick] coordinates {(4.4, 6)(4.4,0)(5.0,0)(5.0, 6)(4.4, 6)};
\addplot[color=red, fill=white!90!red, dotted,very thick] coordinates {(5.4, 6)(5.4,0)(6.0,0)(6.0, 6)(5.4, 6)};
\addplot[color=red, fill=white!90!red, dotted,very thick] coordinates {(6.4, 6)(6.4,0)(7.0,0)(7.0, 6)(6.4, 6)};
\addplot[color=red, fill=white!90!red, dotted,very thick] coordinates {(7.4, 6)(7.4,0)(8.0,0)(8.0, 6)(7.4, 6)};
\addplot[color=red, fill=white!90!red, dotted,very thick] coordinates {(8.4, 6)(8.4,0)(9.0,0)(9.0, 6)(8.4, 6)};
\addplot[color=red, fill=white!90!red, dotted,very thick] coordinates {(9.4, 6)(9.4,4.199999999999999)(10.0,4.199999999999999)(10.0, 6)(9.4, 6)};

\addplot[white!60!black,fill=white!60!black] coordinates {(0.5, 0)(0.5,0)(0.9,0)(0.9, 0)(0.5, 0)};
\addplot[white!60!black,fill=white!60!black] coordinates {(1.5, 0)(1.5,2.0270919799804688)(1.9,2.0270919799804688)(1.9, 0)(1.5, 0)};
\addplot[white!60!black,fill=white!60!black] coordinates {(2.5, 0)(2.5,2.0270919799804688)(2.9,2.0270919799804688)(2.9, 0)(2.5, 0)};
\addplot[white!60!black,fill=white!60!black] coordinates {(3.5, 0)(3.5,2.0270919799804688)(3.9,2.0270919799804688)(3.9, 0)(3.5, 0)};
\addplot[white!60!black,fill=white!60!black] coordinates {(4.5, 0)(4.5,2.0270919799804688)(4.9,2.0270919799804688)(4.9, 0)(4.5, 0)};
\addplot[white!60!black,fill=white!60!black] coordinates {(5.5, 0)(5.5,2.0270919799804688)(5.9,2.0270919799804688)(5.9, 0)(5.5, 0)};
\addplot[white!60!black,fill=white!60!black] coordinates {(6.5, 0)(6.5,2.0270919799804688)(6.9,2.0270919799804688)(6.9, 0)(6.5, 0)};
\addplot[white!60!black,fill=white!60!black] coordinates {(7.5, 0)(7.5,2.0270919799804688)(7.9,2.0270919799804688)(7.9, 0)(7.5, 0)};
\addplot[white!60!black,fill=white!60!black] coordinates {(8.5, 0)(8.5,2.0270919799804688)(8.9,2.0270919799804688)(8.9, 0)(8.5, 0)};
\addplot[white!60!black,fill=white!60!black] coordinates {(9.5, 0)(9.5,2.0270919799804688)(9.9,2.0270919799804688)(9.9, 0)(9.5, 0)};
\addplot[fill=white!0!black,postaction={pattern=north east lines,pattern color=white!100!black,}] coordinates {(0.5, 2.0270919799804688)(0.5,0)(0.9,0)(0.9, 2.0270919799804688)(0.5, 2.0270919799804688)};
\addplot[fill=white!0!black,postaction={pattern=north east lines,pattern color=white!100!black,}] coordinates {(9.5, 3.9729057039368953)(9.5,2.0270919799804688)(9.9,2.0270919799804688)(9.9, 3.9729057039368953)(9.5, 3.9729057039368953)};

\end{axis}

\end{tikzpicture}

%% file: Paper/lp-resolve.tex



In this section, we present another minimax-optimal algorithm, {\OPTReS}, which builds on program~\ref{eq:opt reduced form} using the idea of resolving.

The key observation is that for each day $t \in [T]$, the platform effectively faces a subproblem with $T - t$ remaining days. The parameters of this subproblem—namely, the initial workforce pool size, demand range, and staffing level—are determined by the staffing decisions and prediction intervals from the previous $t-1$ days.
Accordingly, {\OPTReS} \emph{(re-)solves} program~\ref{eq:opt reduced form} for the subproblem starting at day $t$, obtaining the optimal solution $\{\alloc_{i\tau}^{(t)}\}_{i \in [n], \tau \in [t:T]}$. It then implements the staffing profile $\{\alloc_{it}^{(t)}\}_{i \in [n]}$ for the current day $t$. A formal algorithm description can be found in \Cref{alg:lp-resolve}.

\begin{algorithm}
 \setcounter{AlgoLine}{0}
    \SetKwInOut{Input}{input}
    \SetKwInOut{Output}{output}
    \Input{initial pool sizes $\{\supplyi\}_{i\in[n]}$, availability rates $\{\wdiscountit\}_{i\in[n], t\in[T]}$, initial demand range $[\pLzero,\pRzero]$, prediction error upper bounds $\{\perrort\}_{t\in[T]}$, prediction inconsistency upper bounds  $\{\pbiast\}_{t\in[T]}$}

    initialize the state of the system: $\bar{\supply}_i\leftarrow \supply_i$ and $\bar{\cumalloc}_i\leftarrow 0$ for $i\in[n]$. {\small\color{\commentcolor}\tcc{$\mathbf{\bar \supply}$ and $\mathbf{\bar \cumalloc}$ track the number of available workers and the total number of workers hired in each pool}}

    initialize the current prediction upper bound: $\bar{\pR}\leftarrow \pRzero$.

    initialize the current projected availability rates: $\bar{\wdiscount}_{it}\leftarrow \wdiscount_{it}$ for $t\in[T],i\in[n]$
    
    \For{each day $t\in[T]$}{
    {\small\color{\commentcolor}\tcc{prediction $[\pLt,\pRt]$ reveals}}

        update 
        the current prediction upper bound: $\bar{\pR}\leftarrow \min\left\{\bar{\pR}, \pRt + \pbiast\right\}$.
        
       find an optimal solution $(\mathbf{\alloc}^{(t)},\targetcost^{(t)})$ of the ``adjusted'' \ref{eq:opt reduced form} given the current $\left(\mathbf{\bar{\supply}},\mathbf{\bar{\cumalloc}},\bar{\boldsymbol{\wdiscount}},\bar{\pR}\right)$: 
       \begin{align*}
        \begin{array}{llll}
    \min\limits_{\substack{\xbf,\targetcost\geq \zerobf}}\quad\quad
    &
    \targetcost
    & 
    \text{s.t.}
    \\
    &
    \displaystyle\sum\nolimits_{\tau\in [t:T]}\frac{1}{\bar{\wdiscount}_{i\tau}}\alloc_{i\tau} \leq \bar{\supply}_i
    &
    i\in[n]
    \\
    &
    \displaystyle\sum\nolimits_{i\in[n]}\left(\bar{\cumalloc}_i+\sum\nolimits_{\tau\in[t:k]}
    \alloc_{i\tau}\right) \leq \max_{\tau\in[t:k]}\left({\bar{\pR} - \perror_{\tau} - 2\pbias_{\tau}}\right) + \frac{\targetcost}{\overcost}
    \quad\quad
    &
    k\in[t:T]
    \\
    &
    \displaystyle\sum\nolimits_{i\in[n]}\left(\bar{\cumalloc}_i+\sum\nolimits_{\tau\in[t:T]}
    \alloc_{i\tau}\right) \geq \bar{\pR} - \frac{\targetcost}{\undercost}
    &
\end{array}
       \end{align*}
       
       set $\alloc_{it}\gets\alloc^{(k)}_{it}$ for all $i\in[n]$,  and hire $\allocit$ available workers from each pool $i\in[n]$.

        \vspace{1mm}

        update the state, 
        and current projected availability rates:
        \begin{align*}
        \forall i\in[n]:
        ~\bar{\supply}_i\leftarrow \bar{\wdiscount}_{it}\bar{\supply}_i-\alloc_{it}~,
        ~\bar{\cumalloc}_i\leftarrow\bar{\cumalloc}_i+\alloc_{it}~,
        ~\{\bar{\wdiscount}_{i\tau}\}_{\tau\in[t+1:T]}\leftarrow \left\{\frac{\bar{\wdiscount}_{i\tau}}{\bar{\wdiscount}_{it}}\right\}_{\tau\in[t+1:T]}
        \end{align*}
    }
    \caption{\textsc{{\OPTReS}}}
    \label{alg:lp-resolve}
\end{algorithm}

The minimax optimality of {\OPTReS} is established in \Cref{thm:opt alg resolving}. In contrast to {\OPTSim}, which directly employs the ``emulator oracle'' (Procedure~\ref{alg:emulator}), {\OPTReS} does not invoke this oracle in its execution. Instead, its optimality is established through an induction argument in which the minimax optimality of {\OPTSim} and thus emulator oracle serve as a key analytical tool. The complete analysis and the proof of minimax optimality of this algorithm is provided in \Cref{apx:opt alg resolving}.

\begin{restatable}{theorem}{thmoptalgresolving}
\label{thm:opt alg resolving}
{\OPTReS} is minimax optimal. Moreover, its optimal minimax cost $\optcost$ coincides with the objective value of program~\ref{eq:opt reduced form}.
\end{restatable}

As a sanity check, when facing the single-switch prediction sequences, {\OPTSim} and {\OPTReS} produce identical staffing profiles and attain the same optimal cost guarantee. For general prediction sequences, however, {\OPTReS} typically outperforms {\OPTSim}, as it preserves minimax optimality across all subproblems. This intuition on performance advantage is also borne out in our numerical experiments (\Cref{sec:numerical}). On the other hand, while both algorithms run in polynomial time, {\OPTReS} is computationally more demanding because it resolves \ref{eq:opt reduced form} at every time step. In practice, one may consider a \emph{hybrid} approach: solve \ref{eq:opt reduced form} at selected checkpoints based on the current state, emulate the resulting canonical solution with Procedure~\ref{alg:emulator} for intermediate days, and then re-solve the program as needed. It is not hard to verify the minimax optimality of this hybrid algorithm.


%% file: Paper/extension.tex
\label{sec:extension}
In this section, we generalize our approach from \Cref{sec:base model} to two extensions. First, we consider a setting with multiple stations (or multiple operating days), each with its own unknown demand to meet (\Cref{sec:multi-station}). We allow for a general setting with heterogeneous prediction sequences and demands. Second, we consider the model where the platform can release previously hired workers by incurring an additional cost (\Cref{sec:costly-hiring}).

To simplify the presentation, we assume perfectly consistent prediction intervals throughout this section (i.e., $\pbiass = \pprobb = \zerobf$), implying, without loss of generality, that the intervals are nested. All our results can be extended to settings with $(\pbiass,\pprobb)$-consistency in a straightforward way.


\subsection{Workforce Planning for Multiple Stations}
\input{Paper/multi-station}

\subsection{Workforce Planning with Costly Hiring and Releasing}
\label{sec:costly-hiring}
In this section, we consider an extension in which the platform has a budget and the hiring is costly. Moreover, in contrast to our based model, the platform is allowed to revoke previous hiring decisions after paying a cost, which we refer to as \emph{costly relasing}. All missing technical details can be found in \Cref{sec:cancellation}.

\xhdr{The costly hiring and releasing environment.} We consider the following generalization of the base model studied in \Cref{sec:base model}. The platform has a total \emph{budget} of $\budget$ for the staffing decision. On each day~$t\in[T]$, by hiring  $\allocit\in\reals_+\cup\{\infty\}$ available workers from each pool~$i$, the platform needs to \emph{pay} $\allocit\cdot \priceit$ where $\priceit$ is the per-worker wages for pool $i$ on day $t$. Moreover, the platform can also \emph{release} $\revokeit\in\reals_+$ previously hired workers from each pool $i$ by paying a per-worker releasing fee $\cpriceit\in\reals_+\cup\{\infty\}$. We further assume that if a worker is hired and later released by the platform, she cannot be hired again for this operating day. We say a (joint hiring and releasing) staffing profile $\{\allocit, \revokeit\}_{i\in[n],t\in[T]}$ is \emph{feasible} if
it is \emph{supply feasible} ($\sum\nolimits_{t\in[T]}\frac{1}{\wdiscountit}\allocit \leq \supplyi$ for every $i\in[n]$), \emph{budget feasible} ($\sum\nolimits_{i\in[n]}\sum\nolimits_{t\in[T]}
    \priceit  \allocit + \cpriceit \revokeit\leq \budget$),
    and 
    \emph{releasing feasible} 
    ($\sum\nolimits_{t\in[k]}\revokeit \leq 
    \sum\nolimits_{t\in[k]}\allocit$ for every $i\in[n]$, $k\in[T]$).
%
%
We also make the following structural assumption about the per-worker releasing fees.

\begin{assumption}[Piecewise stationary releasing fees]
\label{asp:limited cancellation fee}
    There exists $\cptotal\in [1:T]$ and $0 = t_0 < t_1 < t_2 < \dots < t_\cptotal = T$ such that for every index $\ell\in[\cptotal]$, the per-worker releasing fees remain identical for each time interval $[t_{\ell - 1} + 1: t_{\ell}]$, i.e., $\forall t,t'\in[t_{\ell - 1} + 1: t_{\ell}]$ and $\forall i, i'\in[n]$, we have
        $\cpriceit\equiv \cprice_{i't'}$.
\end{assumption}

To extend our approach to this extension, we face new challenges. First, restricting the adversary to only single-switching prediction sequences in \eqref{eq:single switch prediction construction} is not without loss. 
The adversary can benefit from multiple switches to force the algorithm to hedge more through its releasing decisions. Second, it is not clear how to make releasing decisions and how to emulate them similar to Procedure~\ref{alg:emulator}.

\xhdr{High-level sketch of our approach.} We first introduce a larger subset of $O(T^\cptotal)$ of prediction sequences, formally described in ~\eqref{eq:multi switch prediction construction} (in contrast to $T$ prediction sequences in \eqref{eq:single switch prediction construction}). In short, since the releasing fees remain constant in each epoch $\cintervalell
\subseteq[T]$, we consider an adversary that follows a single-switch strategy in each epoch, such as the one introduced in \eqref{eq:single switch prediction construction} for the base model. (Note that our base model is simply a special case when when $\cptotal = 1$, $\priceit = 0$ for all $i,t$, and $\cprice_{i1} = \infty$ for all $i$.) We then concatenate these prediction sequences for different epochs $\cintervalell$ to obtain the entire prediction sequence. Using this new subset, we introduce a \emph{configuration linear program}~\ref{eq:opt cancellation} that helps us to characterize the optimal minimax cost. When $\cptotal$ is constant (which is generally satisfied in our application), program~\ref{eq:opt cancellation} has a polynomial size.



\revcolor{Emulating the solution of \ref{eq:opt cancellation} for a general prediction sequence introduces additional intricacies in the extension model. We address these challenges in two steps. First, we identify structural properties of the minimax optimal algorithm and incorporate them directly into \ref{eq:opt cancellation}. Second, in contrast to the minimax optimal {\OPTSim} developed in previous sections—which solves a single offline program once and then emulates it throughout the horizon—the minimax optimal algorithm {\OPTSimCan} (\Cref{alg:opt cancellation}) resolves an updated version of \ref{eq:opt cancellation} at the beginning of each epoch $\cintervalell$, following an approach analogous to {\OPTReS}. This resolving procedure also enables an induction-based proof of the minimax optimality of {\OPTSimCan}. Furthermore, {\OPTSimCan} admits a running time of $\texttt{Poly}(n, T^\cptotal)$. The details of the LP formulation and the associated LP-emulator algorithm are presented in \Cref{sec:cancellation}. These components together yield the following theorem.
\vspace{0mm}
\begin{restatable}{theorem}{optalgcan}
\label{thm:opt alg cancellation}
    In the costly hiring and releasing environment and under Assumption~\ref{asp:limited cancellation fee},
    {\OPTSimCan} is minimax optimal. Its optimal minimax cost $\optcost$ is equal to the objective value of program~\ref{eq:opt cancellation}.
\end{restatable}
}


%% file: Paper/multi-station.tex
\label{sec:multi-station}

 In last-mile delivery, the platform might prefer to couple staffing decisions across multiple stations that share the same worker pools (e.g., due to geographical proximity) but have different predictions and final demands. 
We first generalize our model to incorporate this extension, and then study how our algorithmic approach can be adapted to two relevant objective functions: the \emph{egalitarian} and \emph{utilitarian} staffing cost functions.

\xhdr{The multi-station environment.} We consider the setting in which the platform manages staffing decisions for $m$ delivery stations. Each delivery station $j\in[m]$ has an unknown target staffing demand $\demandj$ with a station-dependent initial range $[\pLjzero,\pRjzero]$. On each day $t\in[T]$, the platform receives a demand forecast $\prediction_t = \{[\pLjt,\pRjt]\}_{j\in[m]}$ where the interval $[\pLjt,\pRjt]$ is the prediction of demand~$\demandj$ for station $j$. As we mentioned at the beginning of this section, we assume that demand predictions are consistent (i.e., $\demandj\in[\pLjt,\pRjt]$) and have station-dependent bounded error (i.e., $\pRjt -\pLjt \leq \perrorjt$). To simultaneously fulfill all the $m$ demands $\{\demandj\}_{j\in[m]}$ on the operating day $T+1$, the platform irrevocably hires $\allocijt\in\reals_+$ available workers from pool $i$ to be assigned to station $j$ on each day $t\in[T]$, subject to the feasibility, that is,
$\sum\nolimits_{t\in[T]}\sum\nolimits_{j\in[m]}\frac{1}{\wdiscountit}\allocijt \leq \supplyi$ for every pool $i\in[n]$.

\xhdr{Multi-station cost objective functions.} Given staffing profile $\xbf =\{\allocijt\}_{i\in[n],j\in[m],t\in[T]}$ and demand profile $\dbf = \{\demandj\}_{j\in[m]}$, the \emph{egalitarian-staffing cost} $\CostInf[\dbf]{\xbf}$ and \emph{utilitarian-staffing cost} $\CostOne[\dbf]{\xbf}$ are defined as:
\begin{align*}
    \CostInf[\dbf]{\xbf} &\triangleq 
    \max\nolimits_{j\in[m]}
    \Cost[\demandj]{\xbf_j},\qquad
    \CostOne[\dbf]{\xbf} \triangleq 
    \sum\nolimits_{j\in[m]}\Cost[\demandj]{\xbf_j}
\end{align*}
where $\Cost[\demandj]{\xbf_j}$ is the staffing cost of station $j$ (defined in \Cref{sec:prelim}) given its staffing profile $\xbf_j= \{\allocijt\}_{i\in[n],t\in[T]}$ and demand $\demandj$.

\xhdr{The minimax optimal algorithm.}
Following the recipe in \Cref{sec:base model result}, we introduce the following linear program, with variables $\{\allocijt,\targetcost_j\}_{i\in[n],j\in[m],t\in[T]}$:
\begin{align}
\tag{$\textsc{LP-multi-station}$}
    \label{eq:opt reduced form multi station}
    &\arraycolsep=1.4pt\def\arraystretch{1.2}
    \begin{array}{llll}
    \min\limits_{\substack{\xbf,\targetcosts\geq \zerobf}}
    \qquad\qquad
    &
    \Psi(\targetcost_1, \dots, \targetcost_m)
    & 
    \text{s.t.}
    \\
    &
    \displaystyle\sum\nolimits_{t\in[T]}\frac{1}{\wdiscountit}\cdot\sum\nolimits_{j\in[m]}\allocijt \leq \supplyi
    &
    i\in[n]
    \\
    &
    \displaystyle\sum\nolimits_{i\in[n]}\sum\nolimits_{t\in[k]}
    \allocijt \leq \pRjzero - \perrorjk + \frac{\targetcost}{\overcost}
    \quad\quad
    &
    j\in[m], k\in[T]
    \\
    &
    \displaystyle\sum\nolimits_{i\in[n]}\sum\nolimits_{t\in[T]}
    \allocijt \geq \pRjzero - \frac{\targetcost}{\undercost}
    &
    j\in[m]
\end{array}
\end{align}
where the objective function $\Psi(\boldsymbol{\targetcost})$ with $\boldsymbol{\targetcost} = (\targetcost_1, \dots, \targetcost_m)$ is defined as $\Psi(\boldsymbol{\targetcost}) \triangleq \max\nolimits_{j\in[m]}\targetcost_j$ and $
    \Psi(\boldsymbol{\targetcost}) \triangleq \sum\nolimits_{j\in[m]}\targetcost_j$ for the egalitarian-staffing cost and utilitarian-staffing cost, respectively.
Clearly, program~\ref{eq:opt reduced form multi station} is a natural generalization of program~\ref{eq:opt reduced form} in \Cref{sec:base model} for the multi-station environment. Following the same intuition behind~\ref{eq:opt reduced form}, ~\ref{eq:opt reduced form multi station} essentially characterizes the cost guarantee and the canonical staffing profile of online algorithms that are candidates for optimality against a particular subset of adversarial prediction sequences. Specifically, we again consider sequences that for each station only have a single switch (similar to \eqref{eq:single switch prediction construction}). Though the proof intuition is the same, showing no online algorithms can beat the optimal objective value in \ref{eq:opt reduced form multi station} becomes more complicated, especially for the utilitarian-staffing cost. A key technical step in our analysis requires arguing that when facing single-switch prediction sequence in the multi-station environment, the minimax optimal algorithm always have weakly smaller overstaffing cost than the understaffing cost in the worst case \Cref{lem:L infty overcost vs undercost}.

Now we present the minimax optimal algorithm for egalitarian-staffing cost (\Cref{alg:opt multi station}) and its guarantee of optimality  (\Cref{thm:opt alg multi station}). Its proof follows a similar but more complicated approach as that of \Cref{thm:opt alg} and is deferred to \Cref{apx:optalginfty}. Furthermore, following the same discussion as in \Cref{sec:base model result}, it can be verified that {\OPTSimInfty} is feasible and has a polynomial running time.

\begin{algorithm}
 \setcounter{AlgoLine}{0}
    \SetKwInOut{Input}{input}
    \SetKwInOut{Output}{output}
    \Input{initial pool sizes $\{s_i\}_{i\in[n]}$, availability rates $\{\rho_{it}\}_{i\in[n], t\in[T]}$, initial demand ranges $\{[L_{j0},R_{j0}]\}_{j\in[m]}$, prediction error upper bounds $\{\Delta_{jt}\}_{j\in[m],t\in[T]}$}
    \Output{staffing profile $\xbf$}
    \vspace{2mm}
    find an optimal solution $(\xbf^*, \targetcost^*)$ of program~\ref{eq:opt reduced form multi station}
    
    \vspace{2mm}
    \For{each station $j\in[m]$}{
    invoke Procedure~\ref{alg:emulator} with canonical staffing profile $\canxbf \gets \{\allocijt^*\}_{i\in[n],t\in[T]}$ \tcc{\small\color{\commentcolor} facing prediction sequence $\tilde\predictions\gets \{[\pLjt,\pRjt]\}_{t\in[T]}$ for each station $j$}
    }
    \caption{\OPTSimInfty}
    \label{alg:opt multi station}
\end{algorithm}


\begin{restatable}{theorem}{optalgmultistation}
\label{thm:opt alg multi station}
    In the multi-station environment, 
    {\OPTSimInfty} is minimax optimal. Its optimal minimax cost $\optcost$ is equal to the objective value of~\ref{eq:opt reduced form multi station}.
\end{restatable}
\begin{remark}
    For both staffing cost definitions, our proposed minimax optimal algorithm ({\OPTSimInfty}) has a desired simplicity: after finding an optimal canonical staffing profile using their corresponding linear programs, the staffing profile determined by the subroutine (Procedure~\ref{alg:emulator}) for each station $j$ is independent of the predictions revealed for other stations.
\end{remark}

%% file: Paper/conclude.tex
\vspace{0mm}
\label{sec:conclude}

In this paper, we introduce a new online decision-making problem: how to incrementally allocate resources to meet an uncertain target demand when resource availability diminishes over time but prediction accuracy improves. The cost of over- or undershooting the target may be asymmetric. This framework abstracts real-world challenges like workforce planning (using gig economy workers) in last-mile delivery, prompting us to call it the dynamic staffing problem. Under minimal assumptions about the predictions (i.e., a prior-free approach) and across various settings, we develop computationally efficient minimax-optimal online algorithms that minimize the cost of missing the target against worst-case demand and prediction sequences. 

\xhdr{Future research.} An immediate direction is studying this problem under closely related models (beyond the extensions in \Cref{sec:extension}). From a computational perspective, removing the fluid relaxation---where resource availability shrinks stochastically and integral hiring decisions are required---would necessitate rounding techniques and an integrality-gap analysis. Finally, applying this problem to broader contexts and exploring other online decision-making tasks with progressively improved predictions remain promising avenues for future research.

%% file: Paper/related-work.tex
\label{sec:further}
\emph{Inventory control and dynamic staffing.} 
Building on the work of \citet{Edg-88}, there is a vast literature on multi-period newsvendor models (usually referred to as inventory control) where a sequence of demand is realized over time and the goal is to design inventory policies that minimize the total imbalanced cost.  Early work takes a stochastic modeling approach for demand and uses dynamic programming to study this problem \citep{CS-60}.  More recently, the problem has been studied through the lens of robust optimization \citep[see, e.g.,][]{BT-06} 
and distributionally robust optimization (see e.g., \citealp{XG-22}). 
The important factor distinguishing our work from this literature is that we only have one target demand that is realized at the end of the horizon; to fulfill the demand we made a sequence of decisions based on information that becomes 
progressively more accurate. However, we highlight that modeling demand uncertainty by an interval is a common approach in robust optimization.  

Closer to our work is the literature on dynamic staffing that relies on a hybrid or crowdsourced workforce similar to our leading example \citep[see, e.g.,][]{LJWDP-20, HCD-24, LHS-23,MNR-21,LMS-24}. Some of the papers in the stream also aim to capture the trade-off between hiring less expensive workers early with less accurate information or more expensive ones late with more accurate information. However,  to our knowledge, ours is the first to take a ``prior free'' approach, i.e., studying dynamic staffing without imposing any stochastic structure about the demand or the sequence of predictions received. Also related is the work of  \citet{FW-22} that takes a robust optimization approach to a crowd-sourced staffing problem for on-demand deliveries. Both the setting and the approach of this paper are substantially different from ours. Finally, we considered settings in which previous staffing decisions could be revoked at a cost. This ``costly cancellations'' paradigm has recently explored in other models for online resource allocation~\citep[see, e.g.,][]{babaioff2008selling,ekbatani2022online,ekbatani2024prophet}.

\revcolor{\emph{Inventory management with sequential forecasting.}
Our paper relates to the classical literature on inventory management that leverages sequential forecasts. Starting from foundational multi-period newsvendor models \citep{AHM-51,FR-96}, researchers have explored scenarios where decision-makers can place multiple orders over a planning horizon, updating their inventory levels as they receive increasingly accurate forecasts of uncertain future demand. Prominent examples include the Martingale Model of Forecast Evolution (MMFE), introduced by \citet{hau-69} and extensively analyzed in later work \citep[e.g.,][]{heath1994modeling,SZ-12}. These studies typically adopt a Bayesian modeling approach, assuming that sequential forecasts follow distributions that are progressively less dispersed (Blackwell-ordered), and study how to optimally incorporate these forecasts into inventory decisions \citep{WAK-12, SZ-12,topan2018using} and supply allocation~\citep{papier2016supply}. In contrast, our work departs from this literature by considering a non-Bayesian, robust formulation, in which forecasts are modeled via prediction intervals rather than probabilistic distributions. This approach offers more flexibility and does not rely on precise knowledge of the underlying forecasting process.}

\revcolor{\emph{Online decision making with advice.} A recent stream of work in online decision-making aims to improve upon the paradigm of adversarial modeling---which may lead to conservative algorithms and also disregards any information about the problem instance---by augmenting the problem with some (potentially unreliable) ``advice''. The goal is to design algorithms that take advantage of the advice when accurate but are also robust to inaccurate advice (see, e.g., \citealp{MNS-12, PSK-18, BMS-20, LV-21, APT-22, BGGJ-22, christianson2022chasing, JM-22,   BPSW-23,BKK-23, DNS-23, choo2025learning,feng2024two}; see also \citealp{BFKLM-17} for a survey)
Our framework can also be viewed as a refinement to adversarial modeling by tying the hands of the adversary to reveal progressively more accurate prediction intervals. 
As such, our approach complements this stream of work. However, these two frameworks are not directly comparable due to fundamental differences. Just to name one, in our setting, without any prediction, 
no algorithm can achieve provably good performance; selecting the ``weighted midpoint'' (i.e., $\frac{cR_0 + C L_0}{c + C} \in [L_0, R_0])$) is the best achievable solution. Another line of work is considering designing online algorithms with ``controlled predictions,'' where the online algorithm has side information about the uncertainty in the advice, e.g., \cite{christianson2024risk, sun2024online}. Our approach has conceptual similarities with this line of work in that we also take into account the accuracy of predictions in our decisions, but the problems, models, and results are not comparable.}


\emph{Robust optimization \& mechanism design.} Our problem resembles some aspects of the literature on robust optimization and robust mechanism design, particularly in settings where a decision maker faces ambiguity about latent parameters and must protect against worst‐case realizations.

In classical robust optimization, uncertainty is often modeled via static sets or intervals over unknown parameters, and one seeks solutions that perform well for all realizations in the set, e.g. via a min-max or robust counterpart~\citep{bertsimas2003robust,bertsimas2004price,BT-06,bertsimas2011theory}. However, a less-explored frontier is when uncertainty evolves over time, shrinking or shifting as information is gradually revealed. Some work in adaptive robust optimization introduces dynamic uncertainty sets (for example, in multi-period power dispatch under wind uncertainty) to capture temporal correlation and uncertainty reduction over time~\citep{lorca2014adaptive, lorca2016multistage1,lorca2016multistage2}. More generally, there is growing interest in robust optimization over time for online/distributionally robust frameworks, where ambiguity sets shrink or adjust in response to observed data; for instance, the online data‐driven DRO literature studies how ambiguity shrinks via learning and adapts over time~\citep{aigner2023data}. Another relevant thread is variable-sized uncertainty in robust optimization, where the uncertainty “radius” itself is endogenous or evolves \citep{chassein2018variable}. These dynamic, time-sensitive robust paradigms are closer in spirit to our setting, though none (to our knowledge) consider adversarial changes in the uncertainty set across time, and therefore none is mathematically relevant.

Turning to robust mechanism design, this topic has also been studied extensively in computer science \citep[e.g.,][]{GHW-01,DHKN-11}, economics \citep[e.g.,][]{HS-78,Fra-14,Car-15} and operations research literature \citep[e.g.,][]{PR-08,CLL-17}. The goal is to design mechanisms to achieve worst-case performance guarantees under incomplete information about the problem instances. Evaluating the performance of mechanisms via their worst-case payoff (cost) is a classic approach used in the literature \citep{Fra-14,Car-15,BB-14,BGLT-19,DRT-19,BD-21,BCTYZ-22}. Other objectives include regret \citep{CLL-17,BTXZ-22,GS-23} and competitive ratio \citep{HS-78,GHW-01}. Similarly to our work and robust optimization, many times the ambiguity sets are modeled by uncertainty intervals\citep{BS-08,BB-14,CLL-17,WLZ-20,BCW-22,ABB-23}, but in contrast to our work they are static and given to the decision maker upfront.

  

 

%% file: Paper/numerical.tex
\label{sec:numerical}

To provide numerical justifications for the performance of our proposed algorithms, we conducted numerical experiments on synthetic data. In particular, while we focused on adversarial predictions and demand in our theoretical analysis, to empirically evaluate the performance of our algorithms in practical scenarios beyond worst-case, we have considered stochastic generative models for predictions and demands in this section. In the following, we elaborate on the details of these experiments---in particular, the exact setup, generative processes, and policies we consider in each experiment---and report the results.

\subsection{Experiment I: Predictions from Unbiased Samples}
\label{apx:numerical:predictions from samples}
\xhdr{Experimental setup.} We generate the following staffing instance motivated by our primary application in last-mile delivery operations (see \Cref{sec:intro} and \Cref{fig:tradeoff}): There is a single delivery station with unknown demand $\demand$ on operating day $T + 1 \triangleq 15$. There are two workforce supply pools $i\in\{1, 2\}$, each with size $\supplyi = 20$ initially.\footnote{\vrevcolor{Note that we take a large market perspective, meaning that the actual number of initial workers is $s_i N$ for pool $i$, when $N\to +\infty$ is the scaling parameter of the market. We use the same scaling for the demand and the number of hires from each pool. For simplicity of exposition, we only work with fractional quantities such as $d$ and $s_i$---which are the results of proper normalization of the original quantities by the large market scale parameter $N$.}} We refer to the first pool of workers and the second workforce pool as ``fixed workers'' and ``ready workers,'' respectively. The fixed workers correspond to the full-time  employees, whose work schedule need to be determined ten days before the operating day, that is, they become non-available since day 5 and thus $\wdiscount_{1t}\triangleq \indicator{t \leq 4}$ for every $t\in[14]$. Meanwhile, the ready workers correspond to the gig economy workers who have a S-shaped availability rate defined as $\wdiscount_{2t} \triangleq \frac{1}{1 + \exp(t - 9)}$ for every $t\in[14]$. See \Cref{fig:numerical:availability rate illustration} for an illustration of the constructed availability rate $\{\wdiscountit\}_{i\in[2],t\in[14]}$.
Both the per-unit overstaffing cost and understaffing cost are set as $\overcost = \undercost \triangleq 1$.

Next, we describe the stochastic process that generates demand $\demand$ and prediction sequences~$\predictions$.   The total demand $\demand$ can be divided into $\dailydemands = \{\dailydemandt\}_{t\in[T]}$ such that $\sum\nolimits_{t\in[T]}\dailydemandt = \demand$. Here $\dailydemandt$ is the partial demand that the platform observes in each day $t\in[T]$. In particular, each partial demand $\dailydemandt$ is realized from binomial distribution $\texttt{binom}(5, \priort)$ independently, where success probability (aka., prior) $\priort$ is itself realized from uniform distribution $\texttt{unif}(0, 0.5)$ independently. We refer to the sequence of priors $\priors = \{\priort\}_{t\in[T]}$ as the prior profile. The platform does not observe $\{\priort\}_{t\in[T]}$ throughout the planning horizon, and only observes partial demand $\dailydemandt$ at the beginning of each day $t\in[T]$. \vrevcolor{In addition, the platform has access to extra sample trajectories of (future) partial demands as signals in each day. 
Specifically, in each day $t\in[T]$, the platform receives a 
sampled partial demand profile $\dailydemandsamplest = \{\dailydemandsampletk\}_{k\in[t+1:T]}$ as an extra signal,
where each signal $\dailydemandsampletk$ is drawn 
 from binomial distribution $\texttt{binom}(5, \priork)$ independently, i.e., sampled partial demand signal $\dailydemandsampletk$ and actual partial demand $\dailydemand_k$ are i.i.d.\ generated.\footnote{\vrevcolor{Note that in the large market perspective, the (partial) demand and supply in the large market---before normalization by the scale parameter $N$---are integers. Hence, it makes sense to consider the same small discretization error of $\frac{1}{N}$ for both supply and demand---and so demand could be fractional in the limit when $N\to+\infty$. Here, we pick a coarser distribution over integers $0,\ldots, 5$ for the (normalized) partial demand, mostly for the purpose of tractability and to reduce computational difficulties of computing optimal online policies in our numerical simulations (e.g., as if the backend is using demand batching, and hence each partial demand only takes values $kN$ for $k\in[0:5]$ after scaling back to the actual numbers).}}}

In this setup, sampled partial demands $\dailydemandsamplest = \{\dailydemandsampletk\}_{k\in[t+1:T]}$ play the role of external input data for future forecasts.  The platform constructs prediction intervals $\predictions=\{[\pLt,\pRt]\}_{t\in[T]}$ using the observed partial demands and sampled partial demand profiles. Specifically, for each day $t\in[T]$, given observed partial demands $\{\dailydemand_k\}_{k\in[t]}$ and sampled partial demand profiles $\{\dailydemandsamples^{(\tau)}\}_{\tau\in[t]}$  
in the first $t$ days, the platform computes a point estimator $\estimateddemand^{(t)}$ as following:
\begin{align}
\label{eq:numerical point estimator}
    \estimateddemand^{(t)} \triangleq \sum\nolimits_{k\in[t]}\dailydemand_k + \frac{1}{t}\sum\nolimits_{\tau\in[t]}\sum\nolimits_{k\in[t+1:T]}\dailydemandsample^{(\tau)}_k
\end{align}
As a sanity check, given any prior profile $\priors$, unknown demand $\demand$ and point estimator $\estimateddemand^{(t)}$ have the same expectation, i.e., $\expect{\demand\condition \priors} = \expect{\estimateddemand^{(t)}\condition\priors}$ for every $t\in[T]$---therefore this point estimator is unbiased. Given the unbiased point estimator $\estimateddemandk$ for day $t$, the platform constructs the prediction interval $[\pL_t,\pR_t]$ as
\begin{align}
\label{eq:numerical prediction interval}
    \pL_t \triangleq \estimateddemand^{(t)} - l_t \;\;\mbox{and}\;\; \pR_t \triangleq \estimateddemand^{(t)} + r_t
\end{align} 
Here, $l_t$ and $r_t$ are constants such that the prediction intervals have $5\%$ miscoverage, that is, 
\begin{align*}
    \prob{
    \frac{1}{t}\sum\nolimits_{\tau\in[t]}\sum\nolimits_{k\in[t+1:T]}\dailydemandsample^{(\tau)}_k - l_t
    \leq 
    \sum\nolimits_{k\in[t+1:T]}\dailydemandt \leq \frac{1}{t}\sum\nolimits_{\tau\in[t]}\sum\nolimits_{k\in[t+1:T]}\dailydemandsample^{(\tau)}_k + r_t} = 95\%
\end{align*}
where the average is taken over the randomness in prior profile $\priors$, partial demand profile $\dailydemands$, and sampled partial demand profiles $\{\dailydemandsamples^{(\tau)}\}_{\tau\in[t]}$. See \Cref{fig:numerical:prediction error illustration} for an illustration of the constructed prediction errors $\{\perror_t\}_{t\in[14]}$.

\begin{figure}
    \centering
    \subfloat[Availability rates $\{\wdiscountit\}_{i\in[2],t\in[14]}$.]{
    \input{Paper/figs/fig-numerical-availability-curve}
    \label{fig:numerical:availability rate illustration}
    }
    \subfloat[Prediction errors $\{\perror_t\}_{t\in[14]}$]{
    \input{Paper/figs/fig-numerical-prediction-error}
        \label{fig:numerical:prediction error illustration}
    }

    \caption{Graphical illustration of the experiment setup in \Cref{apx:numerical:predictions from samples}.}
    \label{fig:numerical setup illustration}
    \vspace{-2mm}
\end{figure}

\smallskip
\noindent
The timeline of the demand and prediction sequence generating process below.
\begin{itemize}
    \item On day 0: prior probabilities $\{\prior_t\}_{t\in[T]}$ are realized i.i.d.\ from uniform distribution $\texttt{U}(0, 0.5)$. Prior profile $\priors$ remains unknown to the platform throughout the entire decision making  horizon.
    \item On each day $t\in[T]$:
    \begin{itemize}
        \item partial demand $\dailydemand_t$ is revealed to the platform from binomial distribution $\texttt{binom}(5, \priort)$,
        \item sampled partial demand profile $\dailydemandsamplest = \{\dailydemandsampletk\}_{k\in[t+1:T]}$ is given to the platform, where $\dailydemandsampletk$ is drawn  from binomial distribution $\texttt{binom}(5, \priork)$ independently,
        
        \item point estimator $\estimateddemand^{(t)}$ defined in eqn.~\eqref{eq:numerical point estimator} and prediction interval $[\pL_t,\pR_t]$ defined in eqn.~\eqref{eq:numerical prediction interval} are constructed correspondingly.
    \end{itemize}
\end{itemize}

\xhdr{Policies.} In this numerical experiments, we compare the staffing cost of five different policies:
\begin{enumerate}[label=(\roman*)]
    \item {\FMNN}: this policy is {\OPTSim}~(\Cref{alg:opt}). It solves \ref{eq:opt reduced form} and then invoke Procedure~\ref{alg:emulator} in each day. It is minimax optimal (\Cref{thm:opt alg}).
    \item {\FMNNPlus}: this policy is {\OPTReS} (\Cref{alg:lp-resolve}). In each day, it resolves \ref{eq:opt reduced form} by reviewing the remaining days as a staffing subproblem. It is minimax optimal (\Cref{thm:opt alg resolving}).
    \item {{\FMDP}/{\EMDP}}: we solve for the optimal online policy that minimizes the expected imbalance cost in this Bayesian setting---which can be formulated as a finite horizon MDP--- given either the estimated or exact transition matrix of this underlying MDP.  Specifically,  {\EMDP} constructs the empirical transition matrix of the MDP on day $t$ with the sampled partial demand profiles $\{\dailydemandsamples^{(\tau)}\}_{\tau\in[t]}$. Since the state space and action space are large, a discretization of both of these spaces is employed. We report its numerical performance under different discretization precisions. In addition, for short-horizon settings, we also evaluate the discretized ``full-information'' MDP, called {\FMDP}, which utilizes the true transition matrix rather than the empirical one estimated from the sampled partial demand profiles. Owing to its significantly higher computational requirements, this benchmark is infeasible to implement in long-horizon settings.
    \item {\AGR}: for each day $t\in[T]$, this policy computes $\hat\demand_t \triangleq \frac{\overcost\pLt + \undercost\pRt}{\overcost+\undercost}$, and hires as much as possible in the current day $t$ to meet $\hat{d}_t$. 
   Effectively, this heuristics discards the possibility of  having improved prediction intervals in the future; thus it selects the staffing level that would minimize the worst case cost. Given that supply only shrinks over time, it tries to achieve (or get as close as possible to) that staffing level in the current period.
    \item {\ASIM}: for each day $t\in[T]$, this policy finds the optimal staffing level $\hat\demand_t$ corresponding to a single-shot newsvendor problem on day $t$, where we minimize the expected imbalance cost given that the empirical demand distribution used in the expectation is constructed based on samples $\{\sum\nolimits_{k\in[t]}\dailydemand_k + \sum\nolimits_{k\in[t+1:T]}\dailydemandsample^{(\tau)}_k\}_{\tau\in[t]}$. It then hires as much as possible in the current day $t$ to meet $\hat{d}_t$. The motivation behind this heuristic is similar to the previous one, however, here we take a Bayesian approach and utilize the distributional assumption underlying the generation of the partial demand and samples (or signals).
    
\end{enumerate}
In addition to {\FMNN} and {\FMNNPlus} proposed in this work, {\EMDP} recovers the Bayesian optimal policy with full information, i.e. {\FMDP} (\vrevcolor{which is the same as {\EMDP} with infinitely many extra samples}); however, it requires significantly more computational power (see the empirical run-times reported below). On the other extreme, although {\AGR} and {\ASIM} are quite simple and computationally light (even faster than our algorithms), they are not forward-looking in their decisions (i.e., do not take into account the possibility of receiving more accurate predictions in the future).

\xhdr{Results.} We numerically evaluate the performances of all policies using Monte-Carlo simulation. Besides the parameter setup specified above (which we refer to as the long instance), we also report results from another parameter setup with $T = 5$ (which we refer to as the short instance). In the short instance with $T=5$, all other parameters are adjusted accordingly; in particular, two workforce pools have size $\supplyi \triangleq 5$, availability rates are $\wdiscount_{1t} \triangleq \indicator{t \leq 2},\wdiscount_{2t} \triangleq \frac{1}{1+\exp(t-3)}$ for every $t\in[5]$, and per-unit overstaffing cost and understaffing cost are $\overcost=\undercost\triangleq 1$. For each parameter setup, we conduct 100 iterations of the simulation and record the empirical performance of each policy. 

We first discuss our numerical finding for the long instance (with $T = 14$). In this setup, we implement two versions of {\EMDP} with two levels of discretization. We report the empirical average cost and total running time in \Cref{tab:numerical:long instance}. Among all six policies (two {\EMDP} with different discretization levels), our proposed {\FMNNPlus} attains the best (smallest) average cost of 1.619. Compared with the {\EMDP} (with high precision), which attains the second best cost of 1.710, it has 5.6\% cost reduction. Moreover, {\FMNNPlus}'s running time (0.644 second) is 18194 times faster than {\EMDP} with high precision (11716.999 seconds). Compared with {\EMDP} with low precision, {\FMNNPlus} receives 26.6\% cost reduction and is 33.8 times faster. Compared with the other two naive policies, {\FMNNPlus} receives 71.5\% and 123.9\% cost reduction. Besides {\FMNNPlus} which resolves program~\ref{eq:opt reduced form} each day, our proposed {\FMNN} that solves the program~\ref{eq:opt reduced form} once, achieves the third best cost of 2.040. It also beats {\EMDP} with low precision, and other two naive policies (with significant 36.1\% and 77.7\% cost reductions). Comparing {\FMNN} and {\FMNNPlus} which are both minimax optimal, although the former algorithm suffers a higher cost, it is 3.5 times faster than the latter algorithm.

\begin{table}
    \centering
    \begin{tabular}{c|ccc}
        & {\FMNN} & {\FMNNPlus} & {{\EMDP}} \\
        & & & (high precision) 
        \\
        \hline
        average cost   & 2.040 & 1.619 & 1.710 \\
        \hline
        running time   & 0.186 & 0.644 & 11716.999 \\
        \multicolumn{4}{c}{}
        \\
        & {\AGR} & {\ASIM} & {\EMDP} \\
        & & & (low precision) \\
        \hline
        average cost   & 2.776 & 3.625 & 2.050 \\
        \hline
        running time   & 0.010 & 0.110 & 21.801 \\
    \end{tabular}
    \caption{Comparing average cost and total running time of different policies for parameter setup with $T = 14$. See the experiment setup in \Cref{apx:numerical:predictions from samples}.}
    \label{tab:numerical:long instance}
\end{table}

\begin{table}
    \centering
    \begin{tabular}{c|ccc}
        & {\FMNN} & {\FMNNPlus} & {{\EMDP}} \\
        \hline
        average cost   & 1.154 & 1.003 & 0.890 \\
        \hline
        running time   & 0.072 & 0.111 & 32.155 \\
        \multicolumn{4}{c}{}
        \\
        & {\AGR} & {\ASIM} & {{\FMDP}} \\
        \cline{1-4}
        average cost   & 1.549 & 1.857 & 0.740\\
        \cline{1-4}
        running time   & 0.004 & 0.012 & 3515.48\\
    \end{tabular}
    \caption{Comparing average cost and total running time of different policies for parameter setup with $T = 5$. See the experiment setup in \Cref{apx:numerical:predictions from samples}.}
    \label{tab:numerical:short instance}
\end{table}

We next discuss our numerical finding for the short instance (with $T = 5$). In this setup, we implement {\EMDP} using a sufficiently fine-grained discretization. We report the empirical average cost and total running time in \Cref{tab:numerical:short instance}. Among the six policies evaluated (including {\EMDP}, which uses sampled partial demand profiles to construct an empirical transition matrix, and {\FMDP}, which employs the true transition matrix), our proposed {\FMNNPlus} and {\FMNN} attain the third and fourth best average cost of 1.003 and 1.154, respectively. 
The {\FMDP} attains the lowest cost of 0.740 but is hypothetical as it requires direct access to the true transition matrix. Relative to {\FMDP}, {\FMNNPlus} and {\FMNN} incur 35.5\% and 55.9\% higher costs, but execute 31670 and 48826 times faster, respectively. Compared with {\EMDP}, which achieves the second-best cost of 0.890, our policies incur 12.7\% and 29.6\% higher costs, while running 289 and 446 times faster, respectively. Both {\FMNNPlus} and {\FMNN} also outperform the two native policies, yielding cost reductions of 54.4\% and 85.1\% for {\FMNNPlus}, and 34.2\% and 60.9\% for {\FMNN}. Compared with the long instance with $T = 14$, the cost reduction in the short instance with $T = 5$ becomes smaller. This aligns with the fact that the efficiency loss due to the lack of lookahead becomes less significant as the planning horizon shortens.

\subsection{Experiment II: Predictions from Multiple Machine Learning Models}

\label{apx:numerical:predictions from ML models}

\xhdr{Experimental setup.} 
The second setup mirrors that in \Cref{apx:numerical:predictions from samples}, except for the construction of the predictions. Specifically, all elements of the model, such as the size and availability rate of workforce pools, as well as the data generating process for both partial and final demand, remain unchanged between the two setups. The major difference lies in how the predictions are constructed. In this setting, predictions are generated using three machine learning models: linear regression ({\OLS}), ridge regression ({\Ridge}), and random forest ({\RF}).

For each day $t \in [T]$, given the revealed partial demands $\{\dailydemand_k\}_{k \in [t]}$, each model treats this partial demand profile as a $t$-dimensional feature vector and outputs a point estimate of the total demand. We denote these estimates by $\OLSt$, $\Ridget$, and $\RFt$, respectively. Each machine learning model is trained on 200 sample trajectories generated from the true data-generating process.

Based on the point estimates $\OLSt$, $\Ridget$, and $\RFt$ for day $t$, we consider the following two heuristics for constructing prediction intervals:
\begin{enumerate}
\item \textbf{Unweighted Average:} The platform first computes the unweighted average of the three estimates and then constructs the prediction interval $[\pL_k, \pR_k]$ as
\begin{align*}
\pL_t &\triangleq \frac{1}{3} \left(\OLSt + \Ridget + \RFt \right) - l_t, \;\;\mbox{and}\;\;
\pR_t \triangleq \frac{1}{3} \left(\OLSt + \Ridget + \RFt \right) + r_t,
\end{align*}
where $l_t$ and $r_t$ are constants chosen such that
\begin{align*}
\prob{
\frac{1}{3} \left(\OLSt + \Ridget + \RFt \right) - l_t
\leq \demand \leq
\frac{1}{3} \left(\OLSt + \Ridget + \RFt \right) + r_t
} = 95\%,
\end{align*}
with the probability taken over the randomness in the prior profile $\priors$, the partial demand profile $\dailydemands$, and in the machine learning models.
\item \textbf{Weighted Average:} The platform selects ex-ante weights $\OLSweight$, $\Ridgeweight$, and $\RFweight$ from the set $\{0, 0.2, 0.4, 0.6, 0.8, 1\}$, subject to the constraint $\OLSweight + \Ridgeweight + \RFweight = 1$. It then computes the weighted average of the three estimates and constructs the prediction interval $[\pL_t, \pR_t]$ as
\begin{align*}
    \pL_t &\triangleq \OLSweight \cdot \OLSt + \Ridgeweight \cdot \Ridget + \RFweight \cdot \RFt - l_t, \\
    \pR_t &\triangleq \OLSweight \cdot \OLSt + \Ridgeweight \cdot \Ridget + \RFweight \cdot \RFt + r_t,
\end{align*}
where $l_t$ and $r_t$ are constants satisfying
\begin{align*}
    \prob{
        \OLSweight \cdot \OLSt + \Ridgeweight \cdot \Ridget + \RFweight \cdot \RFt - l_t 
        \leq \demand \leq 
        \OLSweight \cdot \OLSt + \Ridgeweight \cdot \Ridget + \RFweight \cdot \RFt + r_t
    } = 95\%.
\end{align*}
\end{enumerate}
Below, we report the numerical performance of both the unweighted-average and weighted-average constructions. For the weighted-average construction, we first identify the ex ante weight combination that achieves the best empirical performance for each policy,\footnote{\vrevcolor{Note that fixing a policy, one can always conduct off-policy evaluations using Monte-Carlo simulations, and hence find the best set of weights to aggregate the ML predictions for a given policy.}} and then construct the prediction intervals using this selected weight profile.

\xhdr{Policies.} In this numerical experiment, we compare the staffing costs of four different policies: {\FMNN}, {\FMNNPlus}, {\AGR}, and {\ASIM}. The first three policies are defined as in \Cref{apx:numerical:predictions from samples}. The policy {\ASIM} aggregates the three point estimates and then solves a single-shot newsvendor problem, assuming the demand distribution to be Gaussian with mean equal to the aggregated estimate and standard deviation equal to one.\footnote{In the absence of sample access (\Cref{apx:numerical:predictions from samples}), one could alternatively consider an empirical MDP, where the state includes the entire history--i.e., all revealed partial demands and the estimates from the three machine learning models. However, this MDP is significantly higher-dimensional than the one used in \Cref{apx:numerical:predictions from samples}. Due to computational constraints, we do not implement this alternative in the current experiment.}

\xhdr{Results.} We numerically evaluate the performances of all policies using Monte-Carlo simulation. For each parameter setup, we conduct 100 iterations of the simulation and record the empirical performance of each policy. All our results are reported in \Cref{tab:numerical:long instance 3ML,tab:numerical:small instance 3ML} for the long instance ($T = 14$) and the short instance ($T = 5$), respectively.

\begin{table}
    \centering
    \begin{tabular}{c|cccc}
        & {\FMNN} & {\FMNNPlus} & {\AGR} & {\ASIM} \\
        \hline
        unweighted & 1.166 & 0.917 & 3.135 & 2.885 \\
        \hline
        weighted   & 0.983 & 0.548 & 2.638 & 2.760 \\
    \end{tabular}
    \vspace{1mm}
    \caption{Comparing average cost of different policies for parameter setup with $T = 14$. The first row corresponds to the scenario where predictions are generated by the (unweighted) average of estimators from three machine learning models. The second row corresponds to the scenario where predictions are generated by the weighted average of estimators from three machine learning models.
    The empirically optimal ex ante weights used in {\FMNN}, {\FMNNPlus}, {\AGR}, {\ASIM} are (0.4, 0.4, 0.2), (0.2, 0.8, 0), (1, 0, 0), (0, 0.8, 0.2), respectively. Note that here we consider different ex-ante weights, to capture different ways of aggregating our 3 ML algorithms.
    }
    \label{tab:numerical:long instance 3ML}
\end{table}

We first discuss our numerical finding for the long instance (with $T = 14$) in \Cref{tab:numerical:long instance 3ML}. Among the four policies evaluated, our proposed policies, {\FMNN} and {\FMNNPlus}, consistently outperform {\AGR} and {\ASIM} by a significant margin under both the unweighted- and weighted-average prediction constructions. Consistent with the results in \Cref{apx:numerical:predictions from samples}, {\FMNNPlus} achieves the lowest average cost across both settings. Specifically, under the unweighted-average construction, the average cost of {\FMNNPlus} is 3.41 and 3.14 times lower than that of {\AGR} and {\ASIM}, respectively. Under the weighted-average construction, the reductions are even more substantial: {\FMNNPlus} yields a cost that is 4.81 and 5.04 times lower than {\AGR} and {\ASIM}, respectively. Our {\FMNN} policy also achieves strong performance, with costs more than 2.4 times lower than those of {\AGR} and {\ASIM} in both constructions.

\begin{table}
\centering
    \begin{tabular}{c|cccc}
        & {\FMNN} & {\FMNNPlus} & {\AGR} & {\ASIM} \\
        \hline
        unweighted & 1.415 & 1.371 & 1.873 & 1.727 \\
        \hline
        weighted   & 1.336 & 1.291 & 1.691 & 1.685 \\
    \end{tabular}
    \caption{Comparing average cost of different policies for parameter setup with $T = 5$. The first row corresponds to the scenario where predictions are generated by the (unweighted) average of estimators from three machine learning models. The second row corresponds to the scenario where predictions are generated by the weighted average of estimators from three machine learning models.
    The empirically optimal ex ante weights used in {\FMNN}, {\FMNNPlus}, {\AGR}, {\ASIM} are (0.8, 0.2, 0.2), (0.8, 0.2, 0), (0.8, 0.2, 0), (0.2, 0.8, 0), respectively. \vrevcolor{Note that here we consider different ex-ante weights, to capture different ways of aggregating our 3 ML algorithms.}
    See the experiment setup in \Cref{apx:numerical:predictions from ML models}. }
    \label{tab:numerical:small instance 3ML}
\end{table}

We now discuss our numerical finding for the short instance (with $T = 5$) in \Cref{tab:numerical:small instance 3ML}. The qualitative observations from the long instance continue to hold: {\FMNNPlus} consistently achieves the lowest average cost under both prediction constructions. The performance gap between {\FMNNPlus} and the second-best policy, {\FMNN}, is relatively small, i.e., less than 3.5\%. In contrast, {\FMNNPlus} yields substantial cost reductions compared to {\AGR} and {\ASIM}: under the unweighted-average construction, the reductions are 36.6\% and 26.0\%, respectively, while under the weighted-average construction, the reductions are 30.1\% and 30.0\%, respectively.

%% file: Paper/figs/fig-numerical-availability-curve.tex
\begin{tikzpicture}
  \begin{axis}[
    width=8cm,
    height=6cm,
    xlabel={day $t$},
    ylabel={availability rates $\{\wdiscountit\}_{i\in[2]}$},
    xmin=1, xmax=14,
    ymin=0, ymax=1.1,
    grid=major,
    xtick={1,...,14},
    legend style={at={(0.5,1.05)}, anchor=south, legend columns=-1},
    thick,
    mark size=2pt,
  ]

    \addplot[
      blue,
      mark=*,
      smooth
    ] coordinates {
      (1,1.0)
      (2,1.0)
      (3,1.0)
      (4,1.0)
    };

    \addplot[
      blue,
      densely dotted
    ] coordinates {
      (4,1.0)
      (4,0.0)
    };

    \addplot[
      blue,
      mark=*,
      smooth
    ] coordinates {
      (4,0.0)
      (5,0.0)
      (6,0.0)
      (7,0.0)
      (8,0.0)
      (9,0.0)
      (10,0.0)
      (11,0.0)
      (12,0.0)
      (13,0.0)
      (14,0.0)
    };
    \addlegendentry{fixed workers}

    \addplot[
      red,
      mark=square*,
      smooth,
      dashed
    ] coordinates {
      (1,0.9996646498695336)
      (2,0.9990889488055994)
      (3,0.9975273768433653)
      (4,0.9933071490757153)
      (5,0.9820137900379085)
      (6,0.9525741268224334)
      (7,0.8807970779778823)
      (8,0.7310585786300049)
      (9,0.5)
      (10,0.2689414213699951)
      (11,0.11920292202211755)
      (12,0.04742587317756678)
      (13,0.01798620996209156)
      (14,0.0066928509242848554)
    };
    \addlegendentry{ready workers}

  \end{axis}
\end{tikzpicture}

%% file: Paper/figs/fig-numerical-prediction-error.tex
\begin{tikzpicture}
  \begin{axis}[
    width=8cm,
    height=6cm,
    xlabel={day $t$},
    ylabel={prediction error $\perror_t$},
    xmin=1, xmax=14,
    ymin=0,
    grid=major,
    xtick={1,...,14},
    smooth,
    thick,
    mark=*,
    mark size=2pt,
  ]
    \addplot coordinates {
      (1,16.5)
      (2,14.333333333333332)
      (3,13.5)
      (4,12.399999999999999)
      (5,11.5)
      (6,10.857142857142858)
      (7,10.125)
      (8,9.444444444444445)
      (9,8.6)
      (10,7.2727272727272725)
      (11,6.666666666666667)
      (12,5.153846153846153)
      (13,3.7857142857142856)
      (14,0.0)
    };
  \end{axis}
\end{tikzpicture}

%% file: Paper/apx-miscoverage.tex
\label{apx:prob-miscoverage-shocks}
In some practical scenarios, there could be temporary anomalies or shocks in the underlying forecasting method on a certain day, leading to  miscoverage of the final target demand in the prediction interval of that day. In this section, we discuss how the presence of such probabilistic miscoverage shocks, formally defined below, affects the cost of {\OPTSim} (\Cref{alg:opt}). 

To simplify the presentation, we fix $\undercost = \overcost = 1, \pLzero = 0, \pRzero = 1$, and consider only a single pool. We also assume that $\pbias_t = 0$ for every $t\in[T]$; therefore, in the ideal scenario with no miscoverage of the final target demand,  there will be predictions $[\pLt,\pRt]\}_{t\in[T]}$ received sequentially by the algorithm that are perfectly consistent. Without loss of generality, we can further assume (i) the ideal prediction intervals $\{[\pLt,\pRt]\}_{t\in[T]}$ are nested, i.e., $[\pL_{t+1},\pR_{t+1}]\subseteq[\pLt,\pRt], t\in[0:T-1]$, and (ii) prediction error upper bound $\perrort$ is weakly decreasing over time $t$ and is smaller than the initial range of demand, i.e., $\perror_t \leq 1$. Our analysis can be immediately extended to general settings without these assumptions, as long as cost functions remain bounded and Lipschitz.

\smallskip
\noindent

\xhdr{Prediction sequence with probabilistic miscoverage.}
We formalize the model with probabilistic miscoverage shocks of probability $\pprob\in[0,1]$ as follows:
\begin{itemize}
    \item The adversary first decides the final target demand $\demand$ and a perfectly consistent sequence of prediction intervals $\{[\pLt,\pRt]\}_{t\in[T]}$ on day 0 satisfying the prediction error upper bounds $\{\perror_t\}_{t\in[T]}$.
    \item On each day $t\in[T]$:
    \begin{itemize}
        \item With probability $1 - \pprob$, (consistent) prediction $[\pLt,\pRt]$ is revealed to the decision maker.
        \item With the remaining probability $\pprob$, a \emph{``bad event''} happens: adversary decides a (possibly inconsistent) prediction $[\pLt',\pRt']$ and reveals it to the decision maker.
    \end{itemize}
\end{itemize}
In this model, we allow the bad events from different days to be correlated. We also do not make any assumptions on the predictions $[\pLt',\pRt']$ in the case of a bad event. 
Now we consider the following two scenarios:
\begin{enumerate}
    \item \emph{Detection-before-hiring}: On each day $t$, if the bad event occurs, the decision maker detects it before making her staffing decision on that day.
    \item \emph{No-detection}: the decision maker never detects whether bad events happen.
\end{enumerate}

Let $\optcost$ be the optimal minimax cost when bad events never happen. We analyze the additive difference between the expected cost of {\OPTSim} (\Cref{alg:opt}) (with possibly modifications based on the bad-event feedback) and $\optcost$, which we denote by $\mathcal{E}$. Since {\OPTSim} is minimax optimal when there is no probabilistic miscoverage shocks, $\mathcal{E}$ can be interpreted as the extra additive error in expected cost in the presence of probabilistic miscoverage shocks. In the following, we sketch the analysis and provide intuitions on how to bound $\mathcal{E}$.

\xhdr{Detection-before-hiring scenario.} In this setting, consider running {\OPTSim} with the following modification: During days where bad events occur---which could be detected at the beginning of the day, before the hiring decision---no workers should be hired. On the other days where bad events do not happen, invoke \textsc{Emulator Oracle} (Procedure~\ref{alg:emulator}) assuming that \emph{correct} staffing decisions have been made on all previous days, given the knowledge of prediction intervals of days without bad events. More precisely, for each day $t\in[T]$ without a bad event, we consider the following modified history: we use the prediction interval $[\pL_\tau,\pR_\tau]$ for each day $\tau <t$ without a bad event, and we use the prediction interval  $[\pL_{\tau'},\pR_{\tau'}]$ for each day $\tau<t$ with a bad event, where  $\tau' \in (\tau, t]$  is the first day after $\tau$ that has no bad event. We then run {\OPTSim} on this modified history to identify the correct staffing decisions of the previous days, to be used to make the staffing decision $\alloc_t$ of day $t$ as described. 

With this implementation, we claim that $\mathcal{E} = O(\pprob)$.
Let $\tilde{\xbf} = \{\tilde{\alloc}_t\}_{t\in[T]}$ be a hypothetical (randomized) staffing profile generated by {\OPTSim} under a particularly modified sequence of prediction intervals. Specifically, similar to the way we modify history in our implementation, we modify the consistent prediction sequence $\{[\pLt,\pRt]\}_{t\in[T]}$ (where no bad event occurs) by replacing the prediction interval $[\pL_\tau,\pR_\tau]$ for each day $\tau\in[T]$ with a bad event with $[\pL_{\tau'},\pR_{\tau'}]$, where $\tau'\in (\tau,T]$ is the first day after $\tau$ that no bad event occurs (if no such day exists, we use the zero-length interval $\{\demand\}$ as the interval, where $\demand$ is the target demand). First, note that this modified sequence is perfectly consistent and only reveals more information about demand at each time than $\{[\pLt,\pRt]\}_{t\in[T]}$. Therefore, using a simple inductive argument, we can show that the cost of {\OPTSim} under the modified prediction sequence is weakly smaller than the cost of {\OPTSim} under $\{[\pLt,\pRt]\}_{t\in[T]}$. As a result, the expected worst-case cost of {\OPTSim} under the modified sequence is weakly smaller than $\optcost$. Now, let $\xbf = \{\alloc_t\}_{t\in[T]}$ be the (randomized) staffing profile generated by our algorithm and under the realized prediction sequence (with possible bad events). By construction, $\alloc_t$ is either $\tilde{\alloc}_t$ or 0.
In particular $\alloc_t = 0$ only if a bad event occurs on day $t$. Therefore, the expected staffing level $\texpect{\sum_{t\in[T]}\alloc_t}$ can be upper bounded by $\sum_{t\in[T]} \tilde{\alloc}_t $ and lower bounded by $\sum_{t\in[T]}(1-\pprob)\cdot \tilde{\alloc}_t$. Invoking the (piece-wise) linearity of the staffing cost function and the boundedness of its slope by a constant, we conclude that there is an extra additive error of $O(\pprob)$ in expected cost of our algorithm compared to the hypothetical sequence $\{\tilde{\alloc}_t\}$, and therefore an extra additive error of $\mathcal{E}=O(\pprob)$ in the worst-case expected cost of our algorithm versus $\optcost$ (here, the expectation is over the randomness in miscoverage shocks). 


\xhdr{No-detection scenario.} In this setting, consider running {\OPTSim} directly. We now claim that $\mathcal{E} = \mathcal{O}(T\cdot \pprob)$. By the union bound, the probability that no bad event occurs in all $T$ days is at least $1 - T \cdot \pprob$. Invoking the (piece-wise) linearity of the staffing cost function and the boundedness of its slope by a constant, we conclude that the extra additive error in expected cost is no more than $\mathcal{E} = \mathcal{O}(T\cdot\pprob)$ (here, again the expectation is over the randomness in miscoverage shocks). 

%% file: Paper/apx-combined-objective.tex
\label{apx:combinded objective}
In this section, we consider a variant where the goal of the platform is to minimize the staffing cost plus the hiring cost. Specifically, given staffing profile $\xbf$ and demand~$\demand$, the total cost $\CostTotal[\demand]{\xbf}$ is defined as:
\begin{align*}
    \CostTotal[\demand]{\xbf} \triangleq \underbrace{\Cost[\demand]{\xbf}}_{\text{staffing cost}}
    +
    \underbrace{\sum\nolimits_{i\in[n]}\sum\nolimits_{t\in[T]}\priceit\allocit}_{\text{hiring cost}}
\end{align*}
where $\priceit$ is the per-worker hiring fee for the platform to hire an available worker from pool $i$ in day $t$. In this section, we impose no assumption on $\{\priceit\}_{i\in[n],t\in[T]}$, except requiring them to be non-negative. To simplify the presentation, we impose perfect consistency (i.e., $\pbiass = \zerobf$) on the prediction intervals and assume them to be nested throughout this section. All our results can be extended under $\pbiass$-consistency straightforwardly.

Our approach developed in the main text can be easily applied to this variant. Consider the following linear program~\ref{eq:opt variant cost}:
\begin{align}
\tag{${\textsc{LP-joint-cost}}$}
    \label{eq:opt variant cost}
    &\arraycolsep=1.4pt\def\arraystretch{1.2}
    \begin{array}{llll}
    \min\limits_{\substack{\xbf,\lambdabf,\rgap\geq \zerobf}}
    \qquad
    &
    \displaystyle\left(\undercost\cdot\rgap + \sum\nolimits_{i\in[n]}\sum\nolimits_{t\in[T]}\priceit\allocit\right)
    \vee
    \left(\max\limits_{k\in[T]}\overcost\cdot\lgap_k + \sum\nolimits_{i\in[n]}\sum\nolimits_{t\in[k]}\priceit\allocit\right)\quad
    & 
    \text{s.t.}
    \\
    &
    \displaystyle\sum\nolimits_{t\in[T]}\frac{1}{\wdiscountit}\allocit \leq \supplyi
    &
    i\in[n]
    \\
    &
    \displaystyle\sum\nolimits_{i\in[n]}\sum\nolimits_{t\in[k]}
    \allocit \leq \pRzero - \perrork + \lgap_k
    \quad
    &
    k\in[T]
    \\
    &
    \displaystyle\sum\nolimits_{i\in[n]}\sum\nolimits_{t\in[T]}
    \allocit \geq \pRzero - \rgap
    &
\end{array}
\end{align}

Now we present the minimax optimal algorithm (\Cref{alg:opt variant cost}) with its optimality guarantee (\Cref{thm:opt variant cost}). Following the same discussion in \Cref{sec:base model result}, it can be verified that {\OPTSimTilde} is feasible and has polynomial running time. The proof of \Cref{thm:opt variant cost} follows almost the same argument as the proof for \Cref{thm:opt alg} in the base model and is included for completeness.

\begin{algorithm}
\setcounter{AlgoLine}{0}
    \SetKwInOut{Input}{input}
    \SetKwInOut{Output}{output}
    \Input{initial pool sizes $\{\supplyi\}_{i\in[n]}$, availability rates $\{\wdiscountit\}_{i\in[n], t\in[T]}$, initial demand range $[\pLzero,\pRzero]$, prediction error upper bounds $\{\perrort\}_{t\in[T]}$, per-worker wages $\{\priceit\}_{i\in[n],t\in[T]}$ }
    \Output{staffing profile $\xbf$}
    \vspace{2mm}
    find an optimal solution $(\xbf^*, \lgap^*, \rgap^*)$ of program~\ref{eq:opt variant cost}
    
    \vspace{2mm}
     invoke Procedure~\ref{alg:emulator} with canonical staffing profile $\canxbf \gets \xbf^*$ {\small\color{\commentcolor}\tcc{facing prediction sequence $\predictions$}}
    \caption{\OPTSimTilde}
    \label{alg:opt variant cost}
\end{algorithm}

\begin{theorem}
\label{thm:opt variant cost}
    {\OPTSimTilde} is minimax optimal. Furthermore, its optimal minimax cost $\optcost$ is equal to the objective value of program~\ref{eq:opt variant cost}.
\end{theorem}

\begin{proof}
    We first show program~\ref{eq:opt variant cost} lower bounds the optimal minimax cost $\optcost$.
    In this argument, we construct a feasible solution of program~\ref{eq:opt variant cost}, whose objective value is equal to the cost guarantee of an arbitrary algorithm $\ALG$.
    
    Consider prediction sequence subset $\{\predictions\ked\}_{k\in[T]}$ defined in \eqref{eq:single switch prediction construction} with $\pbiass = \zerobf$. By construction, for every $t \in[T]$, the first $t$ predictions from day 1 to day $t$ are the same for all prediction sequence $\predictions\ked$ with $k\geq t$. Therefore, algorithm $\ALG$'s (possibly randomized) staffing decision in each day $t$ should be the same under all prediction sequence $\predictions\ked$ with $k \geq t$.
    
    Motivated by the prediction sequence construction above, we let random variable $\randomallocit$ be the number of workers hired by the algorithm from pool $i$ in day $t$ under prediction sequence $\predictions\Ted$.
    Due to the feasibility of the algorithm under prediction sequence $\predictions\Ted$, for all sample paths (over the randomness of the algorithm), we have
    \begin{align*}
    \forall i\in[n]:\quad &
            \displaystyle\sum\nolimits_{t\in[T]}\frac{1}{\wdiscountit}\randomallocit \leq \supplyi
    \end{align*}
    Now consider the following solution $(\xbf,\lambdabf,\rgap)$ construction:
    \begin{align*}
        i\in[n],t\in[T]:&
        \qquad
        \allocit \gets \expect{\randomallocit}
        \\
        k\in[T]:
        &\qquad
        \lgap_k \gets
        \plus{
        \expect{\sum\nolimits_{i\in[n]}\sum\nolimits_{t\in[k]}\randomallocit} - \pRzero + \perrork
        }
        \\
        & \qquad \rgap \gets 
        \plus{
        \pRzero - \expect{\sum\nolimits_{i\in[n]}\sum\nolimits_{t\in[T]}\randomallocit}
        }
    \end{align*}
    By construction, all four constraints are satisfied. 
    Below we argue the objective value of the constructed solution is at most the cost guarantee of the algorithm $\ALG$ in two different cases.

    Let $k\primed = \argmax_{k\in[T]} \overcost\cdot\lgap_k + \sum\nolimits_{i\in[n]}\sum\nolimits_{t\in[k]}\priceit\allocit$.

    \xhdr{Case 1 $\left[\undercost\cdot\rgap + \sum\nolimits_{i\in[n]}\sum\nolimits_{t\in[T]}\priceit\allocit \geq 
    \overcost\cdot\lgap_{k\primed} + \sum\nolimits_{i\in[n]}\sum\nolimits_{t\in[k\primed]}\priceit\allocit\right]$:}
    In this case, the objective value of the constructed solution is 
    $\undercost\cdot\rgap + \sum\nolimits_{i\in[n]}\sum\nolimits_{t\in[T]}\priceit\allocit$.
    Consider the execution of the algorithm $\ALG$ under prediction sequence $\predictions\Ted$ and demand $\demand \triangleq \pR_T\Ted = \pRzero$. Note that the total cost can be lower bounded as 
    \begin{align*}
        \expect{\CostTotal[\demand]{\ALG(\predictions\Ted)}}
        &\overset{(a)}{\geq}
        \expect{
        \undercost\cdot
        \plus{
        \demand - 
        \sum\nolimits_{i\in[n]}
        \sum\nolimits_{t\in[T]}
        \randomallocit 
        }
        +\sum\nolimits_{i\in[n]}\sum\nolimits_{t\in[T]}
        \priceit\randomallocit
        }
        \\
        &\overset{(b)}{\geq}
        \undercost
        \cdot
        \plus{ 
        \demand - 
        \expect{\sum\nolimits_{i\in[n]}
        \sum\nolimits_{t\in[T]}
        \randomallocit}
        }
        +\expect{\sum\nolimits_{i\in[n]}\sum\nolimits_{t\in[T]}\priceit\randomallocit}
        \\
        &\overset{(c)}{=}
        \undercost
        \cdot
        \plus{ 
        \pRzero - 
        \expect{\sum\nolimits_{i\in[n]}
        \sum\nolimits_{t\in[T]}
        \randomallocit}
        }
        +\expect{\sum\nolimits_{i\in[n]}\sum\nolimits_{t\in[T]}\priceit\randomallocit}
        \\
        &\overset{(d)}{=}
        \undercost\cdot \rgap
        +\sum\nolimits_{i\in[n]}\sum\nolimits_{t\in[T]}\priceit\allocit
    \end{align*}
    where inequality~(a) holds by considering understaffing cost only;
    inequality~(b) holds due to the convexity of
    $\plus{\cdot}$ and Jensen's inequality;
    equality~(c) holds due to the construction of $\demand$;
    and
    equality~(d) holds due to the construction of $\rgap$ and $\allocit$.

    \xhdr{Case 2 $\left[\undercost\cdot\rgap + \sum\nolimits_{i\in[n]}\sum\nolimits_{t\in[T]}\priceit\allocit \geq 
    \overcost\cdot\lgap_{k\primed} + \sum\nolimits_{i\in[n]}\sum\nolimits_{t\in[k\primed]}\priceit\allocit\right]$:}
    In this case, the objective value of the constructed solution is 
    $\overcost\cdot\lgap_{k\primed} + \sum\nolimits_{i\in[n]}\sum\nolimits_{t\in[k\primed]}\priceit\allocit$.
    Consider the execution of the algorithm $\ALG$ under prediction sequence $\predictions^{(k\primed)}$ and demand $\demand \triangleq \pL_T^{(k\primed)} = \pRzero - \perror_{k\primed}$. Note that the staffing cost can be lower bounded as 
    \begin{align*}
        \expect{\Cost[\demand]{\ALG(\predictions^{(k\primed)})}}
        &\overset{(a)}{\geq}
        \expect{
        \overcost
        \cdot
        \plus{
        \sum\nolimits_{i\in[n]}
        \sum\nolimits_{t\in[k\primed]}
        \randomallocit 
        -
        \demand 
        }+\sum\nolimits_{i\in[n]}\sum\nolimits_{t\in[k\primed]}
        \priceit\randomallocit
        }
        \\
        &\overset{(b)}{\geq}
        \overcost
        \cdot
        \plus{ 
        \expect{\sum\nolimits_{i\in[n]}
        \sum\nolimits_{t\in[k\primed]}
        \randomallocit}
        -
        \demand 
        }+\expect{\sum\nolimits_{i\in[n]}\sum\nolimits_{t\in[k\primed]}\priceit\randomallocit}
        \\
        &\overset{(c)}{=}
        \overcost
        \cdot
        \plus{ 
        \expect{\sum\nolimits_{i\in[n]}
        \sum\nolimits_{t\in[k\primed]}
        \randomallocit}
        -
        \pRzero + \perror_{k\primed}
        }+\expect{\sum\nolimits_{i\in[n]}\sum\nolimits_{t\in[k\primed]}\priceit\randomallocit}
        \\
        &\overset{(d)}{=}
        \overcost\cdot\lgap_{k\primed} 
        +\sum\nolimits_{i\in[n]}\sum\nolimits_{t\in[k\primed]}\priceit\allocit
    \end{align*}
    where inequality~(a) holds by considering understaffing cost only and lower bounding the total number of hired workers as $\sum_{i\in[n]}\sum_{t\in[k\primed]}\randomallocit$;
    inequality~(b) holds due to the convexity of 
    $\plus{\cdot}$ and Jensen's inequality;
    equality~(c) holds due to the construction of $\demand$;
    and
    equality~(d) holds due to the construction of $\lgap_{k\primed}$ and $\allocit$.

    Next we argue that the cost guarantee of {\OPTSimTilde} is upper bounded by \ref{eq:opt variant cost}.
    Let $(\xbf^*, \lambdabf^*, \rgap^*)$ be the optimal solution of program~\ref{eq:opt variant cost} used in \OPTSimTilde.
    Let $k$ be the largest index such that $\allocit > 0$ for some pool $i$. The total hiring cost of {\OPTSimTilde} can be upper bounded by 
    \begin{align*}
        \sum\nolimits_{t\in[T]}\sum\nolimits_{i\in[n]}\priceit\allocit 
        &\overset{(a)}{=}
        \sum\nolimits_{t\in[k]}\sum\nolimits_{i\in[n]}\priceit\allocit 
        \\
        &\overset{(b)}{\leq}
        \sum\nolimits_{t\in[k]}\sum\nolimits_{i\in[n]}\priceit\allocit^* 
    \end{align*}
    where equality~(a) holds due to the definition of index $k$; and equality~(b) holds due to 
    the fact that $x_{it}\leq{x}^*_{it}$ by construction in Procedure~\ref{alg:emulator} (and its existence \Cref{lem:emulator}).

    Moreover, we can upper bound the staffing cost of {\OPTSimTilde} by
    \begin{align*}
        &\displaystyle
        \undercost\cdot\plus{\pR_T - \sum\nolimits_{i\in[n]}\sum\nolimits_{t\in[T]}\allocit}
        \vee
        \overcost\cdot\plus{\sum\nolimits_{i\in[n]}\sum\nolimits_{t\in[T]}\allocit-\pL_T}
        \\
        \overset{(a)}{\leq}{}&
        \undercost\cdot\plus{\pR_T\Ted - \sum\nolimits_{i\in[n]}\sum\nolimits_{t\in[T]}\allocit^*}
        \vee
        \overcost\cdot\plus{\sum\nolimits_{i\in[n]}\sum\nolimits_{t\in[T]}\allocit^*-\pL_T\ked}
        \\
        \overset{(b)}{\leq}{}&
        \undercost\cdot\rgap^*
        \vee
        \overcost\cdot\lgap_k^*
    \end{align*}
    where inequality~(a) holds due to the definition of index $k$ and the bounded overstaffing/understaffing cost properties of Procedure~\ref{alg:emulator} in \Cref{lem:emulator}; and inequality~(b) holds due to the construction of $\pR_T\Ted$, $\pL_T\ked$ in \eqref{eq:single switch prediction construction} and the third and fourth constraints in program~\ref{eq:opt variant cost}.

    Combining the upper bounds above for staffing cost and hiring cost, we show that the cost guarantee of {\OPTSimTilde} is upper bounded by program~\ref{eq:opt variant cost} as desired.
\end{proof}

%% file: Paper/apx-cancellation.tex
\label{sec:cancellation}
In this appendix section, we provide full details of the missing parts in \Cref{sec:extension}, and fill all the gaps in our technical argument in that section.

In the base model, the platform's staffing decision is irrevocable. In this section, we consider an extension in which the platform has a budget constraint for hiring workers. Moreover it is allowed to reverse previous hiring decisions by releasing hired workers after paying a cost, which we refer to as \emph{costly release}. 

\xhdr{The costly releasing environment.} We consider the following generalization of the base model studied in \Cref{sec:base model}. The platform has a total \emph{budget} of $\budget$ for the staffing decision. On each day~$t\in[T]$, by hiring  $\allocit\in\reals_+$ available workers from each pool~$i$, the platform needs to \emph{pay} $\allocit\cdot \priceit$ where $\priceit$ is the per-worker wages for pool $i$ on day $t$. Moreover, the platform can also \emph{release} $\revokeit\in\reals_+$ previously hired workers from each pool $i$ by paying a per-worker releasing fee $\cpriceit\in\reals_+\cup\{\infty\}$. We further assume that if a worker is hired and later released by the platform, she cannot be hired again for this operating day. We say a (joint hiring and releasing) staffing profile $\{\allocit, \revokeit\}_{i\in[n],t\in[T]}$ is \emph{feasible} if
\begin{align*}
\begin{array}{lrl}
\text{(supply-feasibility)}\qquad
    & \forall i\in[n]:
    &\qquad
    \displaystyle\sum\nolimits_{t\in[T]}\frac{1}{\wdiscountit}\allocit \leq \supplyi
    \\
    \text{(budget-feasibility)}
    & 
    &\qquad
    \displaystyle\sum\nolimits_{i\in[n]}\sum\nolimits_{t\in[T]}
    \priceit  \allocit + \cpriceit \revokeit\leq \budget
    \\
    \text{(releasing-feasibility)}\qquad
    & \forall i\in[n],k\in[T]:
    &\qquad
    \displaystyle\sum\nolimits_{t\in[k]}\revokeit \leq 
    \displaystyle\sum\nolimits_{t\in[k]}\allocit 
\end{array}
\end{align*}


We also make the following structural assumption about the per-worker releasing fees.

\begin{assumption}[Piecewise stationary releasing fees, restating \Cref{asp:limited cancellation fee}]
    There exists $\cptotal\in [1:T]$ and $0 = t_0 < t_1 < t_2 < \dots < t_\cptotal = T$ such that for every index $\ell\in[\cptotal]$, the per-worker releasing fees remain identical for each time interval $[t_{\ell - 1} + 1: t_{\ell}]$, i.e., $\forall t,t'\in[t_{\ell - 1} + 1: t_{\ell}]$ and $\forall i, i'\in[n]$, we have
        $\cpriceit\equiv \cprice_{i't'}$.
\end{assumption}
We introduce the auxiliary notation $\cintervalell\triangleq [t_{\ell - 1} + 1: t_{\ell}]$ and $\cintervalellplus \triangleq [t_{\ell - 1}: t_{\ell}]$. We refer to each of $\cintervalell$ as an \emph{epoch}. Moreover, with slight abuse of notation, we use $\cpriceiled$ to denote the per-worker releasing fee for all days $t\in\cintervalell$ in the remainder of this subsection.

To extend our approach to this extension, we face new challenges. First, restricting the adversary to only single-switching prediction sequences in \eqref{eq:single switch prediction construction} is not without loss. 
The adversary can benefit from multiple switches to force the algorithm to hedge more through its releasing decisions. Second, it is not clear how to make releasing decisions and how to emulate them similar to Procedure~\ref{alg:emulator}.

\xhdr{High-level sketch of our approach.} We first introduce a larger subset of $O(T^\cptotal)$ of prediction sequences, formally described in ~\eqref{eq:multi switch prediction construction} (in contrast to $T$ prediction sequences in \eqref{eq:single switch prediction construction}). In short, since the releasing fees remain constant in each epoch $\cintervalell
\subseteq[T]$, we consider an adversary that follows a single-switch strategy in each epoch, such as the one introduced in \eqref{eq:single switch prediction construction} for the base model. (Note that our base model is simply a special case when when $\cptotal = 1$, $\priceit = 0$ for all $i,t$, and $\cprice_{i1} = \infty$ for all $i$.) We then concatenate these prediction sequences for different epochs $\cintervalell$ to obtain the entire prediction sequence. Using this new subset , we introduce a \emph{configuration linear program}~\ref{eq:opt cancellation} that helps us to characterize the optimal minimax cost. When $\cptotal$ is constant (which is generally satisfied in the last-mile delivery industry), program~\ref{eq:opt cancellation} has a polynomial size.



Emulating the solution of this new program~\ref{eq:opt cancellation} for a general prediction sequence has intricacies in the extension model. We overcome them first by identifying certain properties of the minimax optimal algorithm and incorporating them directly into program~\ref{eq:opt cancellation}. Second, unlike the minimax optimal algorithms developed in previous sections that only solve an offline program once and then emulate it over the entire time horizon, the minimax optimal algorithm~\OPTSimCan\ (\Cref{alg:opt cancellation}) resolves an updated version of program~\ref{eq:opt cancellation} at the beginning of each epoch $\cintervalell$. Motivated by this resolving idea, we develop an induction argument to show the minimax optimality of \OPTSimCan. 
Below, we explain all the details of our approach.

\xhdr{Configuration LP for the optimal minimax cost.} We describe the subset of prediction sequences  that inspires linear program~\ref{eq:opt cancellation}. First, we introduce the auxiliary notation -- \emph{configuration $\switchseq\in{\times}_{\ell\in[\cptotal]}\cintervalellplus$} and denote $\switchseqspace\triangleq{\times}_{\ell\in[\cptotal]}\cintervalellplus$ as the space of all configurations. As a sanity check, $|\switchseqspace| = O(T^\cptotal)$. Moreover, for any configuration $\switchseq$, we use $\switchseqell \in \cintervalellplus$ to denote its $\ell$-th element, and $\switchseqonetoell \in {\times}_{\ell'\in[\ell]}\cinterval_{\ell'}^+$ to denote its length-$\ell$ prefix. 

Intuitively speaking, configuration $\switchseq\in\switchseqspace$ encodes a prediction sequence $\predictions(\switchseq)= \{[\pL_t(\switchseq),\pR_t(\switchseq)]\}_{t\in[T]}$ constructed as follows:
\begin{align}
\label{eq:multi switch prediction construction}
\begin{array}{lll}
    \ell\in[\cptotal], t\in[t_{\ell - 1} + 1, \switchseqell]:&
    \qquad
    \pLt(\switchseq) \gets \pR_{t - 1}(\switchseq) - \perrort,\quad &
    \pRt(\switchseq) \gets \pR_{t - 1}(\switchseq);
    \\
    \ell\in[\cptotal], t\in[\switchseqell + 1: t_{\ell}]:&
    \qquad
    \pLt(\switchseq) \gets \pL_{t - 1}(\switchseq),\quad &
    \pRt(\switchseq) \gets \pL_{t - 1}(\switchseq) + \perrort.
\end{array}
\end{align}
where $\pLzero(\switchseq) = \pLzero$ and $\pRzero(\switchseq) = \pRzero$. The idea of using this subset is highly non-trivial and substantially reduces the dimensionality of the adversary's problem.\footnote{We note that the constructed subset of prediction sequences $\{\predictions(\switchseq)\}_{\switchseq\in\switchseqspace}$ has $O(T^\cptotal)$ prediction sequences, while the the original prediction sequence space is infinite and uncountable.} To see why, note that construction~\eqref{eq:multi switch prediction construction} ensures that for any two configurations $\switchseq,\switchseq'\in\switchseqspace$, if $\switchseq$ has the same length-$(\ell-1)$ prefix as $\switchseq'$, i.e., $\switchseq_{1:\ell-1} = \switchseq_{1:\ell-1}'$, then predictions $[\pLt(\switchseq),\pRt(\switchseq)]$ and $[\pLt(\switchseq'),\pRt(\switchseq')]$ are identical for every day $t\in[\switchseq_{\ell}\wedge\switchseq_{\ell}']$. Therefore, if we use $\allocitJ$ to denote the number of workers hired from pool $i$ on day $t$ given the prediction sequence $\predictions(\switchseq)$, the following equality should hold for any online algorithm:
\begin{align}
\label{eq:identical allocation constraint}
    \forall\ell\in[\cptotal],
    \forall\switchseq,\switchseq'\in\switchseqspace,
    \switchseq_{1:\ell-1} = \switchseq'_{1:\ell-1},
    \forall t\in[\switchseq_{\ell}\wedge\switchseq_{\ell}']:
    \qquad
    \allocitJ = \allocit(\switchseq')
\end{align}
Similarly, if we use $\revokeiellJ$ to denote the number of workers released from pool $i$ during days in $\cintervalell$ given prediction sequence $\predictions(\switchseq)$, the following equality should hold for any online algorithm:\footnote{Since  the per-worker releasing fee remains the same for all days in each $\cintervalell$, it is without loss of generality to assume that online algorithms only release workers on days $t_1, t_2, \dots, t_L$. Thus, it suffices to introduce a single variable to denote the releasing during days in $\cintervalell$ for each pool.}
\begin{align}
\label{eq:identical cancellation constraint}
    \forall\ell\in[\cptotal],
    \forall\switchseq,\switchseq'\in\switchseqspace,
    \switchseq_{1:\ell} = \switchseq'_{1:\ell}:
    \qquad
    \revokeiellJ = \revokeiell(\switchseq')
\end{align}
Due to the technical reason which we will explain later, from now on we assume that is $\cumalloci$ workers hired from pool $i$ on day 0.\footnote{Both our base and multi-station model then corresponds to $\cumalloci = 0$.}
The feasibility of the staffing profile can be expressed as 
\begin{align}
\label{eq:all feasibility constraint}
    \begin{split}
    \forall i\in[n],
    \forall\switchseq\in\switchseqspace:&
    \qquad
    \displaystyle\sum\nolimits_{t\in[T]}\frac{1}{\wdiscountit} \allocitJ 
    \leq 
    \supplyi
    \\
    \forall \switchseq\in\switchseqspace:&
    \qquad
    \displaystyle\sum\nolimits_{i\in[n]}
    \left(
    \sum\nolimits_{t\in[T]}
    \priceit\allocitJ
    +
    \sum\nolimits_{\ell\in[\cptotal]}
    \cpriceiled\revokeiellJ
    \right)
    \leq \budget
    \\
    \forall i\in[n],
    \forall\ell\in[\cptotal],
    \forall\switchseq\in\switchseqspace:&
    \qquad
    \displaystyle
    \sum\nolimits_{\ell'\in[\ell]}\revoke_{i\ell'}(\switchseq)
    \leq
    \cumalloci + 
    \sum\nolimits_{t\in[t_\ell]}\allocitJ 
    \end{split}
\end{align}
Suppose we use $\lgapJ$ and $\rgapJ$ to denote the largest possible overstaffing and understaffing given prediction sequence $\predictions(\switchseq)$, we have
\begin{align}
\label{eq:bounded cost constraint}
\begin{split}
    \forall \switchseq\in\switchseqspace:&
    \qquad
    \sum\nolimits_{i\in[n]}
    \left(
    \cumalloci + \sum\nolimits_{t\in[T]}\allocitJ
    -
    \sum\nolimits_{\ell\in[\cptotal]}\revokeiellJ
    \right)
    \leq
    \pL_T(\switchseq) + \lgapJ
    \\
    \forall \switchseq\in\switchseqspace:&
    \qquad
    \sum\nolimits_{i\in[n]}
    \left(
    \cumalloci + \sum\nolimits_{t\in[T]}\allocitJ
    -
    \sum\nolimits_{\ell\in[\cptotal]}\revokeiellJ
    \right)
    \geq
    \pR_T(\switchseq) - \rgapJ
\end{split}
\end{align}

We now present linear program~\ref{eq:opt cancellation} that characterizes the optimal minimax cost:
\begin{align}
\tag{$\textsc{LP-release}$}
    \label{eq:opt cancellation}
    \begin{array}{llll}
    \min\limits_{\substack{\xbf,\ybf,\lambdabf,\thetabf\geq \zerobf}}
    &
    \max\limits_{\switchseq\in\switchseqspace}
    \undercost\cdot\rgapJ
    \vee
    \overcost\cdot\lgapJ
    \quad
    & 
    \text{s.t.}
    \\
    &
    \text{constraints}~\labelcref{eq:identical allocation constraint,eq:identical cancellation constraint,eq:all feasibility constraint,eq:bounded cost constraint} 
\end{array}
\end{align}

As we discussed earlier, the optimal minimax algorithm for the costly releasing environment repeatedly resolves program~\ref{eq:opt cancellation} at the beginning of each day $t_\ell$, $\ell\in[\cptotal]$ by viewing the staffing decisions in the remaining time horizon $[t_\ell: T]$ as a new subproblem. We use $\status$ to denote the \emph{initial state} of this subproblem. Here $\status \triangleq (\ell, \bar\cumallocs,\remainsupplies,\remainbudget,[\curpL,\curpR])$ is a tuple where $\cumallocs=\{\bar\cumalloci\}_{i\in[n]}$ is the total number of workers hired before the end of day $t_{\ell-1}$ from each pool $i$, $\remainsupplies=\{\remainsupplyi\}_{i\in[n]}$ is the number of available workers remaining in each pool $i$ at the end of day $t_{\ell - 1}$, $\remainbudget$ is the remaining budget at the end of day $t_{\ell - 1}$, and $[\curpL,\curpR]$ is the prediction revealed on day $t_{\ell - 1}$. See the formal definitions in \eqref{eq:status update}. We use $\lpcancelsubproblem$ to denote the \emph{subprogram} solved at the beginning of day $t_{\ell}$ given initial state $\status$. Specifically, the primitives in program~\ref{eq:opt cancellation} are assigned as 
\begin{align*}
    \cumallocs \gets \bar\cumallocs, ~~
    \supplys\gets\remainsupplies, ~~
    \budget\gets \remainbudget,~~
    \pLzero\gets \curpL,~~\pRzero\gets \curpR,~~
    \{\wdiscountit\}_{t\in[t_{\ell - 1} + 1:T]} \gets 
    \{\tfrac{\wdiscountit}{\wdiscount_{it_{\ell-1}}}\}_{t\in[t_{\ell - 1} + 1:T]}
\end{align*}
and all index sets (e.g., $[T], [\cptotal]$), configuration space $\switchseqspace$ are adjusted correspondingly.

\xhdr{The minimax optimal algorithm and analysis.} 
Now we present the minimax optimal algorithm --- \OPTSimCan (\Cref{alg:opt cancellation}) with its feasibility (\Cref{lem:alg feasibility cancellation}) and optimality guarantee (\Cref{thm:opt alg cancellation}). 

In \OPTSimCan, there are $\cptotal$ phases, corresponding to subintervals/epochs $\cintervalell$ for each $\ell\in[\cptotal]$. Specifically, at the end of day $t_{\ell - 1}$, given the current staffing profile $\{\allocit,\revokeit\}_{i\in[n],t\in[t_{\ell - 1}]}$ and prediction $[\pL_{t_{\ell-1}},\pR_{t_{\ell-1}}]$, the algorithm updates the initial state $\status$ as follows:
\begin{align}
\begin{split}
\label{eq:status update}
    \forall i\in[n]:&\qquad 
    \cumalloci \gets \sum\nolimits_{t\in[t_{\ell - 1}]}\allocit - \revokeit,\quad
    \remainsupplyi \gets \wdiscount_{it_{\ell - 1}}\left(\supplyi - \sum\nolimits_{t\in[t_{\ell - 1}]}\frac{1}{\wdiscountit}\allocit\right),
    \\
    &\qquad 
    \remainbudget \gets \budget - \sum\nolimits_{i\in[n]}\sum\nolimits_{t\in[t_{\ell - 1}]}\priceit\allocit+\cpriceit\revokeit,\quad
    \curpL \gets \pR_{t_{\ell - 1}} - \perror_{t_{\ell - 1}},\quad
    \curpR \gets \pR_{t_{\ell - 1}}.
\end{split}
\end{align}
The algorithm solves subprogram $\lpcancelsubproblem$ given the updated initial state $\status$ and obtains its optimal solution $\{\allocitJ^*,\revoke_{i\ell'}^*(\switchseq),\lgapJ^*,\rgapJ^*\}_{i\in[n],t\in[t_{\ell- 1} + 1:T],\ell'\in[\ell:L],\switchseq\in\switchseqspace_{\ell:L}}$ where $\switchseqspace_{\ell:L}$ is defined as $\switchseqspace_{\ell:L}\triangleq \times_{\ell'\in[\ell:L]}\cinterval_{\ell'}^+$.
The algorithm then constructs the canonical hiring decision $\{\canallocit\}_{i\in[n],t\in\cintervalell}$ as follows:
\begin{align}
\label{eq:canalloc construction}
    \forall i\in[n],\forall t\in\cintervalell:\qquad
    \canallocit \gets \allocit^*(\switchseq^{(t_{\ell})})
\end{align}
where $\switchseq^{(t_{\ell})}$ is an arbitrary sequence such that $\switchseq^{(t_{\ell})}_1 = t_{\ell}$.\footnote{For any configuration $\switchseq\in\switchseqspace_{\ell:L}$, recall $\switchseq_1$ specifies a day in $\cintervalellplus$. Due to constraint~\eqref{eq:identical allocation constraint}, all configurations $\switchseq$ such that $\switchseq_1 = t_{\ell}$ have the same $\allocitJ$ for $t \in[\cintervalell]$.}
Similarly, the algorithm constructs the canonical releasing decision $\{\canrevokeited\}_{i\in[n],t\in\cintervalellplus}$ as follows:
\begin{align*}
    \forall i\in[n],\forall t\in\cintervalellplus:\qquad
    \canrevokeited \gets \revokeiell^*(\switchseq^{(t)})
\end{align*}
where $\switchseq^{(t)}$ is an arbitrary sequence such that $\switchseq^{(t)}_1 = t$.\footnote{Due to constraint~\eqref{eq:identical cancellation constraint}, all configurations $\switchseq$ such that $\switchseq_1 = t$ have the same $\revokeiellJ$.}
Similar to all minimax optimal algorithms in previous models, \OPTSimCan\ implements Procedure~\ref{alg:emulator} with canonical hiring decision $\{\canallocit\}_{i\in[n],t\in\cintervalell}$ constructed above as its input to determine the actual hiring decisions for every day $t\in\cintervalell$. Additionally, on day $t_{\ell}$, the algorithm identifies the largest index $k\in\cintervalellplus$ satisfying\footnote{Such index $k$ always exists, since this equality holds trivially for $k = t_{\ell - 1}.$}
\begin{align}
\label{eq:critical day cancellation}
    \pR_{k} - \sum\nolimits_{i\in[n]}\sum\nolimits_{t\in[t_{\ell - 1}+1:k]}\allocit =
    \pR_{t_{\ell-1}} - \sum\nolimits_{i\in[n]}\sum\nolimits_{t\in[t_{\ell - 1}+1:k]}\canallocit
\end{align}
The algorithm then uses canonical releasing decision $\{\canrevoke_i\ked\}_{i\in[n]}$ to determine the actual releasing decision $\{\revoke_i\}_{i\in[n]}$, i.e., $\revoke_i$ hired workers are released from pool $i$ on day $t_{\ell}$. Specifically, it computes $\{\revoke_i\}_{i\in[n]}$ as an arbitrary solution such that
\begin{align}
    \label{eq:actual cancellation construction}
    \begin{split}
    \forall i\in[n]:&\qquad 0\leq \revoke_i \leq \canrevoke_{ik}
    \\
    &\qquad 
    \sum\nolimits_{i\in[n]}
\left(
\sum\nolimits_{t\in[t_{\ell-1} +1:k]}\left(
\canallocit
-
\allocit
\right)
-
\left(
\canrevoke_{i}\ked
-
\revoke_i
\right)
\right)
=
\pR_{t_{\ell}}(\switchseq\ked)
-
\pR_{t_{\ell}}
    \end{split}
\end{align}
where $\pR_{t_{\ell}}(\switchseq\ked) = \pR_{t_{\ell - 1}} - \perror_k + \perror_{t_\ell}$ by construction~\eqref{eq:multi switch prediction construction}.
If no feasible solution satisfies condition~\eqref{eq:actual cancellation construction}, the algorithm makes no releasing on day $t_{\ell}$.

Finally, if $\ell < [\cptotal]$, the algorithm moves from phase $\ell$ to phase $\ell + 1$ and repeats.
\begin{algorithm}
\setcounter{AlgoLine}{0}
    \SetKwInOut{Input}{input}
    \SetKwInOut{Output}{output}
    \Input{initial pool sizes $\{\supplyi\}_{i\in[n]}$, availability rates $\{\wdiscountit\}_{i\in[n], t\in[T]}$, initial demand range $[\pLzero,\pRzero]$, prediction error upper bounds $\{\perrort\}_{t\in[T]}$, per-worker wages $\{\priceit\}_{i\in[n],t\in[T]}$, per-worker releasing fees $\{\cpriceit\}_{i\in[n],t\in[T]}$}
    \Output{staffing profile $(\xbf,\ybf)$}
    \vspace{2mm}
    
    \For{each $\ell\in[\cptotal]$}{
    
    \vspace{1mm}

    update initial state $\status$ using \eqref{eq:status update} at the end of day $t_{\ell-1}$

    \vspace{2mm}
    
    find an optimal solution $(\xbf^*, \ybf^*, \lambdabf^*, \thetabf^*)$ of subprogram~$\lpcancelsubproblem$

    \vspace{2mm}

    invoke Procedure~\ref{alg:emulator} with canonical staffing profile $\canxbf$ constructed in \eqref{eq:canalloc construction} 
    {\small\color{\commentcolor}\tcc{facing prediction sequence $\predictions$ for each day $t \in \cintervalell$}}

    \vspace{2mm}

    \If{there exists releasing decision $\{\revoke_i\}_{i\in[n]}$ satisfying condition \eqref{eq:actual cancellation construction}}{
    
    \vspace{1mm}

    release $\revoke_i$ hired workers from each pool $i$ on day $t_{\ell}$
    }
    }
    
    \caption{\OPTSimCan}
    \label{alg:opt cancellation}
\end{algorithm}

\begin{remark}
    \OPTSimCan\ has $\texttt{Poly}(n, T^\cptotal)$ running time.
\end{remark}

\begin{remark}
    \OPTSimCan\ recovers \OPTSim\ when $\cptotal = 1$ and $\cprice_{i1} = \infty$ for all $i$.
\end{remark}

\begin{restatable}{lemma}{optalgcanfeasibility}
\label{lem:alg feasibility cancellation}
    In the costly releasing environment,
    the staffing profile outputted by \OPTSimCan\ is feasible.
\end{restatable}
The proof of \Cref{lem:alg feasibility cancellation} (see \Cref{apx:optalgcanfeasibility}) utilizes the properties of Procedure~\ref{alg:emulator} established in \Cref{lem:emulator} and the structural properties about the implementation of releasing decisions in \eqref{eq:actual cancellation construction}.

\optalgcan*


The proof of \Cref{thm:opt alg cancellation} (see \Cref{apx:optalgcan}) follows a similar high-level idea as the proofs of \Cref{thm:opt alg,thm:opt alg multi station}. We first argue that program~\ref{eq:opt cancellation} upper bounds the optimal minimax cost. We then use an induction argument (backward over index $\ell$) showing that the optimal objective of subprogram~$\lpcancelsubproblem$ is at most the cost guarantee of \OPTSimCan\ from phase $\ell$ to phase $\cptotal$ given initial state $\status$.

%% file: Paper/apx-regret.tex
\label{apx:regret}
\newcommand{\benchmark}{\texttt{OPT}}
\newcommand\barbelow[1]{\stackunder[1.2pt]{$#1$}{\rule{1.ex}{.175ex}}}
\newcommand{\perrordelta}{\delta}
\newcommand{\perrordeltaj}{\perrordelta_j}
\newcommand{\perrorLB}{\Delta}
\newcommand{\perrorLBj}{\perrorLB_j}
\newcommand{\pLjT}{\pL_{jT}}
\newcommand{\pRjT}{\pR_{jT}}

\subsection{Regret}

The minimax optimal algorithms developed in this paper also are also regret optimal under the following regularity assumption about the workforce pools.\footnote{To simplify the presentation, we focus on the regret minimization version of the base model. The same argument holds for extensions studied in \Cref{sec:extension} as well.}

\begin{assumption}
\label{asp:sufficient initial supply}
    There exists a feasible staffing profile $\xbf$ such that $\sum_{i\in[n]}\sum_{t\in[T]}\allocit \geq \pRzero$.
\end{assumption}

The classic definition of the regret in the online algorithm design literature is as follows: given an instance $\instance$, the \emph{regret guarantee} of an online algorithm $\ALG$ is defined as
\begin{align*}
    \max_{\predictions,\demand}~\expect{\Cost[\demand]{\ALG(\predictions}} - \Cost[\demand]{\benchmark(\predictions,\demand)}
\end{align*}
where $\ALG(\predictions)$ is the (possibly randomized) staffing profile generated by algorithm $\ALG$ under prediction sequence $\predictions= \{\predictiont\}_{t\in[T]}$, and $\benchmark(\predictions, \demand)$ is the optimal staffing profile generated by the optimal clairvoyant benchmark that knows the entire prediction sequence $\predictions$ and demand $\demand$. 

Under Assumption~\ref{asp:sufficient initial supply}, the optimal clairvoyant benchmark $\benchmark$ can pick a feasible staffing profile~$\xbf^*$ that matches the demand perfectly, i.e., $\sum_{i\in[n]}\sum_{t\in[T]}\allocit = \demand$. Therefore, its cost is $\Cost[\demand]{\benchmark(\predictions,\demand)} = 0$ for all prediction sequences and demand. Consequently, the regret guarantee becomes equivalent to the cost guarantee (\Cref{def:minmaxcost}) studied in the main text. We summarize our discussion into the following proposition.
\begin{proposition}
    Under Assumption~\ref{asp:sufficient initial supply}, an online algorithm is minimax optimal if and only if it is regret optimal.
\end{proposition}

\subsection{Competitive Ratio}
The classic definition of the competitive ratio in the online algorithm design literature is as follows: given an instance $\instance$, the \emph{competitive ratio} of an online algorithm $\ALG$ is defined as
\begin{align*}
    \max_{\predictions,\demand}~\frac{\expect{\Cost[\demand]{\ALG(\predictions}}}{\Cost[\demand]{\benchmark(\predictions,\demand)}}
\end{align*}
where $\ALG(\predictions)$ is the (possibly randomized) staffing profile generated by algorithm $\ALG$ under prediction sequence $\predictions= \{\predictiont\}_{t\in[T]}$, and $\benchmark(\predictions, \demand)$ is the optimal staffing profile generated by the optimal clairvoyant benchmark that knows the entire prediction sequence $\predictions$ and demand $\demand$. 

Following the same argument as the previous subsection, under \Cref{asp:sufficient initial supply}, the benchmark $\Cost[\demand]{\benchmark(\predictions,\demand)} = 0$ for all prediction sequences and demand. Consequently, the problem becomes trivial since all online algorithms have the same competitive ratio.

%% file: Paper/apx-refined-gamma.tex
\label{apx:refined-gamma-single-pool}
We have characterized $\optcost$ as the solution of a fixed-point equation in \Cref{sec:simple instance} in our warm-up single-pool scenario in its general form and later showed in \Cref{sec:base model} that this quantity could alternatively be viewed as the optimal objective value for a special case of \ref{eq:opt reduced form} with a single pool. 

A follow-up question to the above characterization is whether we can obtain a more explicit characterization for this quantity for certain primitives of the model. In this section, our aim is to answer this question by providing a refined (and more explicit) characterization of the optimal minimax cost $\optcost$ for a simple and stylized single-pool instance. We conjecture that the explicit characterization of $\optcost$ for more general instances (e.g., with more than a single pool) is challenging and defer it to future research.  

\xhdr{Setup.} There is a single pool of workforce with initial supply $\supply$. We normalize the initial range of unknown demand $\demand$ as $\pLzero = 0$ and $\pRzero = 1$. The availability curve satisfies $\wdiscount_t = \eta^{t}$. The prediction error upper bounds satisfy $\perror_t = 1 - \perror^{T - t}$. Both parameters $\eta$ and $\perror$ are between 0 and 1. 

\xhdr{Characterization.} Solving the fixed point condition in \Cref{prop:opt alg simple instance}, the optimal minimax cost $\optcost$ admits the following characterization: there are two cases depending on the amount of initial supply pool size $\supply$:

\smallskip
\noindent
    (I) \underline{\textsl{Low initial supply:}}
    Suppose initial supply pool size $\supply$ satisfies
    \begin{align*}
        \supply \leq \frac{\undercost + \overcost\perror^{T-1}}{(\undercost+\overcost)\eta}~,
    \end{align*}
    which is the rearrangement of Condition~\eqref{eq:low-supply-fixed-point} in \Cref{sec:simple instance}.
    Then the optimal minimax cost $\optcost$ is 
    \begin{align*}
        \optcost = \undercost - \undercost\eta \supply~.
    \end{align*}

\smallskip
\noindent
     (II) \underline{\textsl{Sufficient initial supply:}} Suppose initial supply pool size $\supply$ satisfies
    \begin{align*}
        \supply \geq \frac{\undercost + \overcost\perror^{T-1}}{(\undercost+\overcost)\eta}~.
    \end{align*}
    Define auxiliary function $t^{\dagger}(x)$ as 
    \begin{align*}
        t\primed(x) \triangleq 
        \min\left\{ 
        \left\lceil
        \frac{\ln\left(
         \left(\eta \supply - \frac{x}{\overcost} - \perror^{T-1}\right)\cdot \perror^{1-T}\cdot (1 - \perror)^{-1}
        \cdot(1 - \perror\eta) + 1
        \right)}{-\ln\perror - \ln \eta}
        +1
        \right\rceil
        ,
        T
        \right\}~.
    \end{align*}
    For additional intuition, $t^{\dagger}(\targetcost)$ denotes the last day of hiring under the Greedy-Staffing algorithm (\Cref{alg:opt simple instance}) with the target overstaffing cost of $\targetcost$, as illustrated in \Cref{fig:geometric-proof}.
    The optimal minimax cost $\optcost$ is the unique solution of the following equation (with variable $\targetcost$):\footnote{The uniqueness of the solution for the equation is shown in \Cref{prop:opt alg simple instance}, and can also be inferred from \Cref{thm:opt alg}.}
    \begin{align*}
        \left(\frac{\overcost}{\undercost} + \frac{1}{\overcost} - 
        \frac{\eta^{t\primed(\targetcost)-1}}{\overcost}
        \right)
        \targetcost 
        =
        \perror^{T - t\primed(\targetcost) + 1}
        -
        \eta^{t\primed(\targetcost) - 1}
        \left(
        s\eta - {\perror^{T - 1}}
        -
        {(1-\perror)\perror^{T - 1}\left(\left({\perror\eta}\right)^{2-t\primed(\targetcost)}-1\right)}{(1-\perror\eta)^{-1}}
        \right)~.
    \end{align*}
    In the following, we present the detailed derivation of the function $t^{\dagger}(x)$ and the equation for $\optcost$.

    Recall that for any given overstaffing cost upper bound $\targetcost$, under the worst-case prediction sequence specified in \Cref{lem:max-supply-consuming-sequence}, the Greedy-Staffing algorithm (\Cref{alg:opt simple instance}) hires $\bar \pL_1 + \targetcost/\overcost = \pRzero - \perror_1 + \targetcost/\overcost$ in day 1, and $\bar \pL_2 + \targetcost/\overcost - (\bar \pL_1 + \targetcost /\overcost) = \bar \pL_2 - \bar \pL_1 = \perror_1 - \perror_2$ in day~2, and so on. Hence, the last day of hiring  $t^{\dagger}$ under the worst-case prediction sequence satisfies
    \begin{align*}
        \frac{1}{\wdiscount_1}\left(\pRzero - \perror_1 + \frac{\targetcost}{\overcost}\right)
        +
        \sum_{t\in[2:t\primed - 1]}
        \frac{1}{\wdiscount_t} \left(\perror_{t-1} - \perror_{t}\right)
        \leq 
        \supply 
        < 
        \frac{1}{\wdiscount_1}\left(\pRzero - \perror_1 + \frac{\targetcost}{\overcost}\right)
        +
        \sum_{t\in[2:t\primed]}
        \frac{1}{\wdiscount_t} \left(\perror_{t-1} - \perror_{t}\right)~,
    \end{align*}
    where the left-hand side is the projected cumulative hiring from day 1 to day $t\primed - 1$, and the right-hand side is the projected cumulative hiring from day 1 to day $t\primed$ (if supply was not exhausted on day $t\primed$).
    Since $\perror_t = 1 - \perror^{T - t}$, $\wdiscount_t = \eta^t$, and $\pRzero = 1$, the above condition can be rewritten as 
    \begin{align*}
        \frac{1}{\eta}\left(\perror^{T-1} + \frac{\targetcost}{\overcost}\right)
        +
        \sum_{t\in[2:t\primed - 1]}
        \frac{1}{\eta^{t}} \left(\perror^{T - t + 1} - \perror^{T - t}\right)
        \leq {}
        \supply 
        < 
        \frac{1}{\eta}\left(\perror^{T-1} + \frac{\targetcost}{\overcost}\right)
        +
        \sum_{t\in[2:t\primed]}
        \frac{1}{\eta^{t}} \left(\perror^{T - t + 1} - \perror^{T - t}\right)~.
        \end{align*}
    Multiplying all sides by $\eta$ and then subtracting all sides by $\left(\perror^{T-1} + \frac{\targetcost}{\overcost}\right)$, it becomes 
    \begin{align*}
        \sum_{t\in[2:t\primed - 1]}
        \frac{1}{\eta^{t-1}} \left(\perror^{T - t} - \perror^{T - t + 1}\right)
        \leq {}
        & \eta \supply - \left(\perror^{T-1} + \frac{\targetcost}{\overcost}\right)
        < 
        \sum_{t\in[2:t\primed]}
        \frac{1}{\eta^{t-1}} \left(\perror^{T - t} - \perror^{T - t + 1}\right)~.
    \end{align*}
    Using the formula of the sum of the geometric progression, it becomes
    \begin{align*}
        (1 - \perror)\frac{1}{\eta}\perror^{T - 2}
        \frac{(\frac{1}{\perror\eta })^{t\primed - 2} - 1}{\frac{1}{\perror\eta } - 1}
        \leq {}
        & \eta \supply - \left(\perror^{T-1} + \frac{\targetcost}{\overcost}\right)
        < 
        (1 - \perror)\frac{1}{\eta}\perror^{T - 2}
        \frac{(\frac{1}{\perror\eta })^{t\primed - 1}-1}{\frac{1}{\perror\eta } - 1}~.
    \end{align*}
    Rearranging the term, we obtain
    \begin{align*}
    \left(\frac{1}{\perror\eta }\right)^{t\primed - 2}
    \leq 
        \left(\eta \supply - \perror^{T-1} - \frac{\targetcost}{\overcost}\right)\perror^{1-T}(1-\perror)^{-1}(1 - \perror\eta)
        +1
        <
        \left(\frac{1}{\perror\eta }\right)^{t\primed - 1}~,
    \end{align*}
    and thus
    \begin{align*}
        t\primed - 2\leq \frac{1}{-\ln\perror - \ln \eta}
        \ln \left(
        \left(\eta \supply - \perror^{T-1} - \frac{\targetcost}{\overcost}\right)\perror^{1-T}(1-\perror)^{-1}(1 - \perror\eta)
        +1
        \right)
        <
        t\primed - 1~,
    \end{align*}
    which is consistent with our definition of $t\primed(\targetcost)$ function.\footnote{There is also a boundary case where the supply is not exhausted after day T. In this case, we set $t\primed = T$.}

    Given the above closed-form for  $t\primed(\targetcost)$, we next verify the aforementioned equation for $\optcost$. As we argued in \Cref{sec:simple instance} (e.g., see \Cref{fig:geometric-proof} and proof of \Cref{prop:opt alg simple instance}) using the fixed-point approach, under the worst case prediction sequence, the following equation holds for $\targetcost=\optcost$:
    \begin{align*}
        \overcost \cdot \targetcost = \undercost \cdot \left(\pRzero - 
        \left( \underbrace{\bar L_{t\primed - 1} + \frac{\targetcost}{\overcost}}_{\text{hiring in first $t\primed - 1$ days}} + \underbrace{\wdiscount_{t\primed}\left(\supply - 
        \left(\frac{1}{\wdiscount_1}\left(\pRzero - \perror_1 + \frac{\targetcost}{\overcost}\right)
        +
        \sum_{t\in[2:t\primed - 1]}
        \frac{1}{\wdiscount_t} \left(\perror_{t-1} - \perror_{t}\right)
        \right)\right)}_{\text{hiring in day $t\primed$}} \right) \right)~,
    \end{align*}
    where the left-hand and right-hand sides are the worst-case overstaffing cost and worst-case understaffing cost, respectively. Dividing both sides by $\undercost$ and invoking $\pRzero = 1$, $\bar \pL_{t\primed - 1} = 1 - \perror^{T - t\primed + 1}$, and $\wdiscount_t = \eta^{t}$, we obtain 
    \begin{align*}
        \frac{\overcost}{\undercost}\cdot \targetcost
        =
        \perror^{T - t^{\dagger} + 1}
        -
        \frac{\targetcost}{\overcost}
        -
        \eta^{t^{\dagger}}
        \left(
        \supply - \frac{1}{\eta}\left(\perror^{T-1} + \frac{\targetcost}{\overcost}\right)
        -
        (1 - \perror)\frac{1}{\eta^2}\perror^{T - 2}
        \frac{(\frac{1}{\perror\eta })^{t\primed - 2} - 1}{\frac{1}{\perror\eta } - 1}
        \right)~,
    \end{align*}
    which is equivalent to 
    \begin{align*}
        \left(\frac{\overcost}{\undercost} + \frac{1}{\overcost} - 
        \frac{\eta^{t\primed-1}}{\overcost}
        \right)
        \targetcost
        =
        \perror^{T - t^{\dagger} + 1}
        -
        \eta^{t\primed - 1}
        \left(
        \supply\eta - {\perror^{T - 1}}
        -
        {(1-\perror)\perror^{T - 1}\left(\left({\perror\eta}\right)^{2-t\primed}-1\right)}{(1-\perror\eta)^{-1}}
        \right)~.
    \end{align*}
    The above calculations complete the verification of the aforementioned equation for $\optcost$ as desired.

%% file: Paper/apx-proofs.tex
\label{sec:apx-missing}

\subsection{Proof of Lemma~\ref{lem:max-supply-consuming-sequence}}
\input{Paper/appendix-proofs/apx-maxsupplysequence}

\subsection{Proof of Proposition~\ref{prop:opt alg simple instance}}
\label{apx:optalgsimpleinstance}
\input{Paper/appendix-proofs/apx-optalgsimpleinstance}

\subsection{Proof of Lemma~\ref{lem:opt reduced form lower bound optimal minimax cost}}
\label{apx:reduced-form-lower}
\input{Paper/appendix-proofs/apx-opt-reduced-form-lowerbound}

\subsection{Proof of Theorem~\ref{thm:opt alg resolving}}
\label{apx:opt alg resolving}
\input{Paper/appendix-proofs/apx-thmoptalgresolving}

\subsection{Proof of Theorem~\ref{thm:opt alg multi station}}
\label{apx:optalginfty}
\input{Paper/appendix-proofs/apx-optalginfty}

\label{apx:optalgone}
\input{Paper/appendix-proofs/apx-optalgone}

\subsection{Proof of Lemma~\ref{lem:alg feasibility cancellation}}
\label{apx:optalgcanfeasibility}
\input{Paper/appendix-proofs/apx-optalgcanfeasibility}

\subsection{Proof of Theorem~\ref{thm:opt alg cancellation} in a More General Model}
\label{apx:optalgcan}
\input{Paper/appendix-proofs/apx-optalgcan}


%% file: Paper/appendix-proofs/apx-maxsupplysequence.tex
\label{apx:maxsupplysequence}
\maxsupplysequence*

\begin{proof}
    Let $\predictions\primed$ be prediction sequence that maximizes the understaffing cost of \Cref{alg:opt simple instance} against worst-case demand. Suppose $\predictions\primed$ is not equivalent to $\overline{\predictions}$ in the lemma statement. We consider the following two-step argument.

    \xhdr{Step 1:}
    Let $k\in[T]$ be the smallest index such that $\pR_k\primed < \pRzero$. If no such $k$ exists, move to step 2. Otherwise, we construct a new prediction sequence $\predictions\doubleprimed$ as follows:
    \begin{align*}
        t\in[k - 1]:&\qquad\qquad
        \pL_t\doubleprimed = \pL_t\primed,\quad
        \pR_t\doubleprimed = \pR_t\primed,\quad
        \\
        t\in[k:T]:&\qquad\qquad
        \pL_t\doubleprimed = \pL_t\primed + (\pRzero - \pR_k\primed),\quad
        \pR_t\doubleprimed = \pR_t\primed + (\pRzero - \pR_k\primed),\quad
    \end{align*}
    We claim that the worst understaffing cost under $\predictions\doubleprimed$ is weakly higher than $\predictions\primed$. To see this, let $t\primed$ be the day that the supply feasibility binds in the algorithm. If $t\primed \leq k$, by construction the staffing decision under $\predictions\doubleprimed$ is the same as $\predictions\primed$, and consequently $\predictions\doubleprimed$ leads to a weakly higher understaffing cost since $\pR_T\doubleprimed > \pR_T\primed$.
    If $t\primed > k$, by construction the staffing decision under $\predictions\doubleprimed$ is the same as $\predictions\primed$ before day $k$, and after day $k$ until the day $t\doubleprimed$ when supply feasibility binds under prediction sequence $\predictions\doubleprimed$. We have $k\leq t\doubleprimed\leq t \primed$ by definition. Moreover, the change in the total hiring from prediction sequence $\predictions\primed$ to $\predictions\doubleprimed$ is 
    \begin{align*}
        \underbrace{(\pRzero - \pR_k\primed)}_{\text{increase of hires in day $k$}} - \underbrace{(\pRzero - \pR_k\primed)\cdot \frac{\wdiscount_{t\doubleprimed}}{\wdiscount_k}}_{\text{decreases of hires in day $t\doubleprimed$}}
    \end{align*}
    which is weakly less than the increase of $\pRzero - \pR_k\primed$ in the worst-case demand, which ensures our claim.

    Repeating the above construction, we obtain a new $\predictions\primed$ with $\pR_t\primed = \pRzero$ for all $t\in[T]$ that induces weakly higher worst-case understaffing cost. If $\predictions\primed$ is now equal to $\overline{\predictions}$, the lemma is shown. Otherwise, move to step 2.

    \xhdr{Step 2:} 
    Let $k\in[T]$ be the smallest index such that $\pR_k\primed - \pL_k\primed < \perror_k$. We construct a new prediction sequence $\predictions\doubleprimed$ as follows:
    \begin{align*}
        t\in[k - 1]:&\qquad\qquad
        \pL_t\doubleprimed = \pL_t\primed,\quad
        \pR_t\doubleprimed = \pR_t\primed,\quad
        \\
        &\qquad\qquad
        \pL_k\doubleprimed = \pR_k\primed - \perror_k,\quad
        \pR_k\doubleprimed = \pR_k\primed
        \\
        t\in[k - 1:T]:&\qquad\qquad
        \pL_t\doubleprimed = \pL_t\primed,\quad
        \pR_t\doubleprimed = \pR_t\primed,\quad
    \end{align*}
    We claim that the worst understaffing cost under $\predictions\doubleprimed$ is weakly higher than $\predictions\primed$. To see this, let $t\primed$ be the day that the supply feasibility binds in the algorithm. If $t\primed \leq k$, by construction the staffing decision under $\predictions\doubleprimed$ is the same as $\predictions\primed$, and consequently $\predictions\doubleprimed$ has the same understaffing cost.
    If $t\primed > k$, by construction the staffing decision under $\predictions\doubleprimed$ is the same as $\predictions\primed$ before day $k$, and after day $k$ until the day $t\doubleprimed$ when supply feasibility binds under prediction sequence $\predictions\doubleprimed$. We have $t\doubleprimed\geq t \primed$ by definition. Moreover, the change in the total hiring from prediction sequence $\predictions\primed$ to $\predictions\doubleprimed$ is 
    \begin{align*}
        \underbrace{-(\pR_k\primed-\pR_k\doubleprimed)}_{\text{decreases of hires in day $k$}} + \underbrace{(\pR_k\primed-\pR_k\doubleprimed)\cdot \frac{\wdiscount_{t\doubleprimed}}{\wdiscount_k}}_{\text{increases of hires in day $t\doubleprimed$}}
    \end{align*}
    which is non positive, and thus our claim holds.

    Repeating the above construction, we obtain  $\overline{\predictions}$ whose worst-case understaffing cost is weakly higher. Therefore, the lemma is shown.
\end{proof}

%% file: Paper/appendix-proofs/apx-optalgsimpleinstance.tex
\optalgsimpleinstance*
\begin{proof}
Following the argument in Observation (i) and (ii), \Cref{alg:opt simple instance} with $\targetcost = \optcost$ is minimax optimal. It remains to show the second part of the proposition statement, i.e., 
$\optcost=\undercost\cdot (\pRzero - \wdiscount_1\cdot\supply)$ if inequality~\eqref{eq:low-supply-fixed-point} holds, and otherwise it is the fixed point of the function $\underline{\Gamma}$.

By \Cref{lem:max-supply-consuming-sequence}, it suffices to analyze the execution of \Cref{alg:opt simple instance} under prediction sequence $\overline{\predictions}$.
Fix an arbitrary cost $\targetcost$ as the input of \Cref{alg:opt simple instance}. Let $t^{\dagger}$ be the day when the supply feasibility binds (or day $T$). Now, the adversary either picks $d=L_{t^{\dagger}}$, where in that case the algorithm pays an overstaffing cost of no more than $\targetcost$, or picks $d=R_{t^{\dagger}}=R_0$, where in that case the algorithm pays an understaffing cost that we denote by $\underline{\Gamma}$. We claim that $\underline{\Gamma}(\cdot)$ is a (weakly) decreasing function of $\Gamma$. To see this, suppose we increase $\Gamma$ by $\partial\Gamma$. The change of total hires is 
\begin{align*}
    \underbrace{\partial\targetcost}_{\text{increase of hires in day $1$}} - \underbrace{\partial\targetcost\cdot \frac{\wdiscount_{t\primed}}{\wdiscount_1}}_{\text{decreases of hires in day $t\primed$}}
\end{align*}
which is non positive.

Given the monotonicity of the function $\underline{\Gamma}(\cdot)$, one of these two cases can occur depending on the amount of initial supply pool size $s$:

\begin{enumerate}[label=(\Roman*)]
    \item \emph{Sufficient initial supply:} $\underline{\Gamma}(\cdot)$ has a \emph{fixed point} $\Gamma$ in $[0,\overcost\cdot(\rho_1\cdot s-\overline{L}_1)]$, that is, $\Gamma\in [0,\overcost\cdot(\rho_1\cdot s-\overline{L}_1)]$ such that $\underline{\Gamma}(\Gamma)=\Gamma$. In this case, we claim $\optcost=\targetcost$. Suppose by contradiction $\optcost\neq\targetcost$. If $\optcost<\targetcost$, then as $\underline{\Gamma}(\cdot)$ is weakly decreasing we have $\underline{\Gamma}(\optcost)\geq \underline{\Gamma}(\Gamma)=\Gamma>\optcost$, which is a contradiction, as the maximum understaffing cost of the minimax optimal algorithm (i.e.,  \Cref{alg:opt simple instance} with $\optcost$ as input) cannot exceed $\optcost$. If $\optcost>\targetcost$, then $\underline{\Gamma}(\Gamma)=\Gamma<\optcost$, so both understaffing and overstaffing costs of \Cref{alg:opt simple instance} with $\targetcost$ are strictly smaller than $\optcost$, again a contradiction to the minimax optimality of \Cref{alg:opt simple instance} with $\optcost$ as input.
      \item \emph{Low initial supply:} $\underline{\Gamma}(\cdot)$ has no fixed point in $[0,\overcost\cdot(\rho_1\cdot s-\overline{L}_1)]$, that is, $\underline{\Gamma}(\Gamma')>\Gamma'$ for all $\Gamma'\in [0,\overcost\cdot(\rho_1\cdot s-\overline{L}_1)]$. As $\underline{\Gamma}(\cdot)$ is weakly decreasing, this occurs if and only if (see \Cref{fig:geometric-proof}): 
     \begin{align*}
    \underline{\Gamma}\left(\overcost\cdot(\rho_1\cdot s-\overline{L}_1)\right)=\underbrace{\undercost\cdot (\pRzero - \wdiscount_1\cdot\supply)}_{\substack{\textrm{max understaffing cost}\\\textrm{when $x_1=\wdiscount_1\cdot\supply$}}} 
    >
    \underbrace{\overcost\cdot (\wdiscount_1 \cdot \supply - \overline{L}_1)}_{\substack{\textrm{max overstaffing cost} \\\textrm{when $x_1=\wdiscount_1\cdot\supply$}}}
\end{align*}
    In this case, we claim $\optcost= \underline{\Gamma}\left(\overcost\cdot(\rho_1\cdot s-\overline{L}_1)\right)=\undercost\cdot (\pRzero - \wdiscount_1\cdot\supply)$. To see this, first note that $\optcost>\overcost\cdot(\rho_1\cdot s-\overline{L}_1)$, because otherwise $\underline{\Gamma}(\optcost)>\optcost$, which is a contradiction, as the maximum understaffing cost of the minimax optimal algorithm (i.e.,  \Cref{alg:opt simple instance} with $\optcost$ as input) cannot exceed $\optcost$. Now, if we run \Cref{alg:opt simple instance} with such a $\optcost$ as input, it hires $x_1=\rho_1\cdot s$ on day $1$ and runs out of supply later. Due to inequality~\eqref{eq:low-supply-fixed-point}, the maximum understaffing cost is more than the maximum understaffing cost, and hence $\optcost=\undercost\cdot(\pRzero - \wdiscount_1\cdot\supply)$. 
\end{enumerate}
Combining two cases, the second half of the proposition statement is shown.
\end{proof}

%% file: Paper/appendix-proofs/apx-opt-reduced-form-lowerbound.tex
\lemmalower*
\begin{proof}
    To show this lemma, we construct a feasible solution for \ref{eq:opt reduced form}, whose objective value is equal to the cost guarantee of algorithm $\ALG$.
    
    Consider the subset of prediction sequences $\{\predictions\ked\}_{k\in[T]}$ defined in equation~\eqref{eq:single switch prediction construction} (\Cref{sec:base model result}). Recall that by construction, for every $t \in[T]$, the first $t$ predictions from day 1 to day $t$ are the same for all prediction sequences $\predictions\ked$ with $k\geq t$. Therefore, the staffing decision of the online algorithm $\ALG$ (possibly randomized) on each day $t$ should be the same for all prediction sequences $\predictions\ked$ with $k \geq t$.
    
    Given this observation, let the random variable $\randomallocit$ be the number of workers hired by the algorithm from pool $i$ in day $t$ under prediction sequence $\predictions\Ted$.
    Due to the feasibility of the algorithm under prediction sequence $\predictions\Ted$, for all sample paths (over the randomness of the algorithm), we have
    \begin{align*}
    \forall i\in[n]:\quad & 
            \displaystyle\sum\nolimits_{t\in[T]}\frac{1}{\wdiscountit}\randomallocit \leq \supplyi
    \end{align*}
    We introduce two auxiliary notations $\lgap$ and $\rgap$ defined as
    \begin{align*} 
        &
        \qquad
        \lgap \gets
        \max_{k\in[T]}~
        \plus{
        \expect{\sum\nolimits_{i\in[n]}\sum\nolimits_{t\in[k]}\randomallocit} - \left(\max_{\tau\in[0:k]}\pRzero  - \perror_{\tau} - 2 \pbias_\tau\right)
        }
        \\
        & \qquad \rgap \gets 
        \plus{
        \pRzero - \expect{\sum\nolimits_{i\in[n]}\sum\nolimits_{t\in[T]}\randomallocit}
        }
    \end{align*}
    Consider the following candidate (not necessarily optimal) solution $(\xbf,\targetcost)$ for \ref{eq:opt reduced form}:
    \begin{align*}
        i\in[n],t\in[T]:&
        \qquad
        \allocit \gets \expect{\randomallocit}
        \\
        &
        \qquad
        \targetcost \gets \max\{\overcost\cdot \lgap, \undercost\cdot \rgap\}
    \end{align*}
    By construction, all three constraints are satisfied. 
    Below we argue the objective value of the constructed solution is at most the cost guarantee of the algorithm $\ALG$ in two different cases.

    \xhdr{Case 1 $\left[\undercost\cdot\rgap \geq 
    \overcost\cdot\lgap\right]$:}
    In this case, the objective value of the constructed solution is 
    $\undercost\cdot\rgap$.
    Consider the execution of the algorithm $\ALG$ under prediction sequence $\predictions\Ted$ and demand $\demand \triangleq \pRzero$. Note that the $\pbiass$-consistency is satisfied given constructed demand $\demand$ and prediction sequence $\predictions\Ted$. Moreover, the staffing cost can be lower bounded as 
    \begin{align*}
        \expect{\Cost[\demand]{\ALG(\predictions\Ted)}}
        &\overset{(a)}{\geq}
        \expect{
        \undercost\cdot
        \plus{
        \demand - 
        \sum\nolimits_{i\in[n]}
        \sum\nolimits_{t\in[T]}
        \randomallocit 
        }
        }
        \\
        &\overset{(b)}{\geq}
        \undercost
        \cdot
        \plus{ 
        \demand - 
        \expect{\sum\nolimits_{i\in[n]}
        \sum\nolimits_{t\in[T]}
        \randomallocit}
        }
        \\
        &\overset{(c)}{=}
        \undercost
        \cdot
        \plus{ 
        \pRzero - 
        \expect{\sum\nolimits_{i\in[n]}
        \sum\nolimits_{t\in[T]}
        \randomallocit}
        }
        \\
        &\overset{(d)}{=}
        \undercost\cdot \rgap
    \end{align*}
    where inequality~(a) holds by considering understaffing cost only;
    inequality~(b) holds due to the convexity of
    $\plus{\cdot}$ and Jensen's inequality;
    equality~(c) holds due to the choice of $\demand$;
    and
    equality~(d) holds due to the construction of $\rgap$.

    \xhdr{Case 2 $\left[\undercost\cdot\rgap <
    \overcost\cdot\lgap\right]$:}
    In this case, the objective value of the constructed solution is 
    $\overcost\cdot\lgap$.
    Let $k\doubleprimed$ be the index such that $\lgap = 
        \plus{
        \expect{\sum\nolimits_{i\in[n]}\sum\nolimits_{t\in[k]}\randomallocit} - \left(\max_{\tau\in[0:k]}\pRzero  - \perror_{\tau} - 2 \pbias_\tau\right)
        }$.
    Consider the execution of the algorithm $\ALG$ under prediction sequence $\predictions^{(k\doubleprimed)}$ and demand $\demand \triangleq \pL_T^{(k\doubleprimed)} = \max_{\tau\in[0:k\doubleprimed]}\pRzero  - \perror_{\tau} - 2 \pbias_\tau$. Note that the $\pbiass$-consistency is satisfied given constructed demand $\demand$ and prediction sequence $\predictions^{(k\doubleprimed)}$. Moreover, the staffing cost can be lower bounded as 
    \begin{align*}
        \expect{\Cost[\demand]{\ALG(\predictions^{(k\doubleprimed)})}}
        &\overset{(a)}{\geq}
        \expect{
        \overcost
        \cdot
        \plus{
        \sum\nolimits_{i\in[n]}
        \sum\nolimits_{t\in[k\doubleprimed]}
        \randomallocit 
        -
        \demand 
        }
        }
        \\
        &\overset{(b)}{\geq}
        \overcost
        \cdot
        \plus{ 
        \expect{\sum\nolimits_{i\in[n]}
        \sum\nolimits_{t\in[k\doubleprimed]}
        \randomallocit}
        -
        \demand 
        }
        \\
        &\overset{(c)}{=}
        \overcost
        \cdot
        \plus{ 
        \expect{\sum\nolimits_{i\in[n]}
        \sum\nolimits_{t\in[k\doubleprimed]}
        \randomallocit}
        -
        \pRzero + \perror_{k\doubleprimed}
        }
        \\
        &\overset{(d)}{=}
        \overcost\cdot\lgap 
    \end{align*}
    where inequality~(a) holds by considering overstaffing cost only and lower bounding the total number of hired workers as $\sum_{i\in[n]}\sum_{t\in[k\doubleprimed]}\randomallocit$;
    inequality~(b) holds due to the convexity of $\plus{\cdot}$ and Jensen's inequality;
    equality~(c) holds due to the choice of $\demand$;
    and
    equality~(d) holds due to the construction of $\lgap$.
    This completes the proof of \Cref{lem:opt reduced form lower bound optimal minimax cost}.
\end{proof}

%% file: Paper/appendix-proofs/apx-thmoptalgresolving.tex
\revcolor{
\thmoptalgresolving*
\label{apx:proof-resolving}
\begin{proof}
    We prove the theorem by an induction argument. We claim that for any day $t\in[T]$ and any current state $\left(\mathbf{\bar{\supply}},\mathbf{\bar{\cumalloc}},\bar{\boldsymbol{\wdiscount}},\bar{\pR}\right)$ of that day (see their definitions in {\OPTReS}), the cost guarantee of {\OPTReS} (conditioned on the current state) is equal to the optimal objective value of the following program solved by the algorithm:
       \begin{align}
       \label{eq:lp resolve}
        \begin{array}{llll}
    \min\limits_{\substack{\xbf,\targetcost\geq \zerobf}}\quad\quad
    &
    \targetcost
    & 
    \text{s.t.}
    \\
    &
    \displaystyle\sum\nolimits_{\tau\in [t:T]}\frac{1}{\bar{\wdiscount}_{i\tau}}\alloc_{i\tau} \leq \bar{\supply}_i
    &
    i\in[n]
    \\
    &
    \displaystyle\sum\nolimits_{i\in[n]}\left(\bar{\cumalloc}_i+\sum\nolimits_{\tau\in[t:k]}
    \alloc_{i\tau}\right) \leq \max_{\tau\in[t:k]}\left({\bar{\pR} - \perror_{\tau} - 2\pbias_{\tau}}\right) + \frac{\targetcost}{\overcost}
    \quad\quad
    &
    k\in[t:T]
    \\
    &
    \displaystyle\sum\nolimits_{i\in[n]}\left(\bar{\cumalloc}_i+\sum\nolimits_{\tau\in[t:T]}
    \alloc_{i\tau}\right) \geq \bar{\pR} - \frac{\targetcost}{\undercost}
    &
\end{array}
\end{align}
    We note that this claim implies the theorem: at day $0$, above program reduces to \ref{eq:opt reduced form}, whose optimal objective value is, by \Cref{thm:opt alg}, the optimal minimax cost $\optcost$.
    Below we prove our claim using an induction argument over $t\in[0:T]$.

    \xhdr{Base Case ($t = T$):} In this case, note that demand $\demand$ can take any value such that 
    \begin{align*}
        \bar{\pR} - \perror_T - 2\pbias_T
        \leq \demand \leq \bar{\pR}
    \end{align*}
    due to the construction of the current prediction upper bound $\bar{\pR}$ together with \Cref{asp:prediction sequence}. Hence, given any feasible staffing profile $\{\alloc_{iT}\}_{i\in[n]}$, its staffing cost is 
    \begin{align*}
        &\undercost\cdot\plus{\demand  -\sum\nolimits_{i\in[n]}\cumalloc_i + \alloc_{iT}}
        \vee
        \overcost\cdot \plus{\sum\nolimits_{i\in[n]}\cumalloc_i + \alloc_{iT} - \demand}
        \\
        ={}&
        \undercost\cdot \plus{
        \bar{\pR} -  \left(\sum\nolimits_{i\in[n]}\cumalloc_i + \alloc_{iT}\right)
        }
        \vee
        \overcost\cdot \plus{
        \left(\sum\nolimits_{i\in[n]}\cumalloc_i + \alloc_{iT}\right)
        -
        \bar{\pR} - \perror_T - 2\pbias_T
        }
    \end{align*}
    which is equal to the minimum objective value of program~\eqref{eq:lp resolve} when partial solution $\{\alloc_{iT}\}_{i\in[n]}$ is fixed. Therefore, the claim holds as desired.

    \xhdr{Inductive step ($t\in[0:T-1]$):}
    Assume the claim holds for day $t+1$ and all states for day $t+1$. 
    Fix an arbitrary state $\left(\mathbf{\bar{\supply}},\mathbf{\bar{\cumalloc}},\bar{\boldsymbol{\wdiscount}},\bar{\pR}\right)$ for day $t$, and consider program~\eqref{eq:lp resolve} under this current state. 
    
    Consider a hypothetical execution of {\OPTSim} based on program~\eqref{eq:lp resolve} from day $t$ to day $T$. 
    By construction, {\OPTReS} chooses the same staffing profile as {\OPTSim} on day $t$ (although their staffing profiles may differ in future days). 
    Therefore, both algorithms transition to the same updated state on day $t+1$ for every possible realization of predictions revealed on day $t+1$. 
    By the inductive hypothesis, under this updated state on day $t+1$, the cost guarantee of {\OPTReS} is optimal and hence no larger than that of {\OPTSim}. 
    Consequently, the cost guarantee of {\OPTReS} on day $t$ is no larger than that of {\OPTSim} on day $t$. 
    Moreover, by \Cref{thm:opt alg}, running {\OPTSim} from day $t$ is minimax optimal (conditioned on the current state), and its optimal cost guarantee equals the optimal objective value of program~\eqref{eq:lp resolve} under the current state. 
    Therefore, the cost guarantee of {\OPTReS} on day $t$ is also equal to the optimal objective value of program~\eqref{eq:lp resolve} under the current state, completing the inductive step.

    Therefore, by the base case and the inductive step, our claim follows by induction and the proof of \Cref{thm:opt alg resolving} is completed.
\end{proof}
}

%% file: Paper/appendix-proofs/apx-optalginfty.tex
\optalgmultistation*

In this subsection, we prove \Cref{thm:opt alg multi station} for the multi-station environment. We present two similar but non-identical analysis for the egalitarian and utilitarian staffing cost functions, respectively. For the egalitarian staffing cost, we consider a slightly more general model by allowing station dependent, weakly convex staffing cost functions.

\subsubsection{Egalitarian Staffing Cost}
Here we show \Cref{thm:opt alg multi station} for the egalitarian staffing cost in a more general model. Specifically, we allow station-dependent overstaffing cost function $\overcostj:\reals_+\rightarrow\reals_+$ and understaffing cost function $\undercostj:\reals_+\rightarrow\reals_+$ that are weakly increasing, weakly convex with $\overcostj(0) = \undercostj(0) = 0$. In this generalized model, we modify the constraints of program~\ref{eq:opt reduced form multi station}. Specifically, we generalize its second and third constraints as 
\begin{align*}
    \overcostj\left(\plus{\displaystyle\sum\nolimits_{i\in[n]}\sum\nolimits_{t\in[k]}
    \allocijt - \left(\pRjzero - \perrorjk\right)}\right) \leq  \targetcost_j
    \quad\quad
    &j\in[m], k\in[T]
    \\
    \undercostj\left(\plus{\pRjzero - \displaystyle\sum\nolimits_{i\in[n]}\sum\nolimits_{t\in[T]}
    \allocijt} \right)
    \leq
     \targetcost_j
    \quad\quad
    &j\in[m]
\end{align*}
Since both under/overstaffing cost functions $\undercostj(\cdot)$, $\overcostj(\cdot)$ are convex, the modified version of program~\ref{eq:opt reduced form multi station} is a convex program.

We first show that program~\ref{eq:opt reduced form multi station} is a lower bound of the optimal minimax cost $\optcost$ in the multi-station environment with egalitarian staffing cost.

\begin{lemma}
\label{lem:opt reduced form lower bound optimal minimax cost L infty}
    In the multi-station environment with egalitarian staffing cost, for every (possibly randomized) online algorithm $\ALG$, its cost guarantee is at least the optimal objective value of program~\ref{eq:opt reduced form multi station}.
\end{lemma}
\begin{proof}
    In this argument, we construct a feasible solution of program~\ref{eq:opt reduced form multi station}, whose objective value is equal to the cost guarantee of algorithm $\ALG$.
    
    Consider a prediction sequence subset $\{\predictions\ked\}_{k\in[T]}$ parameterized by $k\in[T]$. Specifically, for each $k\in[T]$, prediction sequence $\predictions\ked = \{[\pLjt\ked,\pRjt\ked]\}_{j\in[m],t\in[T]}$ is constructed as follows:
    \begin{align*}
    \begin{array}{lll}
        t\in[k],j\in[m]:\qquad&
        \pLjt\ked \gets \pR_{j,t-1}\ked - \perrorjt,\quad 
        &\pRjt\ked \gets \pR_{j,t-1}\ked;
        \\
        t\in[k+1:T],j\in[m]:
        \qquad
        &
        \pLjt\ked \gets \pL_{j,t-1}\ked,\quad 
        &\pRjt\ked \gets \pL_{j,t-1}\ked + \perrorjt.
    \end{array}
    \end{align*}
    where $\pLjzero\ked = \pLjzero$ and $\pRjzero\ked = \pRjzero$. In short, this prediction sequence subset is constructed such that prediction $\{[\pLjt\ked,\pRjt\ked]\}_{j\in[m]}$ on every day $t$ are the same for all prediction sequence $\predictions\ked$ with $k\geq t$. Thus, no online algorithm can distinguish them. Namely, the staffing decision in each day $t$ should be the same under all prediction sequences $\predictions\ked$ with $k \geq t$. 
    
    Motivated by the prediction sequence construction above, we let random variable $\randomallocijt$ be the number of workers hired by the algorithm from pool $i$ to station $j$ in day $t$ under prediction sequence $\predictions\Ted$.
    Due to the feasibility of the algorithm under prediction sequence $\predictions\Ted$, for all sample paths (over the randomness of the algorithm), we have
    \begin{align*}
    \forall i\in[n]:\quad &
            \displaystyle\sum\nolimits_{t\in[T]}\frac{1}{\wdiscountit}\cdot\sum\nolimits_{j\in[m]}\randomallocijt \leq \supplyi
    \end{align*}
    Now consider the following solution $(\xbf,\targetcosts)$ construction with auxiliary variables $\{\lgapj,\rgapj\}_{j\in[m]}$:
    \begin{align*}
        i\in[n],j\in[m],t\in[T]:&
        \qquad
        \allocijt \gets \expect{\randomallocijt}
        \\
        j\in[m]:&
        \qquad
        \lgapj \gets
        \max_{k\in[T]}~
        \plus{
        \expect{\sum\nolimits_{i\in[n]}\sum\nolimits_{t\in[k]}\randomallocijt} - \pRjzero + \perrorjk
        }
        \\
        j\in[m]:
        & \qquad \rgapj \gets 
        \plus{
        \pRjzero - \expect{\sum\nolimits_{i\in[n]}\sum\nolimits_{t\in[T]}\randomallocijt}
        }
        \\
        j\in[m]:
        & \qquad
        \targetcost_j \gets \max\{\overcostj(\lgapj),\undercostj(\rgapj)\}
    \end{align*}
    By construction, all three constraints are satisfied. 
    Below we argue the objective value of the constructed solution is at most the cost guarantee of algorithm $\ALG$.

    Let $j\primed = \argmax_{j\in[m]} \undercostj(\rgapj)$ and 
    $j\doubleprimed = \argmax_{j\in[m]} \overcostj(\lgapj)$.
    Now we consider two cases.

    \paragraph{Case 1- $\undercost_{j\primed}(\rgap_{j\primed}) \geq 
    \overcost_{j\doubleprimed}(\lgap_{j\doubleprimed})$.}
    In this case, the objective value of the constructed solution is 
    $\undercost_{j\primed}(\rgap_{j\primed})$.
    Consider the execution of algorithm $\ALG$ under prediction sequence $\predictions\Ted$ and the demand profile $\dbf\primed \triangleq \{\pR_{jT}\Ted\}_{j\in[m]} = \{\pRjzero\}_{j\in[m]}$. Note that the staffing cost can be lower bounded as 
    \begin{align*}
        \expect{\CostInf[\dbf\primed]{\ALG(\predictions\Ted)}}
        &\overset{(a)}{\geq}
        \expect{
        \max_{j\in[m]} 
        \undercostj 
        \left(
        \plus{
        \demandj\primed - 
        \sum\nolimits_{i\in[n]}
        \sum\nolimits_{t\in[T]}
        \randomallocijt 
        }
        \right)
        }
        \\
        &\overset{(b)}{\geq}
        \max_{j\in[m]} 
        \undercostj
        \left(
        \plus{ 
        \demandj\primed - 
        \expect{\sum\nolimits_{i\in[n]}
        \sum\nolimits_{t\in[T]}
        \randomallocijt}
        }
        \right)
        \\
        &\overset{(c)}{=}
        \max_{j\in[m]} 
        \undercostj
        \left(
        \plus{ 
        \pRjzero - 
        \expect{\sum\nolimits_{i\in[n]}
        \sum\nolimits_{t\in[T]}
        \randomallocijt}
        }
        \right)
        \\
        &\overset{(d)}{=}
        \undercost_{j\primed}(\rgap_{j\primed}) 
    \end{align*}
    where inequality~(a) holds by considering understaffing cost only;
    inequality~(b) holds due to the convexity of $\max\{\cdot\}$, $\undercostj(\cdot)$,
    $\plus{\cdot}$ and Jensen's inequality;
    equality~(c) holds due to the construction of $\demandj\primed$;
    and
    equality~(d) holds due to the construction of $\rgap_{j\primed}$.

    \paragraph{Case 2- $\undercost_{j\primed}(\rgap_{j\primed}) < 
    \overcost_{j\doubleprimed}(\lgap_{j\doubleprimed})$.}
    In this case, the objective value of the constructed solution is 
    $\overcost_{j\doubleprimed}(\lgap_{j\doubleprimed})$.
    Let $k\doubleprimed$ be the index such that $\lgap_{j\doubleprimed} = \plus{\expect{\sum_{i\in[n],t\in[k]}\randomallocijt - \pRjzero + \perrorjk}}$.
    Consider the execution of algorithm $\ALG$ under prediction sequence $\predictions^{(k\doubleprimed)}$ and the demand profile $\dbf^{(k\doubleprimed)} \triangleq \{\pL_{jT}^{(k\doubleprimed)}\}_{j\in[m]} = \{\pRjzero - \perror_{jk\doubleprimed}\}_{j\in[m]}$. Note that the staffing cost can be lower bounded as 
    \begin{align*}
        \expect{\CostInf[\dbf^{(k\doubleprimed)}]{\ALG(\predictions^{(k\doubleprimed)})}}
        &\overset{(a)}{\geq}
        \expect{
        \max_{j\in[m]} 
        \overcostj 
        \left(
        \plus{
        \sum\nolimits_{i\in[n]}
        \sum\nolimits_{t\in[k\doubleprimed]}
        \randomallocijt 
        -
        \demandj^{(k\doubleprimed)} 
        }
        \right)
        }
        \\
        &\overset{(b)}{\geq}
        \max_{j\in[m]} 
        \overcostj
        \left(
        \plus{ 
        \expect{\sum\nolimits_{i\in[n]}
        \sum\nolimits_{t\in[k\doubleprimed]}
        \randomallocijt}
        -
        \demandj^{(k\doubleprimed)} 
        }
        \right)
        \\
        &\overset{(c)}{=}
        \max_{j\in[m]} 
        \overcostj
        \left(
        \plus{ 
        \expect{\sum\nolimits_{i\in[n]}
        \sum\nolimits_{t\in[k\doubleprimed]}
        \randomallocijt}
        -
        \pRjzero + \perror_{jk\doubleprimed}
        }
        \right)
        \\
        &\overset{(d)}{=}
        \overcost_{j\doubleprimed}(\lgap_{j\doubleprimed}) 
    \end{align*}
    where inequality~(a) holds by considering understaffing cost only and lower bounding the number of hired worker in each station $j$ as $\sum_{i\in[n]}\sum_{t\in[k\doubleprimed]}\randomallocijt$;
    inequality~(b) holds due to the convexity of $\max\{\cdot\}$, $\overcostj(\cdot)$,
    $\plus{\cdot}$ and Jensen's inequality;
    equality~(c) holds due to the construction of $\demandj^{(k\doubleprimed)}$;
    and
    equality~(d) holds due to the construction of $\lgap_{j\doubleprimed}$.
\end{proof}
Next we argue that the cost guarantee of {\OPTSimInfty} is upper bounded by program~\ref{eq:opt reduced form multi station}.
\begin{lemma}
\label{lem:opt alg minimax cost L infty}
    In the multi-station environment with egalitarian staffing cost, the cost guarantee of {\OPTSimInfty} is at most the optimal objective value of program~\ref{eq:opt reduced form multi station}.
\end{lemma}
\begin{proof}
    Let $(\xbf^*, \targetcosts^*)$ be the optimal solution of program~\ref{eq:opt reduced form multi station} used in {\OPTSimInfty}. Moreover, 
    It suffices to show that for every prediction sequence $\predictions$ and demand profile $\dbf$,
    for every station $j$, the total number of hired workers  $\sum_{i\in[n]}\sum_{t\in[T]}\allocijt$ for this station satisfies that
    \begin{align*}
    \undercostj\left(\plus{\demandj - \sum\nolimits_{i\in[n]}\sum\nolimits_{t\in[T]}\allocijt}\right) \leq \targetcost_j^*
    \;\;
    \mbox{and}
    \;\;
    \overcostj\left(\plus{\sum\nolimits_{i\in[n]}\sum\nolimits_{t\in[T]}\allocijt - \demandj}\right) \leq \targetcost_j^*
    \end{align*}
    which are implied by 
    \begin{align*}
    \undercostj\left(\plus{\pRjT - \sum\nolimits_{i\in[n]}\sum\nolimits_{t\in[T]}\allocijt}\right) \leq \targetcost_j^*
    \;\;
    \mbox{and}
    \;\;
    \overcostj\left(\plus{\sum\nolimits_{i\in[n]}\sum\nolimits_{t\in[T]}\allocijt - \pLjT}\right) \leq \targetcost_j^*
    \end{align*}
    due to the consistency condition (Assumption~\ref{asp:prediction sequence} with $\pbiass = \zerobf$) of prediction sequence $\predictions$.
    Invoking the bounded over/understaffing cost properties in \Cref{lem:emulator} and second, third constraints about $\targetcosts^*$ in program~\ref{eq:opt reduced form multi station} proves the two inequalities above as desired.
\end{proof}

Combining \Cref{lem:opt reduced form lower bound optimal minimax cost L infty} and \Cref{lem:opt alg minimax cost L infty}, we prove \Cref{thm:opt alg multi station} for the egalitarian staffing cost as desired.

%% file: Paper/appendix-proofs/apx-optalgone.tex

\subsubsection{Utilitarian Staffing Cost}
Now we prove \Cref{thm:opt alg multi station} for the utilitarian staffing cost.
We first show that program~\ref{eq:opt reduced form multi station} is a lower bound of the optimal minimax cost $\optcost$ in the multi-station environment with utilitarian staffing cost.

\begin{lemma}
\label{lem:opt reduced form lower bound optimal minimax cost L one}
    In the multi-station environment with utilitarian-staffing cost, for every (possibly randomized) online algorithm $\ALG$, its cost guarantee is at least the optimal objective value of program~\ref{eq:opt reduced form multi station}.
\end{lemma}
\Cref{lem:opt reduced form lower bound optimal minimax cost L one} is proved by the following \Cref{lem:opt reduced form lower bound optimal minimax cost L one modified,lem:L infty overcost vs undercost} that analyze the modified program~\ref{eq:opt reduced form L one modified} with variables $\{\allocijt,\lgapjk,\rgapj\}_{i\in[n],j,k\in[m],t\in[T]}$ defined as follows:
\begin{align}
\tag{$\textsc{LP-multi-util}\primed$}
    \label{eq:opt reduced form L one modified}
    &\arraycolsep=1.4pt\def\arraystretch{2.2}
    \begin{array}{llll}
    \min\limits_{\substack{\xbf,\lambdabf,\thetabf\geq \zerobf}}
    \qquad
    &
    \displaystyle
    \sum\nolimits_{j\in[m]}
    \undercost\cdot \rgapj
    \vee
    \left(
    \max_{k\in[T]}
    \sum\nolimits_{j\in[m]}
    \overcost\cdot \lgapjk
    \right)
    & 
    \text{s.t.}
    \\
    &
    \displaystyle\sum\nolimits_{t\in[T]}\frac{1}{\wdiscountit}\cdot\sum\nolimits_{j\in[m]}\allocijt \leq \supplyi
    &
    i\in[n]
    \\
    &
    \displaystyle\sum\nolimits_{i\in[n]}\sum\nolimits_{t\in[k]}
    \allocijt \leq \pRjzero - \perrorjk + \lgapjk
    \quad
    &
    j\in[m],k\in[T]
    \\
    &
    \displaystyle\sum\nolimits_{i\in[n]}\sum\nolimits_{t\in[T]}
    \allocijt \geq \pRjzero - \rgapj
    &
    j\in[m]
\end{array}
\end{align}

\begin{lemma}
\label{lem:opt reduced form lower bound optimal minimax cost L one modified}
    In the multi-station environment with utilitarian-staffing cost, for every (possibly randomized) online algorithm $\ALG$, its cost guarantee is at least the optimal objective value of program~\ref{eq:opt reduced form L one modified}.
\end{lemma}
\begin{proof}
    The argument is similar to the ones for \Cref{lem:opt reduced form lower bound optimal minimax cost,lem:opt reduced form lower bound optimal minimax cost L infty}. Specifically, we construct a feasible solution of program~\ref{eq:opt reduced form L one modified}, whose objective value is equal to the cost guarantee of algorithm $\ALG$.
    
    Consider a prediction sequence subset $\{\predictions\ked\}_{k\in[T]}$ parameterized by $k\in[T]$. Specifically, for each $k\in[T]$, prediction sequence $\predictions\ked = \{[\pLjt\ked,\pRjt\ked]\}_{j\in[m],t\in[T]}$ is constructed as follows:
    \begin{align*}
    \begin{array}{lll}
        t\in[k],j\in[m]:\qquad&
        \pLjt\ked \gets \pR_{j,t-1}\ked - \perrorjt,\quad 
        &\pRjt\ked \gets \pR_{j,t-1}\ked;
        \\
        t\in[k+1:T],j\in[m]:
        \qquad
        &
        \pLjt\ked \gets \pL_{j,t-1}\ked,\quad 
        &\pRjt\ked \gets \pL_{j,t-1}\ked + \perrorjt.
    \end{array}
    \end{align*}
    where $\pLjzero\ked = \pLjzero$ and $\pRjzero\ked = \pRjzero$. In short, this prediction sequence subset is constructed such that prediction $\{[\pLjt\ked,\pRjt\ked]\}_{j\in[m]}$ on every day $t$ are the same for all prediction sequence $\predictions\ked$ with $k\geq t$. Thus, no online algorithm can distinguish them. Namely, the staffing decision in each day $t$ should be the same under all prediction sequences $\predictions\ked$ with $k \geq t$. 
    
    Motivated by the prediction sequence construction above, we let random variable $\randomallocijt$ be the number of workers hired by the algorithm from pool $i$ to station $j$ in day $t$ under prediction sequence $\predictions\Ted$.
    Due to the feasibility of the algorithm under prediction sequence $\predictions\Ted$, for all sample paths (over the randomness of the algorithm), we have
    \begin{align*}
    \forall i\in[n]:\quad &
            \displaystyle\sum\nolimits_{t\in[T]}\frac{1}{\wdiscountit}\cdot\sum\nolimits_{j\in[m]}\randomallocijt \leq \supplyi
    \end{align*}
    Now consider the following solution $(\xbf,\lambdabf,\thetabf)$ construction:
    \begin{align*}
        i\in[n],j\in[m],t\in[T]:&
        \qquad
        \allocijt \gets \expect{\randomallocijt}
        \\
        j\in[m],k\in[T]:&
        \qquad
        \lgapjk \gets
        \plus{
        \expect{\sum_{i\in[n]}\sum\nolimits_{t\in[k]}\randomallocijt} - \pRjzero + \perrorjk
        }
        \\
        j\in[m]:
        & \qquad \rgapj \gets 
        \plus{
        \pRjzero - \expect{\sum\nolimits_{i\in[n]}\sum\nolimits_{t\in[T]}\randomallocijt}
        }
    \end{align*}
    By construction, all four constraints are satisfied. 
    Below we argue the objective value of the constructed solution is at most the cost guarantee of algorithm $\ALG$.

    Let $k\primed = \argmax_{k\in[T]} \sum_{j\in[m]} \overcost\cdot \lgapjk$.
    Now we consider two cases.

    \paragraph{Case 1- $\sum_{j\in[m]}\undercost\cdot \rgapj \geq 
    \sum_{j\in[m]}\overcost\cdot\lgap_{jk\doubleprimed}$.}
    In this case, the objective value of the constructed solution is 
    $\sum_{j\in[m]}\undercost\cdot \rgapj$.
    Consider the execution of algorithm $\ALG$ under prediction sequence $\predictions\Ted$ and the demand profile $\dbf\primed \triangleq \{\pR_{jT}\Ted\}_{j\in[m]} = \{\pRjzero\}_{j\in[m]}$. Note that the staffing cost can be lower bounded as 
    \begin{align*}
        \expect{\CostOne[\dbf\primed]{\ALG(\predictions\Ted)}}
        &\overset{(a)}{\geq}
        \expect{
        \sum\nolimits_{j\in[m]} 
        \undercost
        \cdot
        \plus{
        \demandj\primed - 
        \sum\nolimits_{i\in[n]}
        \sum\nolimits_{t\in[T]}
        \randomallocijt 
        }
        }
        \\
        &\overset{(b)}{\geq}
        \sum\nolimits_{j\in[m]} 
        \undercost
        \cdot
        \plus{ 
        \demandj\primed - 
        \expect{\sum\nolimits_{i\in[n]}
        \sum\nolimits_{t\in[T]}
        \randomallocijt}
        }
        \\
        &\overset{(c)}{=}
        \sum\nolimits_{j\in[m]} 
        \undercost
        \cdot
        \plus{ 
        \pRjzero - 
        \expect{\sum\nolimits_{i\in[n]}
        \sum\nolimits_{t\in[T]}
        \randomallocijt}
        }
        \\
        &\overset{(d)}{=}
        \sum\nolimits_{j\in[m]}\undercost\cdot \rgapj
    \end{align*}
    where inequality~(a) holds by considering understaffing cost only;
    inequality~(b) holds due to the convexity of 
    $\plus{\cdot}$ and Jensen's inequality;
    equality~(c) holds due to the construction of $\demandj\primed$;
    and
    equality~(d) holds due to the construction of $\rgapj$.

    \paragraph{Case 2- $\sum_{j\in[m]}\undercost\cdot \rgapj < 
    \sum_{j\in[m]}\overcost\cdot\lgap_{jk\doubleprimed}$.}
    In this case, the objective value of the constructed solution is 
    $\sum_{j\in[m]}\overcost\cdot\lgap_{jk\doubleprimed}$.
    Consider the execution of algorithm $\ALG$ under prediction sequence $\predictions^{(k\doubleprimed)}$ and the demand profile $\dbf^{(k\doubleprimed)} \triangleq \{\pL_{jT}^{(k\doubleprimed)}\}_{j\in[m]} = \{\pRjzero - \perror_{jk\doubleprimed}\}_{j\in[m]}$. Note that the staffing cost can be lower bounded as 
    \begin{align*}
        \expect{\CostOne[\dbf^{(k\doubleprimed)}]{\ALG(\predictions^{(k\doubleprimed)})}}
        &\overset{(a)}{\geq}
        \expect{
        \sum\nolimits_{j\in[m]} 
        \overcost 
        \cdot
        \plus{
        \sum\nolimits_{i\in[n]}
        \sum\nolimits_{t\in[k\doubleprimed]}
        \randomallocijt 
        -
        \demandj^{(k\doubleprimed)} 
        }
        }
        \\
        &\overset{(b)}{\geq}
        \sum\nolimits_{j\in[m]} 
        \overcost 
        \cdot
        \plus{ 
        \expect{\sum\nolimits_{i\in[n]}
        \sum\nolimits_{t\in[k\doubleprimed]}
        \randomallocijt}
        -
        \demandj^{(k\doubleprimed)} 
        }
        \\
        &\overset{(c)}{=}
        \sum\nolimits_{j\in[m]} 
        \overcost 
        \cdot
        \plus{ 
        \expect{\sum\nolimits_{i\in[n]}
        \sum\nolimits_{t\in[k\doubleprimed]}
        \randomallocijt}
        -
        \pRjzero + \perror_{jk\doubleprimed}
        }
        \\
        &\overset{(d)}{=}
        \sum\nolimits_{j\in[m]}\overcost\cdot\lgap_{jk\doubleprimed}
    \end{align*}
    where inequality~(a) holds by considering understaffing cost only and lower bounding the number of hired worker in each station $j$ as $\sum_{i\in[n]}\sum_{t\in[k\doubleprimed]}\randomallocijt$;
    inequality~(b) holds due to the convexity of
    $\plus{\cdot}$ and Jensen's inequality;
    equality~(c) holds due to the construction of $\demandj^{(k\doubleprimed)}$;
    and
    equality~(d) holds due to the construction of $\lgap_{j\doubleprimed}$.
\end{proof}

Next we identify four properties of the optimal solution of program~\ref{eq:opt reduced form L one modified} in \Cref{lem:L infty overcost vs undercost}. Clearly, properties (i) and (iii) are implied by properties (ii) and (iv). Nonetheless, we list all of them, since in our argument, we start with an arbitrary optimal solution, and then introduce a series of modifications to convert the original optimal solution to another optimal solution satisfying properties (i) through (iv) step by step.
\begin{lemma}
    \label{lem:L infty overcost vs undercost}
    There exists an optimal solution $(\xbf^*, \lambdabf^*, \thetabf^*)$ of program~\ref{eq:opt reduced form L one modified} satisfying the following four properties:
    \begin{enumerate}
        \item[(i)] $\sum_{j\in[m]}\lgapjk^* \geq \sum_{j\in[m]}\lgap_{j,k+1}^*$ for all $k\in[T - 1]$;
        \item[(ii)] $\lgapjk^* \geq \lgap_{j,k + 1}^*$ for all $k\in[T - 1]$
        \item[(iii)] $\sum_{j\in[m]}\overcost\cdot \lgapjk^* \leq \sum_{j\in[m]}\undercost\cdot\rgapj^*$ for all $k\in[T]$;
        \item[(iv)] $\overcost\cdot \lgapjk^* \leq \undercost\cdot \rgapj^*$ for all $j\in[m],k\in[T]$.
    \end{enumerate}
\end{lemma}
\begin{proof}
    It is easy to verify that there exists an optimal solution $(\xbf^*, \lambdabf^*, \thetabf^*)$ of program~\ref{eq:opt reduced form L one modified} such that $\lgapjk^* = \plus{\sum_{i\in[n]}\sum_{t\in[k]} \allocijt^* - \pRjzero + \perrorjk}$
    and $\rgapj^* = \plus{\pRjzero - \sum_{i\in[n]}\sum_{t\in[T]} \allocijt^*}$
    for all $j\in[m]$, $k\in[T]$. Going forward, we assume these two equalities always hold.

    \xhdr{Step i- obtaining property~(i):}    
    In this step we introduce a modification procedure that converts the original optimal solution $(\xbf^*, \lambdabf^*, \thetabf^*)$ into another optimal solution $(\xbf\primed, \lambdabf\primed, \thetabf\primed)$ that satisfies property~(i) in the lemma statement.

    Suppose there exists $k \in[T - 1]$ such that 
    \begin{align*}
        \sum\nolimits_{j\in[m]}\lgapjk^* < \sum\nolimits_{j\in[m]}\lgap_{j,k+1}^*
    \end{align*}
    By definition, there must exist $0\leq \lgapjk^* < \lgap_{j,k+1}^*$ and thus $\alloc_{ij,k+1}^* > 0$ for some pool $i\in[n]$ and station $j\in[m]$. Now we modify variables $\xbf^*$ into variables $\xbf\doubleprimed$ where we set $\alloc_{ijk}\doubleprimed \gets \alloc_{ijk}^* + \epsilon$, $\alloc_{ij,k+1}\doubleprimed \gets \alloc_{ij,k+1}^* - \epsilon$ for sufficiently small $\epsilon > 0$ while holding all other variables in $\xbf^*$ fixed.
    Finally, we set $\lgapjk\doubleprimed = \plus{\sum_{i\in[n]}\sum_{t\in[k]} \allocijt\doubleprimed - \pRjzero + \perrorjk}$
        and $\rgapj\doubleprimed = \plus{\pRjzero - \sum_{i\in[n]}\sum_{t\in[T]} \allocijt\doubleprimed}$ for all $j\in[m]$, $k\in[T]$.
    By construction, the modified solution is still feasible and achieves the same objective value. Moreover, term $\sum\nolimits_{j\in[m]}\lgapjk\doubleprimed > \sum\nolimits_{j\in[m]}\lgapjk^*$, i.e., strictly increases;
    while term $\sum\nolimits_{j\in[m]}\lgap_{j\ell} = \sum\nolimits_{j\in[m]}\lgap_{j\ell}^*$, i.e., remains the same, for all $\ell\not=k$. 
    
    Repeating the procedure above, we obtain an optimal solution $(\xbf\primed, \lambdabf\primed, \thetabf\primed)$ that satisfies property~(i) in the lemma statement as desired.

    \xhdr{Step ii- obtaining property~(ii):}  In this step we introduce a modification procedure that converts the optimal solution $(\xbf\primed, \lambdabf\primed, \thetabf\primed)$ in step (i) into another optimal solution $(\xbf\doubleprimed, \lambdabf\doubleprimed, \thetabf\doubleprimed)$ that satisfies properties~(i) and (ii) in the lemma statement.

    Let $k\in[T - 1]$ be the largest index such that there exists station $j\in[m]$ satisfying
    \begin{align*}
        \lgapjk\primed < \lgap_{j,k+1}\primed
    \end{align*}
    Due to property~(i),
    there must exist another station $j'\not=j$ such that
    \begin{align*}
        \lgap_{j'k}\primed > \lgap_{j',k+1}\primed
    \end{align*}
    Note that $\lgapjk\primed < \lgap_{j,k+1}\primed$ further implies $\alloc_{ij,k+1}\primed > 0$ for some pool $i\in[n]$.
    Let $\tau$ be the largest index such that $\tau \leq k$ and $\alloc_{i'j'\tau}\primed > 0$ for some pool $i'\in[n]$. Since $\lgap_{j'k}\primed > \lgap_{j',k+1}\primed \geq 0$, index $\tau$ must exist. By the definition of index $\tau$, we also have $\lgap_{j'\ell}\primed \geq \lgap_{j'k}\primed > 0$ for every index $\ell\in[\tau:k]$, since variables $\{\alloc_{ij'\ell}\primed\}$ are all zero and $\perror_{j'\ell}$ is weakly decreasing for $\ell\in[\tau:k]$.
    Now we modify variables $\xbf\primed$ into variables $\xbf\doubleprimed$ where we set $\alloc_{i'j\tau}\doubleprimed\gets \alloc_{i'j\tau}\primed + \epsilon$, $\alloc_{ij,k+1}\doubleprimed \gets \alloc_{ij,k+1}\primed - \epsilon$, $\alloc_{i'j'\tau}\doubleprimed \gets \alloc_{i'j'\tau}\primed - \epsilon$, $\alloc_{ij',k+1}\doubleprimed \gets \alloc_{ij',k+1}\primed + \epsilon$ for sufficiently small $\epsilon > 0$ while holding all other variables in $\xbf\primed$ fixed. Finally, we set $\lgapjk\doubleprimed = \plus{\sum_{i\in[n]}\sum_{t\in[k]} \allocijt\doubleprimed - \pRjzero + \perrorjk}$
        and $\rgapj\doubleprimed = \plus{\pRjzero - \sum_{i\in[n]}\sum_{t\in[T]} \allocijt\doubleprimed}$ for all $j\in[m]$, $k\in[T]$.
    By construction, the modified solution is still feasible,
    achieves a weakly smaller objective value, and still satisfies property~(i).
    Moreover, term $\lgap_{j'\ell}\doubleprimed < \lgap_{j'\ell}\primed$, i.e., strictly decreases;
    term $\lgap_{j\ell}\doubleprimed \geq \lgap_{j\ell}\primed$, i.e., weakly increases for every $\ell\in[\tau:k]$; 
    while all other variables in $\lambdabf\doubleprimed$ remains the same.

    Repeating the procedure above, we obtain an optimal solution $(\xbf\doubleprimed, \lambdabf\doubleprimed, \thetabf\doubleprimed)$ that satisfies properties~(i) and (ii) in the lemma statement.

    \xhdr{Step iii- obtaining property~(iii):}  In this step we introduce a modification procedure that converts the optimal solution $(\xbf\doubleprimed, \lambdabf\doubleprimed, \thetabf\doubleprimed)$ in step (ii) into another optimal solution $(\xbf\csuitted, \lambdabf\csuitted, \thetabf\csuitted)$ that satisfies properties~(i) (ii) and (iii) in the lemma statement.

    Due to property~(i), it suffices to compare $\sum_{j\in[m]}\overcost\cdot \lgap_{j1}\doubleprimed$ 
    and $\sum_{j\in[m]}\undercost\cdot\rgapj\doubleprimed$. Suppose 
    \begin{align*}
        \sum\nolimits_{j\in[m]}\overcost\cdot \lgap_{j1}\doubleprimed > \sum\nolimits_{j\in[m]}\undercost\cdot\rgapj\doubleprimed
    \end{align*}
    Then, there must exist station $j\in[m]$ such that 
    \begin{align*}
        \overcost\cdot \lgap_{j1}\doubleprimed > \undercost\cdot \rgapj\doubleprimed \geq 0
    \end{align*}
    which further implies $\alloc_{ij1}\doubleprimed > 0$ for some pool $i\in[n]$. 
    Now we modify variables $\xbf\doubleprimed$ into variables $\xbf\csuitted$ where we set $\alloc_{ij1}\csuitted\gets \alloc_{ij1}\doubleprimed - \epsilon$ for sufficiently small $\epsilon > 0$ while holding all other variables in $\xbf\doubleprimed$ fixed. Finally, we set $\lgapjk\csuitted = \plus{\sum_{i\in[n]}\sum_{t\in[k]} \allocijt\csuitted - \pRjzero + \perrorjk}$
    and $\rgapj\csuitted = \plus{\pRjzero - \sum_{i\in[n]}\sum_{t\in[T]} \allocijt\csuitted}$ for all $j\in[m]$, $k\in[T]$. 
    By construction, the modified solution is still feasible, achieves a weakly smaller objective value, and still satisfies properties (i) and (ii). Moreover, term $\lgap_{jk}\csuitted \leq \lgap_{jk}\doubleprimed$, i.e., weakly decreases for every $k\in[T]$;
    term $\rgap_{j}\csuitted > \rgap_{j}\doubleprimed$, i.e., strictly increases; 
    while all other variables in $\lambdabf,\thetabf$ remains the same.
    
    Repeating the procedure above, we obtain an optimal solution $(\xbf\csuitted, \lambdabf\csuitted, \thetabf\csuitted)$ that satisfies properties~(i) (ii) and (iii) in the lemma statement.

    \xhdr{Step iv- obtaining property~(iv):} In this step we introduce a modification procedure that converts the optimal solution $(\xbf\csuitted, \lambdabf\csuitted, \thetabf\csuitted)$ in step (iii) into another optimal solution $(\xbf\spsuitted, \lambdabf\spsuitted, \thetabf\spsuitted)$ that satisfies all four properties in the lemma statement.

    Due to property~(ii), it suffices to compare $\overcost\cdot \lgap_{j1}\csuitted$ 
    and $\undercost\cdot\rgapj\csuitted$ for all $j\in[m]$. 
    Suppose there exists station $j\in[m]$ such that 
    \begin{align*}
        \overcost\cdot \lgap_{j1}\csuitted >\undercost\cdot\rgapj\csuitted\geq 0
    \end{align*}
    which further implies $\alloc_{ij1}\csuitted > 0$ for some pool $i\in[n]$.
    If there are multiple such index $j$, we select the one with the largest $|\{k\in[T]:\lgap_{jk}\csuitted > 0\}|$. 
    Due to property~(iii), there must exists another station $j'\not=j$ such that
    \begin{align*}
        \overcost\cdot \lgap_{j'1}\csuitted < \undercost\cdot\rgap_{j'}\csuitted
    \end{align*}
    Now we modify variables $\xbf\csuitted$ into variables $\xbf\spsuitted$ where we set $\alloc_{ij1}\spsuitted\gets \alloc_{ij1}\spsuitted - \epsilon$, $\alloc_{ij'1}\spsuitted\gets \alloc_{ij'1}\spsuitted + \epsilon$ for sufficiently small $\epsilon > 0$ while holding all other variables in $\xbf\csuitted$ fixed. Finally, we set $\lgapjk\spsuitted = \plus{\sum_{i\in[n]}\sum_{t\in[k]} \allocijt\spsuitted - \pRjzero + \perrorjk}$
    and $\rgapj\spsuitted = \plus{\pRjzero - \sum_{i\in[n]}\sum_{t\in[T]} \allocijt\spsuitted}$ for all $j\in[m]$, $k\in[T]$. 
    By construction, the modified solution is still feasible and achieves a weakly smaller objective value.
    
    To see why the objective value weakly decreases, first note that $\sum_{j\in[m]}\rgapj\spsuitted$ remains the same as $\sum_{j\in[m]}\rgapj\csuitted$. Let $k$ and $k'$ be the largest index such that $\lgapjk\csuitted > 0$ and $\lgap_{j'k'}\csuitted > 0$, respectively. Due to property~(ii), we know $\lgap_{j\ell}\csuitted > 0$ and $\lgap_{j'\ell'} \csuitted >0$ for every $\ell\in[k],\ell'\in[k']$. Note that by our modification procedure, $\lgap_{j\ell}\spsuitted$ decreases by $\overcost\cdot \epsilon$ for every $\ell\in[k]$ and $\lgap_{j'\ell'}\spsuitted$ increases by $\overcost\cdot\epsilon$ for every $\ell'\in[k']$. Clearly, for every $\ell\not\in[k:k']$, $\sum_{s\in[m]}\overcost\cdot\lgap_{s\ell}\spsuitted \leq \sum_{s\in[m]}\overcost\cdot\lgap_{s\ell}\csuitted$ as desired. For every $\ell\in[k + 1:k']$, we claim that $\sum_{s\in[m]}\overcost\cdot\lgap_{s\ell}\csuitted < \sum_{s\in[m]}\undercost\cdot\rgap_{s}\csuitted$ and thus $\sum_{s\in[m]}\overcost\cdot\lgap_{s\ell}\spsuitted \leq \sum_{s\in[m]}\undercost\cdot\rgap_{s}\spsuitted$ as desired. This claim can be shown by contradiction, suppose $\sum_{s\in[m]}\overcost\cdot\lgap_{s\ell}\csuitted \geq \sum_{s\in[m]}\undercost\cdot\rgap_{s}\csuitted$ for some $\ell\in[k + 1:k']$. The definition of $k$ implies $\lgap_{j\ell}\csuitted = 0$ and thus there exists $j''$ such that $\overcost\cdot\lgap_{j''\ell}\csuitted > \undercost\cdot\rgap_{j''}\csuitted$, which is a contradiction since station $j$ is selected among all stations such that $\overcost\cdot\lgap_{j1}\csuitted > \undercost\cdot\rgapj\csuitted$ with the largest $|\{k\in[T]:\lgap_{jk}\csuitted > 0\}|$. Using this argument, we claim that the modified solution still satisfies properties~(i) (ii) and (iii).

    Finally, since after this modification procedure, $\overcost\cdot \lgap_{j1}\spsuitted$ strictly decreases, $\undercost\cdot\rgapj\spsuitted$ strictly increases, while $\overcost\cdot \lgap_{j'1}\spsuitted$ is still weakly smaller than $\undercost\cdot\rgap_{j'}\spsuitted$, we conclude that 
    by repeating the procedure, an optimal solution $(\xbf\spsuitted, \lambdabf\spsuitted, \thetabf\spsuitted)$ that satisfies all four in the lemma statement is obtained as desired.
\end{proof}

Now we are ready to prove \Cref{lem:opt reduced form lower bound optimal minimax cost L one}.
\begin{proof}[\textsl{Proof of \Cref{lem:opt reduced form lower bound optimal minimax cost L one}.}]
    Let $(\xbf^*, \lambdabf^*, \thetabf^*)$ be an optimal solution of program~\ref{eq:opt reduced form L one modified} satisfying property~(iv) in \Cref{lem:L infty overcost vs undercost}. Clearly, it is also a feasible solution in program~\ref{eq:opt reduced form multi station} and its objective value in two programs is the same. Therefore, the optimal objective value of program~\ref{eq:opt reduced form multi station} is at most the optimal objective value of program~\ref{eq:opt reduced form L one modified}. Finally, invoking \Cref{lem:opt reduced form lower bound optimal minimax cost L one modified} finishes the proof.
\end{proof}

Next we argue that the cost guarantee of \OPTSimOne\ is upper bounded by program~\ref{eq:opt reduced form multi station}.

\begin{lemma}
\label{lem:opt alg minimax cost L one}
    In the multi-station environment with utilitarian-staffing cost, the cost guarantee of \OPTSimOne\ is at most the optimal objective value of program~\ref{eq:opt reduced form multi station}.
\end{lemma}
\begin{proof}
    Let $(\xbf^*, \lambdabf^*, \thetabf^*)$ be the optimal solution of program~\ref{eq:opt reduced form multi station} used in \OPTSimOne.
    It suffices to show that for every prediction sequence $\predictions$ and demand profile $\dbf$,
    for every station $j$, the total number of hired workers  $\sum_{i\in[n]}\sum_{t\in[T]}\allocijt$ for this station satisfies that
    \begin{align*}
    \plus{\demandj - \sum\nolimits_{i\in[n]}\sum\nolimits_{t\in[T]}\allocijt} \leq \rgapj^*
    \;\;
    \mbox{and}
    \;\;
    \plus{\sum\nolimits_{i\in[n]}\sum\nolimits_{t\in[T]}\allocijt - \demandj} \leq \rgapj^*
    \end{align*}
    which are implied by 
    \begin{align*}
    {\pR_{jT} - \sum_{i\in[n]}\sum_{t\in[T]}\allocijt} \leq \rgapj^*
    \;\;
    \mbox{and}
    \;\;
    {\sum_{i\in[n]}\sum_{t\in[T]}\allocijt - \pL_{jT}} \leq \rgapj^*
    \end{align*}
    due to the consistency (Assumption~\ref{asp:prediction sequence}) of prediction sequence $\predictions$.
    Invoking the bounded over/understaffing cost properties in \Cref{lem:emulator} and second, third, and fourth constraints about $\lambdabf^*,\thetabf^*$ in program~\ref{eq:opt reduced form multi station} proves the two inequality above as desired.
\end{proof}

Combining \Cref{lem:opt reduced form lower bound optimal minimax cost L one} and \Cref{lem:opt alg minimax cost L one}, we prove \Cref{thm:opt alg multi station} for the utilitarian staffing cost as desired.

%% file: Paper/appendix-proofs/apx-optalgcanfeasibility.tex
\optalgcanfeasibility*

\begin{proof}
    We first verify the supply feasibility of the staffing profile. For each $\ell\in[\cptotal]$ and each pool $i\in[n]$,
    canonical hiring decision $\{\canallocit\}_{i\in[n],t\in\cintervalell}$ constructed in \eqref{eq:canalloc construction} satisfies 
    \begin{align*}
        \displaystyle\sum\nolimits_{t\in[t_{\ell - 1} + 1:t_{\ell}]}\frac{\wdiscount_{it_{\ell-1}}}{\wdiscountit} \canallocit
        \leq 
        \remainsupplyi
    \end{align*}
    since it is constructed from a feasible solution of subprogram $\lpcancelsubproblem$ (with constraint~\eqref{eq:all feasibility constraint}) where the remaining supply $\remainsupplyi$ is constructed in \eqref{eq:status update} as 
    \begin{align*}
        \remainsupplyi = \wdiscount_{it_{\ell - 1}}\left(\supplyi - \sum\nolimits_{t\in[t_{\ell - 1}]}\frac{1}{\wdiscountit}\allocit\right)
    \end{align*}
    Combining with the fact that $x_{it}\leq\tilde{x}_{it}$ by construction in Procedure~\ref{alg:emulator} (and its existence \Cref{lem:emulator}), we obtain the supply feasibility of the staffing profile as desired.

    Next we verify the budget feasibility of the staffing profile. For each $\ell\in[\cptotal]$, recall the algorithm identifies the largest index $k\in\cintervalellplus$ satisfying
    \begin{align*}
        \pR_{k} - \sum\nolimits_{i\in[n]}\sum\nolimits_{t\in[t_{\ell - 1}+1:k]}\allocit =
        \pR_{t_{\ell-1}} - \sum\nolimits_{i\in[n]}\sum\nolimits_{t\in[t_{\ell - 1}+1:k]}\canallocit
    \end{align*}
    The construction of Procedure~\ref{alg:emulator} ensures $\allocit = 0$ for all $i\in[n]$ and $t\in[k + 1:t_\ell]$. Moreover, due to the construction \eqref{eq:canalloc construction} of canonical hiring decision $\{\canallocit\}_{i\in[n],t\in\cintervalell}$, the construction of canonical discharging decision $\{\canrevoke_{ik}\}_{i\in[n]}$, and constraints~\labelcref{eq:identical allocation constraint,eq:identical cancellation constraint} in subprogram $\lpcancelsubproblem$, we have 
    \begin{align*}
        \sum\nolimits_{i\in[n]}\left(\sum\nolimits_{t\in[t_{\ell-1} + 1:t_\ell]}\priceit\canallocit + \cpriceiled\canrevoke_{ik}\right)
        \leq \remainbudget
    \end{align*}
    where the remaining budget $\remainbudget$ is constructed in \eqref{eq:status update} as
    \begin{align*}
        \remainbudget = \budget - \sum\nolimits_{i\in[n]}\sum\nolimits_{t\in[t_{\ell - 1}]}\priceit\allocit+\cpriceit\revokeit
    \end{align*}
    Combining with the fact that $x_{it}\leq\tilde{x}_{it}$ by construction in Procedure~\ref{alg:emulator} (and its existence \Cref{lem:emulator}) and the discharging implementation in \eqref{eq:actual cancellation construction}, we obtain the budget feasibility of the staffing profile as desired.
    
    Finally, the discharging feasibility holds due to the discharging implementation in \eqref{eq:actual cancellation construction} and the construction of canonical discharging decision $\{\canrevoke_{ik}\}_{i\in[n]}$ argued above.
\end{proof}

%% file: Paper/appendix-proofs/apx-optalgcan.tex
\optalgcan*

In this subsection, we prove \Cref{thm:opt alg cancellation} in a more general model. Specifically, we allow overstaffing cost function $\overcostj:\reals_+\rightarrow\reals_+$ and understaffing cost function $\undercostj:\reals_+\rightarrow\reals_+$ that are weakly increasing, weakly convex with $\overcostj(0) = \undercostj(0) = 0$. In this generalized model, we modify the objective of program~\ref{eq:opt cancellation} from 
$
\max_{\switchseq\in\switchseqspace}
    \undercost\cdot \rgapJ\vee
    \overcost\cdot \lgapJ$
to 
$
\max_{\switchseq\in\switchseqspace}
    \undercost(\rgapJ)\vee
    \overcost(\lgapJ)$.
Since both under/overstaffing cost functions $\undercost$, $\overcost$ are convex, the modified objective function is convex and thus the modified version of program~\ref{eq:opt cancellation} is a convex program.

We first show that program~\ref{eq:opt cancellation} is a lower bound of the optimal minimax cost $\optcost$ in the costly discharging environment.

\begin{lemma}
\label{lem:opt reduced form lower bound optimal minimax cost cancellation}
    In the costly discharging environment, for every (possibly randomized) online algorithm $\ALG$, its cost guarantee is at least the optimal objective value of program~\ref{eq:opt cancellation}.
\end{lemma}

\begin{proof}
    The argument is similar to the ones for \Cref{lem:opt reduced form lower bound optimal minimax cost,lem:opt reduced form lower bound optimal minimax cost L infty,lem:opt reduced form lower bound optimal minimax cost L one modified}. Specifically, we construct a feasible solution of program~\ref{eq:opt cancellation}, whose objective value is equal to the cost guarantee of algorithm $\ALG$.
    
    Consider prediction sequence subset $\{\predictions(\switchseq)\}_{\switchseq\in\switchseqspace}$ defined in \eqref{eq:multi switch prediction construction}. Let random variable $\randomallocit(\switchseq)$ ($\randomrevoke_{it}(\switchseq)$) be the number of workers hired (discharged) by the algorithm from pool $i$ in day $t$ under prediction sequence $\predictions(\switchseq)$.
    Due to the feasibility of the algorithm under prediction sequence $\predictions(\switchseq)$, for all sample paths (over the randomness of the algorithm), we have
    \begin{align*}
    \forall i\in[n],\forall \switchseq\in\switchseqspace:\quad &
            \displaystyle\sum\nolimits_{t\in[T]}\frac{1}{\wdiscountit}\randomallocit(\switchseq) \leq \supplyi
    \\
    \forall \switchseq\in\switchseqspace:\quad 
    & \displaystyle\sum\nolimits_{i\in[n]}
    \sum\nolimits_{t\in[T]}
    \priceit \randomallocit(\switchseq)
    +
    \cpriceit\randomrevokeit(\switchseq)
    \leq \budget
    \\
    \forall i\in[n],\forall k\in[T],\forall\switchseq\in\switchseqspace:\quad
    &
    \sum\nolimits_{t\in[k]}\randomrevokeit(\switchseq) \leq 
    \sum\nolimits_{t\in[k]}\randomallocit(\switchseq)
    \end{align*}
    Now consider the following solution $(\xbf,\ybf,\lambdabf,\thetabf)$ construction:
    \begin{align*}
        i\in[n],t\in[T],\switchseq\in\switchseqspace:&
        \qquad
        \allocitJ \gets \expect{\randomallocit(\switchseq)}
        \\
        i\in[n],\ell\in[\cptotal],\switchseq\in\switchseqspace:&
        \qquad
        \revokeiellJ \gets \sum\nolimits_{t\in\cintervalell}\expect{\randomrevokeit(\switchseq)}
        \\
        \switchseq\in\switchseqspace:
        &
        \qquad
        \lgapJ \gets
        \plus{
        \expect{\sum\nolimits_{i\in[n]}\sum\nolimits_{t\in[T]}\randomallocit(\switchseq) - \randomrevokeit(\switchseq)} -\pL_T(\switchseq)
        }
        \\
        \switchseq\in\switchseqspace:
        &
        \qquad
        \rgapJ \gets
        \plus{
        \pR_T(\switchseq) - 
        \expect{\sum\nolimits_{i\in[n]}\sum\nolimits_{t\in[T]}\randomallocit(\switchseq) - \randomrevokeit(\switchseq)} 
        }
    \end{align*}
    Note constraints~\labelcref{eq:identical allocation constraint,eq:identical cancellation constraint} are satisfied due to the construction of prediction sequence subset and the construction of the solution. Constraints~\labelcref{eq:all feasibility constraint,eq:bounded cost constraint} are satisfied due to Assumption~\ref{asp:limited cancellation fee}, the solution construction and the linearity of expectation. Therefore, the constructed solution is feasible in program~\ref{eq:opt cancellation}. Moreover, the objective value of the constructed solution at least the cost guarantee of the algorithm $\ALG$ due to the solution construction, the convexity of overstaffing/understaffing cost functions and Jensen's inequality.
\end{proof}

Next we argue that the cost guarantee of \OPTSimCan\ is upper bounded by program~\ref{eq:opt cancellation}.
\begin{lemma}
\label{lem:opt alg minimax cost cancellation}
    In the multi-station environment with Rawlsian-staffing cost, the cost guarantee of \OPTSimCan\ is at most the optimal objective value of program~\ref{eq:opt cancellation}.
\end{lemma}

Before proving \Cref{lem:opt alg minimax cost cancellation}, we introduce the following auxiliary concept. With slight abuse of notation, we use $\optcost(t, \status)$ to denote the optimal minimax cost given status $\status$ at the end of day $t - 1$. Here status $\status = (t, \bar\cumallocs,\remainsupplies,\remainbudget,[\curpL,\curpR])$ follows the same definition as in the main text except that we expand its first coordinates from $[\cptotal]$ (aka., $\{t_1,\dots,t_\cptotal\}$) to the whole time horizon $[T]$. We establish two properties about $\optcost(t,\status)$ between different status $\status$.

\begin{lemma}
\label{lem:vtg comparison}
    In the costly discharging environment, the following properties holds:
    \begin{enumerate}
        \item[(i)] for any $t\in[T]$ and two statuses $\status = (t, \bar\cumallocs,\remainsupplies, \remainbudget, [\curpL,\curpR])$ and $\status\primed = (t, \bar\cumallocs\primed,\remainsupplies\primed, \remainbudget\primed, [\curpL\primed,\curpR\primed])$, if $\remainsupplyi\geq \remainsupplyi\primed$ for all $i\in[n]$, $\remainbudget \geq \remainbudget\primed$, $\sum_{i\in[n]}\cumallociBar - \curpL \leq \sum_{i\in[n]}\cumallociBar\primed - \curpL\primed$, and $\curpR - \sum_{i\in[n]}\cumallociBar \leq \curpR\primed -\sum_{i\in[n]}\cumallociBar\primed$, then the optimal minimax cost given status $\status$ is weakly smaller than the optimal minimax cost given status $\status\primed$, i.e., $\optcost(t,\status) \leq \optcost(t,\status\primed)$;
        
        \item[(ii)] for any $t\in[T]$ and three statuses $\status = (t, \bar\cumallocs,\remainsupplies, \remainbudget, [\curpL,\curpR])$, $\status\primed = (t, \bar\cumallocs\primed,\remainsupplies\primed, \remainbudget\primed, [\curpL\primed,\curpR\primed])$, and $\status\doubleprimed = (t, \bar\cumallocs\doubleprimed,\remainsupplies\doubleprimed, \remainbudget\doubleprimed, [\curpL\doubleprimed,\curpR\doubleprimed])$, if $\remainsupplyi\geq \remainsupplyi\primed\vee\remainsupplyi\doubleprimed$ for all $i\in[n]$, $\remainbudget \geq \remainbudget\primed\vee\remainbudget\doubleprimed$, $\curpR-\curpL = \curpR\primed - \curpL\primed = \curpR\doubleprimed - \curpL\doubleprimed$, and 
        \begin{align*}
            \curpR\primed - \sum\nolimits_{i\in[n]}\cumallociBar\primed \leq 
            \curpR - \sum\nolimits_{i\in[n]}\cumallociBar \leq 
            \curpR\doubleprimed - \sum\nolimits_{i\in[n]}\cumallociBar\doubleprimed 
        \end{align*}
        then
        \begin{align*}
            \optcost(t, \status) \leq \optcost(t, \status\primed)\vee \optcost(t, \status\doubleprimed)
        \end{align*}
    \end{enumerate}
\end{lemma}
\begin{proof}
    When the condition in property (i) is satisfied, due to Assumption~\ref{asp:limited cancellation fee} that the per-worker discharging fees are identical across pools, the platform with status $\status$ can mimic any online algorithm for status $\status\primed$ and obtain a weakly smaller staffing cost, which implies $\optcost(t, \status) \leq \optcost(t,\status\primed)$ as desired.

    We prove property (ii) with an induction argument over $t\in[T + 1]$.

    \xhdr{Base case ($t = T + 1$):} In this case, 
    \begin{align*}
        \optcost(T + 1,\status) &= \undercost\left(\plus{\curpR - \sum\nolimits_{i\in[n]}\cumallociBar}\right)\vee\overcost\left(\plus{\sum\nolimits_{i\in[n]}\cumallociBar - \curpL}\right) 
        \\
        &\leq  
        \undercost\left(\plus{\curpR\doubleprimed - \sum\nolimits_{i\in[n]}\cumallociBar\doubleprimed}\right)
        \vee \overcost\left(\plus{\sum\nolimits_{i\in[n]}\cumallociBar\primed - \curpL\primed}\right) 
        \\
        &\leq 
        \optcost(T + 1,\status\doubleprimed)\vee \optcost(T + 1,\status\primed)
    \end{align*}
    where the first inequality holds due to the monotonicity of the under/overstaffing cost functions and the property condition, i.e., $\curpR-\curpL = \curpR\primed - \curpL\primed = \curpR\doubleprimed - \curpL\doubleprimed$ and $\curpR\primed - \sum\nolimits_{i\in[n]}\cumallociBar\primed \leq 
    \curpR - \sum\nolimits_{i\in[n]}\cumallociBar \leq 
    \curpR\doubleprimed - \sum\nolimits_{i\in[n]}\cumallociBar\doubleprimed$.

    \xhdr{Inductive step for $t \in [T]$:} Suppose property~(ii) holds for all $\tau \in [t + 1: T + 1]$. Due to the property condition that $\curpR-\curpL = \curpR\primed - \curpL\primed = \curpR\doubleprimed - \curpL\doubleprimed$, we assume $\curpR = \curpR\primed = \curpR\doubleprimed$ and $\curpL = \curpL\primed = \curpL\doubleprimed$ without loss of generality. Then condition
    \begin{align*}
    \curpR\primed - \sum\nolimits_{i\in[n]}\cumallociBar\primed \leq 
    \curpR - \sum\nolimits_{i\in[n]}\cumallociBar \leq 
    \curpR\doubleprimed - \sum\nolimits_{i\in[n]}\cumallociBar\doubleprimed    
    \end{align*}
    becomes
    \begin{align*}
    \sum\nolimits_{i\in[n]}\cumallociBar\doubleprimed \leq 
    \sum\nolimits_{i\in[n]}\cumallociBar \leq 
    \sum\nolimits_{i\in[n]}\cumallociBar\primed
    \end{align*}
    Now consider an arbitrary prediction $[\pLt,\pRt]\subseteq[\curpL,\curpR]$ revealed in day~$t$. Let $\{\cumallociHat\primed\}_{i\in[n]}$ and $\{\cumallociHat\doubleprimed\}_{i\in[n]}$ be the total number of hired workers at the end of day $t$ in the minimax optimal algorithm with status $\status\primed$, $\status\doubleprimed$, respectively. We analyze three cases and claim that in all cases, the platform with status $\status$ can achieve a smaller cost guarantee than the cost guarantee of the minimax optimal algorithm with status $\status\primed$ or status $\status\doubleprimed$.
    \begin{itemize}
        \item \textbf{Case-1:} Suppose $\sum\nolimits_{i\in[n]}\cumallociHat\doubleprimed\leq \sum\nolimits_{i\in[n]}\cumallociBar \leq 
        \sum\nolimits_{i\in[n]}\cumallociHat\primed$. In this case, the platform with status $\status$ can make no hiring decision nor discharging decision in day $t$, which enables us to invoke the induction hypotheses for day $t + 1$ and thus ensure our claim as desired.
        \item \textbf{Case-2:} Suppose $\sum\nolimits_{i\in[n]}\cumallociHat\doubleprimed\geq \sum\nolimits_{i\in[n]}\cumallociBar$. In this case, the platform with status $\status$ can emulate the hiring decisions of the minimax optimal algorithm with status $\status\doubleprimed$ such that $\sum\nolimits_{i\in[n]}\cumallociHat\doubleprimed=\sum\nolimits_{i\in[n]}\cumallociHat$ where $\cumallociHat$ is the total number of hired workers at the end of day $t$ after this emulation from status $\status$. Since $\sum\nolimits_{i\in[n]}\cumallociBar\doubleprimed \leq    \sum\nolimits_{i\in[n]}\cumallociBar$, the platform can achieve this emulation with more remaining supplies and remaining budgets. Thus, the claim holds by invoking property~(i).
        \item \textbf{Case-3:} Suppose $\sum\nolimits_{i\in[n]}\cumallociHat\primed\leq \sum\nolimits_{i\in[n]}\cumallociBar$. In this case, the platform with status $\status$ can emulate the hiring decisions of the minimax optimal algorithm with status $\status\primed$ such that $\sum\nolimits_{i\in[n]}\cumallociHat\primed=\sum\nolimits_{i\in[n]}\cumallociHat$ where $\cumallociHat$ is the total number of hired workers at the end of day $t$ after this emulation from status $\status$. Since $\sum\nolimits_{i\in[n]}\cumallociBar\primed \geq    \sum\nolimits_{i\in[n]}\cumallociBar$, the platform can achieve this emulation with more remaining supplies and remaining budgets. Thus, the claim holds by invoking property~(i).
    \end{itemize}
    Since for all predictions revealed in day $t$, the platform with status $\status$ can achieve a smaller cost guarantee than either the cost guarantee of the minimax optimal algorithm with status $\status\primed$ or $\status\doubleprimed$, the inductive step is completed and the property~(ii) is shown by induction as desired.
\end{proof}

Now we are ready to prove \Cref{lem:opt alg minimax cost cancellation}.
\begin{proof}[\textsl{Proof of \Cref{lem:opt alg minimax cost cancellation}}]
    In this analysis, we claim that given any status $\status \triangleq (\ell, \bar\cumallocs,\remainsupplies,\remainbudget,[\curpL,\curpR])$ at the beginning of phase $\ell$ (aka., at the end of day $t_{\ell - 1}$), the cost guarantee of \OPTSimCan\ is at most the optimal objective value of subprogram~$\lpcancelsubproblem$. We prove this claim using an induction argument over $\ell\in[\cptotal + 1]$.

    \xhdr{Base Case ($\ell = \cptotal + 1$):} In this case, the optimal objective of subprogram~$\lpcancelsubproblem$ becomes 
    \begin{align*}
        \undercost(\rgap^*) \vee \overcost(\lgap^*) &= 
        \undercost\left(\plus{\curpR - \sum\nolimits_{i\in[n]}\cumallociBar}\right)
        \vee
        \overcost\left(\plus{\sum\nolimits_{i\in[n]}\cumallociBar-\curpL}\right)
        \\
        &=
        \undercost\left(\plus{\pR_T - \sum\nolimits_{i\in[n]}\sum\nolimits_{t\in[T]}(\allocit - \revokeit)}\right)
        \vee
        \overcost\left(\plus{\sum\nolimits_{i\in[n]}\sum\nolimits_{t\in[T]}(\allocit - \revokeit)-\pL_T}\right)
    \end{align*}
    which is equal to the worst-case staffing cost of the platform. Here the second equality holds due to the construction of status $\status$ in \eqref{eq:status update}.

    \xhdr{Inductive step ($\ell\in[\cptotal]$):}  Suppose the claim holds for all $\ell' \in [\ell + 1: \cptotal + 1]$. 
    Let $\{\allocit,\revokeit\}_{i\in[n],t\in[t_{\ell}]}$ be the staffing decision made by \OPTSimCan, and $\status'$ be the status at the end of day $t_\ell$. Now we consider two hypothetical hiring and discharging decisions $\{\allocit\primed,\revokei\primed\}_{i\in[n],t\in\cintervalell}$ and $\{\allocit\doubleprimed,\revokei\doubleprimed\}_{i\in[n],t\in\cintervalell}$ constructed from the optimal solution $(\xbf^*,\ybf^*,\lambdabf^*,\thetabf^*)$ of subprogram~$\lpcancelsubproblem$.  Define
    \begin{align*}
        \forall i\in[n],\forall t\in\cintervalell:&\qquad
        \allocit\primed \triangleq \allocit^*(\switchseq^{(k)}),
        \quad
        \allocit\doubleprimed \triangleq \allocit^*(\switchseq^{(t_\ell)}),
        \\
        \forall i\in[n]:&\qquad
        \revokei\primed \triangleq \revokeiell^*(\switchseq^{(k)}),
        \quad
        \revokei\doubleprimed \triangleq \revokeiell^*(\switchseq^{(t_\ell)}),
    \end{align*}
    where  $k$ is the largest index in $\cintervalellplus$ satisfying condition~\eqref{eq:critical day cancellation}, and $\switchseq^{(t)}$ is an arbitrary sequence such that $\switchseq^{(t)}_1 = t$.\footnote{Due to constraints~\labelcref{eq:identical allocation constraint,eq:identical cancellation constraint}, all configurations $\switchseq$ such that $\switchseq_1 = t$ have the same $\allocitJ,\revokeiellJ$.} 
    Let $\status\primed$ (resp.\ $\status\doubleprimed$) be the status at the end of day $t_{\ell}$ if the platform implements staffing profile $\{\allocit\primed,\revokei\primed\}_{i\in[n],t\in\cintervalell}$ (resp. $\{\allocit\doubleprimed,\revokei\doubleprimed\}_{i\in[n],t\in\cintervalell}$) under prediction sequence $\predictions(\switchseq^{(k)})$ (resp.\ $\predictions(\switchseq^{(t_\ell)})$). Next we compare three statuses $\status',\status\primed,\status\doubleprimed$ and apply properties~(i) and (ii) in \Cref{lem:vtg comparison}.
    
    Note that $\{\allocit\doubleprimed,\revokei\doubleprimed\}_{i\in[n],t\in\cintervalell}$ is used as the canonical hiring decision for Procedure~\ref{alg:emulator} in \OPTSimCan. 
    Due to the feasibility of Procedure~\ref{alg:emulator} in \Cref{lem:emulator}, we have    
        $\allocit\leq \allocit\doubleprimed$
    for all $i\in[n],t\in\cintervalell$. Moreover, due to constraint~\ref{eq:identical allocation constraint} in subprogram~$\lpcancelsubproblem$, we have $\allocit\primed = \allocit\doubleprimed$ and thus $\allocit\leq \allocit\primed$ for all $i\in[n],t\in[t_{\ell - 1} + 1:k]$. Due to index $k$'s definition (i.e., largest index such that condition~\eqref{eq:critical day cancellation} holds), the construction of Procedure~\ref{alg:emulator} ensures $\allocit = 0$ for all $i\in[n],t\in[k + 1: t_\ell]$. To sum up, for all $i\in[n],t\in\cintervalell$, we have
    \begin{align}
    \label{eq:alloc is at most alloc primed doubleprimed}
        \allocit\leq \allocit\primed\;\;\mbox{and}\;\;\allocit\leq \allocit\doubleprimed
    \end{align}
    Note the construction of Procedure~\ref{alg:emulator} and the definition of index $k$ guarantees
    \begin{align}
    \label{eq:rgap is at most rgap doubleprimed}
    \begin{split}
        \pR_{t_\ell}(\switchseq^{(t_\ell)}) - \sum\nolimits_{i\in[n]}\sum\nolimits_{t\in\cintervalell}\allocit\doubleprimed
        &\geq
        \pR_{t_\ell} - 
        \sum\nolimits_{i\in[n]}\sum\nolimits_{t\in\cintervalell}\allocit
        \\
        &\geq
        \pR_{t_\ell}(\switchseq\ked) - 
        \sum\nolimits_{i\in[n]}\sum\nolimits_{t\in[t_{\ell - 1} + 1:k]}\allocit\primed
    \end{split}
    \end{align}
    where the second inequality further implies 
    \begin{align}
    \label{eq:rgap is at least rgap primed}
        \pR_{t_\ell} - 
        \sum\nolimits_{i\in[n]}\sum\nolimits_{t\in\cintervalell}\allocit
        \geq
        \pR_{t_\ell}(\switchseq\ked) - 
        \sum\nolimits_{i\in[n]}\sum\nolimits_{t\in\cintervalell}\allocit\primed
    \end{align}
    due to the non-negativity of $\allocit\primed$.
        Combining inequalities~\labelcref{eq:alloc is at most alloc primed doubleprimed,eq:rgap is at most rgap doubleprimed,eq:rgap is at least rgap primed} and condition~\eqref{eq:actual cancellation construction} in \OPTSimCan, we know that 
    \begin{itemize}
        \item if \OPTSimCan\ makes no discharging decision in day $t_\ell$ (i.e., condition~\eqref{eq:actual cancellation construction} is not satisfied), property~(ii)'s condition in \Cref{lem:vtg comparison} is satisfied for $\status',\status\primed,\status\doubleprimed$, and thus the cost guarantee of the algorithm is at most $\optcost(\ell + 1, \status\primed)\vee\optcost(\ell+ 1,\status\doubleprimed)$;
        \item if \OPTSimCan\ makes discharging decision in day $t_\ell$ (i.e., condition~\eqref{eq:actual cancellation construction} is satisfied), due to Assumption~\ref{asp:limited cancellation fee}, property~(i)'s condition in \Cref{lem:vtg comparison} is satisfied for $\status',\status\primed$, and thus the cost guarantee of the algorithm is at most $\optcost(\ell + 1, \status\primed)$.
    \end{itemize}
    Invoking the induction hypothesis for $\ell' = \ell + 1$, we know the cost guarantee of the algorithm (under the revealed prediction sequence $\predictions$) is at most the optimal objective value of subprogram~$\text{\ref{eq:opt cancellation}}[\ell+1,\status\primed]$ or the optimal objective value of subprogram~$\text{\ref{eq:opt cancellation}}[\ell+1,\status\doubleprimed]$. Therefore, it suffices to show the optimal objective value of subprogram~$\lpcancelsubproblem$ is weakly higher than subprogram~$\text{\ref{eq:opt cancellation}}[\ell+1,\status\primed]$ and subprogram~$\text{\ref{eq:opt cancellation}}[\ell+1,\status\doubleprimed]$. To see this, note that any optimal solution in subprogram~$\text{\ref{eq:opt cancellation}}[\ell+1,\status\primed]$ can be converted straightforwardly into a feasible solution of subprogram~$\text{\ref{eq:opt cancellation}}[\ell+1,\status\primed]$ (subprogram~$\text{\ref{eq:opt cancellation}}[\ell+1,\status\doubleprimed]$) with weakly smaller objective value.
\end{proof}

Finally, combining \Cref{lem:opt reduced form lower bound optimal minimax cost cancellation} and \Cref{lem:opt alg minimax cost cancellation}, we prove \Cref{thm:opt alg cancellation} as desired.

%% file: main.bbl
\begin{thebibliography}{85}
\providecommand{\natexlab}[1]{#1}
\providecommand{\url}[1]{\texttt{#1}}
\expandafter\ifx\csname urlstyle\endcsname\relax
  \providecommand{\doi}[1]{doi: #1}\else
  \providecommand{\doi}{doi: \begingroup \urlstyle{rm}\Url}\fi

\bibitem[Aigner et~al.(2023)Aigner, B{\"a}rmann, Braun, Liers, Pokutta,
  Schneider, Sharma, and Tschuppik]{aigner2023data}
Kevin-Martin Aigner, Andreas B{\"a}rmann, Kristin Braun, Frauke Liers,
  Sebastian Pokutta, Oskar Schneider, Kartikey Sharma, and Sebastian Tschuppik.
\newblock Data-driven distributionally robust optimization over time.
\newblock \emph{INFORMS Journal on Optimization}, 5\penalty0 (4):\penalty0
  376--394, 2023.

\bibitem[Altman et~al.(2013)Altman, Machin, Bryant, and Gardner]{GA-90}
Douglas Altman, David Machin, Trevor Bryant, and Martin Gardner.
\newblock \emph{Statistics with confidence: confidence intervals and
  statistical guidelines}.
\newblock John Wiley \& Sons, 2013.

\bibitem[{Amazon Flex}(2025)]{amazonflex2025}
{Amazon Flex}.
\newblock Deliver with amazon flex.
\newblock \url{https://flex.amazon.com/}, February 2025.
\newblock Accessed on: 02/08/2025.

\bibitem[{Amazon Forecast}(2025)]{amazonforecast2025}
{Amazon Forecast}.
\newblock Amazon forecast algorithms.
\newblock
  \url{https://docs.aws.amazon.com/forecast/latest/dg/aws-forecast-choosing-recipes.html/},
  February 2025.
\newblock Accessed on: 02/08/2025.

\bibitem[Anari et~al.(2019)Anari, Niazadeh, Saberi, and Shameli]{ANSS-19}
Nima Anari, Rad Niazadeh, Amin Saberi, and Ali Shameli.
\newblock Nearly optimal pricing algorithms for production constrained and
  laminar bayesian selection.
\newblock In Anna~R. Karlin, Nicole Immorlica, and Ramesh Johari, editors,
  \emph{Proceedings of the 2019 {ACM} Conference on Economics and Computation,
  {EC} 2019, Phoenix, AZ, USA, June 24-28, 2019}, pages 91--92. {ACM}, 2019.

\bibitem[Angelopoulos and Bates(2023)]{AB-23}
Anastasios~N. Angelopoulos and Stephen Bates.
\newblock Conformal prediction: {A} gentle introduction.
\newblock \emph{Found. Trends Mach. Learn.}, 16\penalty0 (4):\penalty0
  494--591, 2023.

\bibitem[Anunrojwong et~al.(2023)Anunrojwong, Balseiro, and Besbes]{ABB-23}
Jerry Anunrojwong, Santiago~R. Balseiro, and Omar Besbes.
\newblock Robust auction design with support information.
\newblock In \emph{Proceedings of the 24th ACM Conference on Economics and
  Computation}, EC '23, page 113, New York, NY, USA, 2023. Association for
  Computing Machinery.
\newblock ISBN 9798400701047.

\bibitem[Arrow et~al.(1951)Arrow, Harris, and Marschak]{AHM-51}
Kenneth~J. Arrow, Theodore Harris, and Jacob Marschak.
\newblock Optimal inventory policy.
\newblock \emph{Econometrica}, 19\penalty0 (3):\penalty0 250--272, 1951.
\newblock ISSN 00129682, 14680262.

\bibitem[Azar et~al.(2022)Azar, Panigrahi, and Touitou]{APT-22}
Yossi Azar, Debmalya Panigrahi, and Noam Touitou.
\newblock Online graph algorithms with predictions.
\newblock In Joseph~(Seffi) Naor and Niv Buchbinder, editors, \emph{Proceedings
  of the 2022 {ACM-SIAM} Symposium on Discrete Algorithms, {SODA} 2022, Virtual
  Conference / Alexandria, VA, USA, January 9 - 12, 2022}, pages 35--66.
  {SIAM}, 2022.

\bibitem[Babaioff et~al.(2008)Babaioff, Hartline, and
  Kleinberg]{babaioff2008selling}
Moshe Babaioff, Jason Hartline, and Robert Kleinberg.
\newblock Selling banner ads: Online algorithms with buyback.
\newblock In \emph{Fourth workshop on ad auctions}, 2008.

\bibitem[Babichenko et~al.(2022)Babichenko, Talgam{-}Cohen, Xu, and
  Zabarnyi]{BTXZ-22}
Yakov Babichenko, Inbal Talgam{-}Cohen, Haifeng Xu, and Konstantin Zabarnyi.
\newblock Regret-minimizing bayesian persuasion.
\newblock \emph{Games Econ. Behav.}, 136:\penalty0 226--248, 2022.

\bibitem[Bachrach et~al.(2022)Bachrach, Chen, Talgam-Cohen, Yang, and
  Zhang]{BCTYZ-22}
Nir Bachrach, Yi-Chun Chen, Inbal Talgam-Cohen, Xiangqian Yang, and Wanchang
  Zhang.
\newblock Distributionally robust auction design.
\newblock Technical report, Working paper, 2022.

\bibitem[Balkanski et~al.(2023)Balkanski, P{\'{e}}rivier, Stein, and
  Wei]{BPSW-23}
Eric Balkanski, No{\'{e}}mie P{\'{e}}rivier, Clifford Stein, and Hao{-}Ting
  Wei.
\newblock Energy-efficient scheduling with predictions.
\newblock In Alice Oh, Tristan Naumann, Amir Globerson, Kate Saenko, Moritz
  Hardt, and Sergey Levine, editors, \emph{Advances in Neural Information
  Processing Systems 36: Annual Conference on Neural Information Processing
  Systems 2023, NeurIPS 2023, New Orleans, LA, USA, December 10 - 16, 2023},
  2023.

\bibitem[Ball and Queyranne(2009)]{BQ-09}
Michael~O. Ball and Maurice Queyranne.
\newblock Toward robust revenue management: Competitive analysis of online
  booking.
\newblock \emph{Oper. Res.}, 57\penalty0 (4):\penalty0 950--963, 2009.

\bibitem[Balseiro et~al.(2023)Balseiro, Kroer, and Kumar]{BKK-23}
Santiago~R. Balseiro, Christian Kroer, and Rachitesh Kumar.
\newblock Single-leg revenue management with advice.
\newblock In Kevin Leyton{-}Brown, Jason~D. Hartline, and Larry Samuelson,
  editors, \emph{Proceedings of the 24th {ACM} Conference on Economics and
  Computation, {EC} 2023, London, United Kingdom, July 9-12, 2023}, page 207.
  {ACM}, 2023.

\bibitem[Bamas et~al.(2020)Bamas, Maggiori, and Svensson]{BMS-20}
{\'{E}}tienne Bamas, Andreas Maggiori, and Ola Svensson.
\newblock The primal-dual method for learning augmented algorithms.
\newblock In Hugo Larochelle, Marc'Aurelio Ranzato, Raia Hadsell,
  Maria{-}Florina Balcan, and Hsuan{-}Tien Lin, editors, \emph{Advances in
  Neural Information Processing Systems 33: Annual Conference on Neural
  Information Processing Systems 2020, NeurIPS 2020, December 6-12, 2020,
  virtual}, 2020.

\bibitem[Bandi and Bertsimas(2014)]{BB-14}
Chaithanya Bandi and Dimitris Bertsimas.
\newblock Optimal design for multi-item auctions: {A} robust optimization
  approach.
\newblock \emph{Math. Oper. Res.}, 39\penalty0 (4):\penalty0 1012--1038, 2014.

\bibitem[Banerjee et~al.(2022)Banerjee, Gkatzelis, Gorokh, and Jin]{BGGJ-22}
Siddhartha Banerjee, Vasilis Gkatzelis, Artur Gorokh, and Billy Jin.
\newblock Online nash social welfare maximization with predictions.
\newblock In Joseph~(Seffi) Naor and Niv Buchbinder, editors, \emph{Proceedings
  of the 2022 {ACM-SIAM} Symposium on Discrete Algorithms, {SODA} 2022, Virtual
  Conference / Alexandria, VA, USA, January 9 - 12, 2022}, pages 1--19. {SIAM},
  2022.

\bibitem[Bei et~al.(2019)Bei, Gravin, Lu, and Tang]{BGLT-19}
Xiaohui Bei, Nick Gravin, Pinyan Lu, and Zhihao~Gavin Tang.
\newblock Correlation-robust analysis of single item auction.
\newblock In Timothy~M. Chan, editor, \emph{Proceedings of the Thirtieth Annual
  {ACM-SIAM} Symposium on Discrete Algorithms, {SODA} 2019, San Diego,
  California, USA, January 6-9, 2019}, pages 193--208. {SIAM}, 2019.

\bibitem[Ben-Tal et~al.(2009)Ben-Tal, El~Ghaoui, and Nemirovski]{BEN-09}
Aharon Ben-Tal, Laurent El~Ghaoui, and Arkadi Nemirovski.
\newblock \emph{Robust optimization}, volume~28.
\newblock Princeton university press, 2009.

\bibitem[Bergemann and Schlag(2008)]{BS-08}
Dirk Bergemann and Karl~H. Schlag.
\newblock Pricing without priors.
\newblock \emph{Journal of the European Economic Association}, 6\penalty0
  (2/3):\penalty0 560--569, 2008.
\newblock ISSN 15424766, 15424774.

\bibitem[Bergemann et~al.(2022)Bergemann, Castro, and Weintraub]{BCW-22}
Dirk Bergemann, Francisco Castro, and Gabriel~Y. Weintraub.
\newblock Third-degree price discrimination versus uniform pricing.
\newblock \emph{Games Econ. Behav.}, 131:\penalty0 275--291, 2022.

\bibitem[Bertsimas and Sim(2003)]{bertsimas2003robust}
Dimitris Bertsimas and Melvyn Sim.
\newblock Robust discrete optimization and network flows.
\newblock \emph{Mathematical programming}, 98\penalty0 (1):\penalty0 49--71,
  2003.

\bibitem[Bertsimas and Sim(2004)]{bertsimas2004price}
Dimitris Bertsimas and Melvyn Sim.
\newblock The price of robustness.
\newblock \emph{Operations research}, 52\penalty0 (1):\penalty0 35--53, 2004.

\bibitem[Bertsimas and Thiele(2006)]{BT-06}
Dimitris Bertsimas and Aur{\'{e}}lie Thiele.
\newblock A robust optimization approach to inventory theory.
\newblock \emph{Oper. Res.}, 54\penalty0 (1):\penalty0 150--168, 2006.

\bibitem[Bertsimas et~al.(2011)Bertsimas, Brown, and
  Caramanis]{bertsimas2011theory}
Dimitris Bertsimas, David~B Brown, and Constantine Caramanis.
\newblock Theory and applications of robust optimization.
\newblock \emph{SIAM review}, 53\penalty0 (3):\penalty0 464--501, 2011.

\bibitem[Blackwell(1953)]{bla-53}
David Blackwell.
\newblock Equivalent comparisons of experiments.
\newblock \emph{The Annals of Mathematical Statistics}, 24\penalty0
  (2):\penalty0 265--272, 1953.
\newblock ISSN 00034851.

\bibitem[Borodin and El-Yaniv(2005)]{BE-05}
Allan Borodin and Ran El-Yaniv.
\newblock \emph{Online computation and competitive analysis}.
\newblock cambridge university press, 2005.

\bibitem[Boyar et~al.(2017)Boyar, Favrholdt, Kudahl, Larsen, and
  Mikkelsen]{BFKLM-17}
Joan Boyar, Lene~M. Favrholdt, Christian Kudahl, Kim~S. Larsen, and Jesper~W.
  Mikkelsen.
\newblock Online algorithms with advice: {A} survey.
\newblock \emph{{ACM} Comput. Surv.}, 50\penalty0 (2):\penalty0 19:1--19:34,
  2017.

\bibitem[Brooks and Du(2021)]{BD-21}
Benjamin Brooks and Songzi Du.
\newblock Optimal auction design with common values: An informationally robust
  approach.
\newblock \emph{Econometrica}, 89\penalty0 (3):\penalty0 1313--1360, 2021.

\bibitem[Caldentey et~al.(2017)Caldentey, Liu, and Lobel]{CLL-17}
Ren{\'{e}} Caldentey, Ying Liu, and Ilan Lobel.
\newblock Intertemporal pricing under minimax regret.
\newblock \emph{Oper. Res.}, 65\penalty0 (1):\penalty0 104--129, 2017.

\bibitem[Carroll(2015)]{Car-15}
Gabriel Carroll.
\newblock Robustness and linear contracts.
\newblock \emph{American Economic Review}, 105\penalty0 (2):\penalty0 536–63,
  February 2015.

\bibitem[Chassein and Goerigk(2018)]{chassein2018variable}
Andr{\'e} Chassein and Marc Goerigk.
\newblock Variable-sized uncertainty and inverse problems in robust
  optimization.
\newblock \emph{European Journal of Operational Research}, 264\penalty0
  (1):\penalty0 17--28, 2018.

\bibitem[Choe and Ramdas(2024)]{CR-24}
Yo~Joong Choe and Aaditya Ramdas.
\newblock Comparing sequential forecasters.
\newblock \emph{Oper. Res.}, 72\penalty0 (4):\penalty0 1368--1387, 2024.

\bibitem[Choo et~al.(2025)Choo, Jin, and Shin]{choo2025learning}
Davin Choo, Billy Jin, and Yongho Shin.
\newblock Learning-augmented online bipartite fractional matching.
\newblock \emph{arXiv preprint arXiv:2505.19252}, 2025.

\bibitem[Christianson et~al.(2022)Christianson, Handina, and
  Wierman]{christianson2022chasing}
Nicolas Christianson, Tinashe Handina, and Adam Wierman.
\newblock Chasing convex bodies and functions with black-box advice.
\newblock In \emph{Conference on Learning Theory}, pages 867--908. PMLR, 2022.

\bibitem[Christianson et~al.(2024)Christianson, Sun, Low, and
  Wierman]{christianson2024risk}
Nicolas Christianson, Bo~Sun, Steven Low, and Adam Wierman.
\newblock Risk-sensitive online algorithms.
\newblock \emph{ACM SIGMETRICS Performance Evaluation Review}, 52\penalty0
  (2):\penalty0 6--8, 2024.

\bibitem[Clark and Scarf(1960)]{CS-60}
Andrew~J. Clark and Herbert Scarf.
\newblock Optimal policies for a multi-echelon inventory problem.
\newblock \emph{Management Science}, 6\penalty0 (4):\penalty0 475--490, 1960.

\bibitem[Devanur et~al.(2011)Devanur, Hartline, Karlin, and Nguyen]{DHKN-11}
Nikhil~R. Devanur, Jason~D. Hartline, Anna~R. Karlin, and C.~Thach Nguyen.
\newblock Prior-independent multi-parameter mechanism design.
\newblock In Ning Chen, Edith Elkind, and Elias Koutsoupias, editors,
  \emph{Internet and Network Economics - 7th International Workshop, {WINE}
  2011, Singapore, December 11-14, 2011. Proceedings}, volume 7090 of
  \emph{Lecture Notes in Computer Science}, pages 122--133. Springer, 2011.

\bibitem[Drygala et~al.(2023)Drygala, Nagarajan, and Svensson]{DNS-23}
Marina Drygala, Sai~Ganesh Nagarajan, and Ola Svensson.
\newblock Online algorithms with costly predictions.
\newblock In Francisco J.~R. Ruiz, Jennifer~G. Dy, and Jan{-}Willem van~de
  Meent, editors, \emph{International Conference on Artificial Intelligence and
  Statistics, 25-27 April 2023, Palau de Congressos, Valencia, Spain}, volume
  206 of \emph{Proceedings of Machine Learning Research}, pages 8078--8101.
  {PMLR}, 2023.

\bibitem[D{\"{u}}tting et~al.(2019)D{\"{u}}tting, Roughgarden, and
  Talgam{-}Cohen]{DRT-19}
Paul D{\"{u}}tting, Tim Roughgarden, and Inbal Talgam{-}Cohen.
\newblock Simple versus optimal contracts.
\newblock In Anna~R. Karlin, Nicole Immorlica, and Ramesh Johari, editors,
  \emph{Proceedings of the 2019 {ACM} Conference on Economics and Computation,
  {EC} 2019, Phoenix, AZ, USA, June 24-28, 2019}, pages 369--387. {ACM}, 2019.

\bibitem[Edgeworth(1888)]{Edg-88}
Francis~Y Edgeworth.
\newblock The mathematical theory of banking.
\newblock \emph{Journal of the Royal Statistical Society}, 51\penalty0
  (1):\penalty0 113--127, 1888.

\bibitem[Ekbatani et~al.(2022)Ekbatani, Feng, and Niazadeh]{ekbatani2022online}
Farbod Ekbatani, Yiding Feng, and Rad Niazadeh.
\newblock Online matching with cancellation costs.
\newblock \emph{arXiv preprint arXiv:2210.11570}, 2022.

\bibitem[Ekbatani et~al.(2024)Ekbatani, Niazadeh, Nuti, and
  Vondrak]{ekbatani2024prophet}
Farbod Ekbatani, Rad Niazadeh, Pranav Nuti, and Jan Vondrak.
\newblock Prophet inequalities with cancellation costs.
\newblock \emph{arXiv preprint arXiv:2404.00527}, 2024.

\bibitem[Fatehi and Wagner(2022)]{FW-22}
Soraya Fatehi and Michael~R. Wagner.
\newblock Crowdsourcing last-mile deliveries.
\newblock \emph{Manuf. Serv. Oper. Manag.}, 24\penalty0 (2):\penalty0 791--809,
  2022.

\bibitem[Feng et~al.(2024{\natexlab{a}})Feng, Niazadeh, and Saberi]{FNS-24}
Yiding Feng, Rad Niazadeh, and Amin Saberi.
\newblock Near-optimal bayesian online assortment of reusable resources.
\newblock \emph{Oper. Res.}, 72\penalty0 (5):\penalty0 1861--1873,
  2024{\natexlab{a}}.

\bibitem[Feng et~al.(2024{\natexlab{b}})Feng, Niazadeh, and
  Saberi]{feng2024two}
Yiding Feng, Rad Niazadeh, and Amin Saberi.
\newblock Two-stage stochastic matching and pricing with applications to ride
  hailing.
\newblock \emph{Operations Research}, 72\penalty0 (4):\penalty0 1574--1594,
  2024{\natexlab{b}}.

\bibitem[Fisher and Raman(1996)]{FR-96}
Marshall~L. Fisher and Ananth Raman.
\newblock Reducing the cost of demand uncertainty through accurate response to
  early sales.
\newblock \emph{Oper. Res.}, 44\penalty0 (1):\penalty0 87--99, 1996.

\bibitem[Frankel(2014)]{Fra-14}
Alexander Frankel.
\newblock Aligned delegation.
\newblock \emph{American Economic Review}, 104\penalty0 (1):\penalty0 66–83,
  January 2014.

\bibitem[Gibbs and Candes(2021)]{gibbs2021adaptive}
Isaac Gibbs and Emmanuel Candes.
\newblock Adaptive conformal inference under distribution shift.
\newblock \emph{Advances in Neural Information Processing Systems},
  34:\penalty0 1660--1672, 2021.

\bibitem[Goldberg et~al.(2001)Goldberg, Hartline, and Wright]{GHW-01}
Andrew~V. Goldberg, Jason~D. Hartline, and Andrew Wright.
\newblock Competitive auctions and digital goods.
\newblock In S.~Rao Kosaraju, editor, \emph{Proceedings of the Twelfth Annual
  Symposium on Discrete Algorithms, January 7-9, 2001, Washington, DC, {USA}},
  pages 735--744. {ACM/SIAM}, 2001.

\bibitem[Golrezaei et~al.(2023)Golrezaei, Jaillet, and Zhou]{GJZ-23}
Negin Golrezaei, Patrick Jaillet, and Zijie Zhou.
\newblock Online resource allocation with convex-set machine-learned advice.
\newblock \emph{CoRR}, abs/2306.12282, 2023.

\bibitem[Guo and Shmaya(2025)]{GS-23}
Yingni Guo and Eran Shmaya.
\newblock Robust monopoly regulation.
\newblock \emph{American Economic Review}, 115\penalty0 (2):\penalty0
  599–634, February 2025.

\bibitem[Hausman(1969)]{hau-69}
Warren~H. Hausman.
\newblock Sequential decision problems: A model to exploit existing
  forecasters.
\newblock \emph{Management Science}, 16\penalty0 (2):\penalty0 B93--B111, 1969.
\newblock ISSN 00251909, 15265501.

\bibitem[Heath and Jackson(1994{\natexlab{a}})]{HJ-94}
David~C Heath and Peter~L Jackson.
\newblock Modeling the evolution of demand forecasts ith application to safety
  stock analysis in production/distribution systems.
\newblock \emph{IIE Transactions}, 26\penalty0 (3):\penalty0 17--30,
  1994{\natexlab{a}}.

\bibitem[Heath and Jackson(1994{\natexlab{b}})]{heath1994modeling}
David~C Heath and Peter~L Jackson.
\newblock Modeling the evolution of demand forecasts ith application to safety
  stock analysis in production/distribution systems.
\newblock \emph{IIE transactions}, 26\penalty0 (3):\penalty0 17--30,
  1994{\natexlab{b}}.

\bibitem[Hu et~al.(2025)Hu, Chan, and Dong]{HCD-24}
Yue Hu, Carri~W. Chan, and Jing Dong.
\newblock Prediction-driven surge planning with application to emergency
  department nurse staffing.
\newblock \emph{Manag. Sci.}, 71\penalty0 (3):\penalty0 2079--2126, 2025.

\bibitem[Hurwicz and Shapiro(1978)]{HS-78}
Leonid Hurwicz and Leonard Shapiro.
\newblock Incentive structures maximizing residual gain under incomplete
  information.
\newblock \emph{The Bell Journal of Economics}, 9\penalty0 (1):\penalty0
  180--191, 1978.
\newblock ISSN 0361915X.

\bibitem[Jin and Ma(2022)]{JM-22}
Billy Jin and Will Ma.
\newblock Online bipartite matching with advice: Tight robustness-consistency
  tradeoffs for the two-stage model.
\newblock In Sanmi Koyejo, S.~Mohamed, A.~Agarwal, Danielle Belgrave, K.~Cho,
  and A.~Oh, editors, \emph{Advances in Neural Information Processing Systems
  35: Annual Conference on Neural Information Processing Systems 2022, NeurIPS
  2022, New Orleans, LA, USA, November 28 - December 9, 2022}, 2022.

\bibitem[Karlin et~al.(1994)Karlin, Manasse, McGeoch, and Owicki]{KMMO-94}
Anna~R. Karlin, Mark~S. Manasse, Lyle~A. McGeoch, and Susan~S. Owicki.
\newblock Competitive randomized algorithms for nonuniform problems.
\newblock \emph{Algorithmica}, 11\penalty0 (6):\penalty0 542--571, 1994.

\bibitem[Lei et~al.(2020)Lei, Jasin, Wang, Deng, and Putrevu]{LJWDP-20}
Yanzhe~Murray Lei, Stefanus Jasin, Jingyi Wang, Houtao Deng, and Jagannath
  Putrevu.
\newblock Dynamic workforce acquisition for crowdsourced last-mile delivery
  platforms.
\newblock \emph{Available at SSRN 3532844}, 2020.

\bibitem[Lobel et~al.(2024)Lobel, Martin, and Song]{LMS-24}
Ilan Lobel, S{\'{e}}bastien Martin, and Haotian Song.
\newblock Frontiers in operations: Employees vs. contractors: An operational
  perspective.
\newblock \emph{Manuf. Serv. Oper. Manag.}, 26\penalty0 (4):\penalty0
  1306--1322, 2024.

\bibitem[Lorca and Sun(2014)]{lorca2014adaptive}
Alvaro Lorca and Xu~Andy Sun.
\newblock Adaptive robust optimization with dynamic uncertainty sets for
  multi-period economic dispatch under significant wind.
\newblock \emph{IEEE Transactions on Power Systems}, 30\penalty0 (4):\penalty0
  1702--1713, 2014.

\bibitem[Lorca and Sun(2016)]{lorca2016multistage1}
Alvaro Lorca and Xu~Andy Sun.
\newblock Multistage robust unit commitment with dynamic uncertainty sets and
  energy storage.
\newblock \emph{IEEE Transactions on Power Systems}, 32\penalty0 (3):\penalty0
  1678--1688, 2016.

\bibitem[Lorca et~al.(2016)Lorca, Sun, Litvinov, and
  Zheng]{lorca2016multistage2}
Alvaro Lorca, X~Andy Sun, Eugene Litvinov, and Tongxin Zheng.
\newblock Multistage adaptive robust optimization for the unit commitment
  problem.
\newblock \emph{Operations Research}, 64\penalty0 (1):\penalty0 32--51, 2016.

\bibitem[Luy et~al.(2023)Luy, Hiermann, and Schiffer]{LHS-23}
Julius Luy, Gerhard Hiermann, and Maximilian Schiffer.
\newblock Strategic workforce planning in crowdsourced delivery with hybrid
  driver fleets.
\newblock \emph{CoRR}, abs/2311.17935, 2023.

\bibitem[Lykouris and Vassilvitskii(2021)]{LV-21}
Thodoris Lykouris and Sergei Vassilvitskii.
\newblock Competitive caching with machine learned advice.
\newblock \emph{J. {ACM}}, 68\penalty0 (4):\penalty0 24:1--24:25, 2021.

\bibitem[Mahdian et~al.(2012)Mahdian, Nazerzadeh, and Saberi]{MNS-12}
Mohammad Mahdian, Hamid Nazerzadeh, and Amin Saberi.
\newblock Online optimization with uncertain information.
\newblock \emph{{ACM} Trans. Algorithms}, 8\penalty0 (1):\penalty0 2:1--2:29,
  2012.

\bibitem[Manshadi et~al.(2023)Manshadi, Niazadeh, and Rodilitz]{MNR-21}
Vahideh~H. Manshadi, Rad Niazadeh, and Scott Rodilitz.
\newblock Fair dynamic rationing.
\newblock \emph{Manag. Sci.}, 69\penalty0 (11):\penalty0 6818--6836, 2023.

\bibitem[M{\"o}rters and Peres(2010)]{MP-10}
Peter M{\"o}rters and Yuval Peres.
\newblock \emph{Brownian motion}, volume~30.
\newblock Cambridge University Press, 2010.

\bibitem[Oh and {\"{O}}zer(2013)]{OO-13}
Sechan Oh and {\"{O}}zalp {\"{O}}zer.
\newblock Mechanism design for capacity planning under dynamic evolutions of
  asymmetric demand forecasts.
\newblock \emph{Manag. Sci.}, 59\penalty0 (4):\penalty0 987--1007, 2013.

\bibitem[Papadimitriou and Tsitsiklis(1987)]{PT-87}
Christos~H. Papadimitriou and John~N. Tsitsiklis.
\newblock The complexity of markov decision processes.
\newblock \emph{Math. Oper. Res.}, 12\penalty0 (3):\penalty0 441--450, 1987.

\bibitem[Papier(2016)]{papier2016supply}
Felix Papier.
\newblock Supply allocation under sequential advance demand information.
\newblock \emph{Operations Research}, 64\penalty0 (2):\penalty0 341--361, 2016.

\bibitem[Perakis and Roels(2008)]{PR-08}
Georgia Perakis and Guillaume Roels.
\newblock Regret in the newsvendor model with partial information.
\newblock \emph{Oper. Res.}, 56\penalty0 (1):\penalty0 188--203, 2008.

\bibitem[Purohit et~al.(2018)Purohit, Svitkina, and Kumar]{PSK-18}
Manish Purohit, Zoya Svitkina, and Ravi Kumar.
\newblock Improving online algorithms via {ML} predictions.
\newblock In Samy Bengio, Hanna~M. Wallach, Hugo Larochelle, Kristen Grauman,
  Nicol{\`{o}} Cesa{-}Bianchi, and Roman Garnett, editors, \emph{Advances in
  Neural Information Processing Systems 31: Annual Conference on Neural
  Information Processing Systems 2018, NeurIPS 2018, December 3-8, 2018,
  Montr{\'{e}}al, Canada}, pages 9684--9693, 2018.

\bibitem[Shafer and Vovk(2008)]{SV-08}
Glenn Shafer and Vladimir Vovk.
\newblock A tutorial on conformal prediction.
\newblock \emph{J. Mach. Learn. Res.}, 9:\penalty0 371--421, 2008.

\bibitem[Shalev-Shwartz and Ben-David(2014)]{SB-14}
Shai Shalev-Shwartz and Shai Ben-David.
\newblock \emph{Understanding machine learning: From theory to algorithms}.
\newblock Cambridge university press, 2014.

\bibitem[Smithson(2003)]{smi-03}
Michael Smithson.
\newblock \emph{Confidence intervals}.
\newblock Sage, 2003.

\bibitem[Song and Zipkin(2012)]{SZ-12}
Jing-Sheng Song and Paul~H. Zipkin.
\newblock Newsvendor problems with sequentially revealed demand information.
\newblock \emph{Naval Research Logistics (NRL)}, 59\penalty0 (8):\penalty0
  601--612, 2012.

\bibitem[Sun et~al.(2024)Sun, Huang, Christianson, Hajiesmaili, Wierman, and
  Boutaba]{sun2024online}
Bo~Sun, Jerry Huang, Nicolas Christianson, Mohammad Hajiesmaili, Adam Wierman,
  and Raouf Boutaba.
\newblock Online algorithms with uncertainty-quantified predictions.
\newblock In \emph{International Conference on Machine Learning}, pages
  47056--47077. PMLR, 2024.

\bibitem[Topan et~al.(2018)Topan, Tan, van Houtum, and Dekker]{topan2018using}
Engin Topan, Tarkan Tan, Geert-Jan van Houtum, and Rommert Dekker.
\newblock Using imperfect advance demand information in lost-sales inventory
  systems with the option of returning inventory.
\newblock \emph{IISE Transactions}, 50\penalty0 (3):\penalty0 246--264, 2018.

\bibitem[Wang et~al.(2024)Wang, Liu, and Zhang]{WLZ-20}
Shixin Wang, Shaoxuan Liu, and Jiawei Zhang.
\newblock Minimax regret robust screening with moment information.
\newblock \emph{Manuf. Serv. Oper. Manag.}, 26\penalty0 (3):\penalty0
  992--1012, 2024.

\bibitem[Wang et~al.(2012)Wang, Atasu, and Kurtulus]{WAK-12}
Tong Wang, Atalay Atasu, and M{\"{u}}min Kurtulus.
\newblock A multiordering newsvendor model with dynamic forecast evolution.
\newblock \emph{Manuf. Serv. Oper. Manag.}, 14\penalty0 (3):\penalty0 472--484,
  2012.

\bibitem[Whitin(1955)]{Whi-55}
Thomson~M Whitin.
\newblock Inventory control and price theory.
\newblock \emph{Management Science}, 2\penalty0 (1):\penalty0 61--68, 1955.

\bibitem[Xin and Goldberg(2022)]{XG-22}
Linwei Xin and David~Alan Goldberg.
\newblock Distributionally robust inventory control when demand is a
  martingale.
\newblock \emph{Math. Oper. Res.}, 47\penalty0 (3):\penalty0 2387--2414, 2022.

\end{thebibliography}
